\documentclass[11pt]{article}
\usepackage{jheppub}
\makeatletter\gdef\@fpheader{}\makeatother
\usepackage{amsmath,amssymb}
\usepackage{mathtools}
\usepackage{bm}
\usepackage{mathrsfs}
\usepackage{tikz}
\usetikzlibrary{arrows.meta,positioning}
\usepackage{xcolor}
\usepackage{amsthm}
\usepackage{longtable,array}
\newtheorem{proposition}{Proposition}
\newtheorem{conjecture}{Conjecture}
\hypersetup{pdftitle={Root-of-unity limit of the A(n-1) lens elliptic gamma function solution of the star-star relation},pdfauthor={Masahito Yamazaki}}

\newcommand{\ii}{\mathsf{i}}
\newcommand{\p}{{\mathsf p}}
\newcommand{\q}{{\mathsf q}}
\newcommand{\EXP}{\mathrm{e}}
\newcommand{\ovth}{\overline{\vartheta}}
\newcommand{\olW}{\overline{W}}
\newcommand{\bW}{\mathbb{W}}
\newcommand{\obW}{\overline{\mathbb{W}}}
\newcommand{\bS}{\mathbb{S}}
\newcommand{\ZZ}{\mathbb{Z}}
\newcommand{\RR}{\mathbb{R}}
\newcommand{\CC}{\mathbb{C}}
\newcommand{\vx}{\bm{x}}
\newcommand{\vm}{\bm{m}}
\newcommand{\vxi}{\bm{\xi}}
\newcommand{\va}{\bm{a}}
\newcommand{\re}{\operatorname{Re}}
\newcommand{\im}{\operatorname{Im}}

\title{Root-of-unity limit of the $A_{n-1}$ lens elliptic gamma function solution of the star-star relation}

\author[a,b,c,d]{Masahito Yamazaki}

\affiliation[a]{Department of Physics, University of Tokyo, Tokyo 113-0033, Japan}
\affiliation[b]{Kavli Institute for Physics and Mathematics of the Universe (Kavli IPMU), UTIAS, University of Tokyo, Chiba 277-8583, Japan}
\affiliation[c]{Trans-Scale Quantum Science Institute, University of Tokyo, Tokyo 113-0033, Japan}
\affiliation[d]{Center for Interdisciplinary Theoretical and Mathematical Sciences (iTHEMS), RIKEN, Saitama 351-0198, Japan}

\emailAdd{masahito.yamazaki@ipmu.jp}

\preprint{RIKEN-iTHEMS-Report-26}

\abstract{
We study the root-of-unity limit of the multi-component spin solution of the star-star relation whose Boltzmann weights are products of
lens elliptic gamma functions, and whose spins take values in $\RR^{n-1}\times(\ZZ_r)^{n-1}$. When one of the elliptic nomes tends to a primitive
$2rN$-th root of unity, the quasi-classical expansion gives at leading order the classical star-star relation of the Bazhanov--Sergeev
multi-spin model; at the next order, under the standard assumptions on the saddle points, a $(\ZZ_{rN})^{n-1}$-state solution of the
star-star relation, which depends on $r$ and $N$ only through $rN$.
Real $\ZZ_n$-symmetric configurations of the continuous spins give elliptic $\mathfrak{sl}(n)$ Kashiwara--Miwa type models labeled by $n$-torsion points. Extended to general primitive torsion
points, these include the broken $\ZZ_{rN}$ model of Quano and Fujii, for which the star-star relation follows from their Yang--Baxter equation
when $\gcd(n,rN)=1$. In the trigonometric limit the saddle point
equations single out the spectral curve of the $\mathfrak{sl}(n)$ chiral Potts model of Bazhanov, Kashaev, Mangazeev and Stroganov, and the
Boltzmann weights become those of that model.
}

\begin{document}

\maketitle

\section{Introduction}

The star-triangle relation \cite{Baxter:1982zz,Au-Yang:1989mk} and its four-edge counterpart, the star-star relation \cite{Baxter:1986phd,Bazhanov:1992jqa,Baxter:1997tn}, are local conditions
ensuring the integrability of two-dimensional lattice models with spins on the vertices of the lattice: they imply the Yang--Baxter equation \cite{Yang:1967bm,Baxter:1972hz} and hence
the commutativity of transfer matrices. A large class of solutions with continuous spins has been found since \cite{Bazhanov:2010kz,Spiridonov:2010em}, with
Boltzmann weights built from hyperbolic or elliptic gamma functions \cite{Bazhanov:2010kz,Bazhanov:2011mz}; their star-triangle and star-star
relations are identities for elliptic hypergeometric integrals \cite{Spiridonov:2001eb,Rains:2005lkc}. Many of them have a natural origin in the
Gauge/YBE correspondence \cite{Yamazaki:2012cp,Terashima:2012cx,Yamazaki:2013nra} (see \cite{Yamazaki:2018xbx,Gahramanov:2017ysd,Gahramanov:2022qge} for reviews), in which the star-triangle and star-star relations are identities
between partition functions of supersymmetric gauge theories related by Seiberg duality \cite{Seiberg:1994pq}. In particular, the lens index \cite{Benini:2011nc,Razamat:2013opa} of four-dimensional
$\mathcal N=1$ theories on $S^1\times (S^3/\ZZ_r)$ gives solutions whose Boltzmann weights are products of lens elliptic gamma functions and whose spins
have continuous as well as $\ZZ_r$-valued components \cite{Yamazaki:2013nra,Kels:2015bda,Gahramanov:2016ilb}. For gauge group $\mathrm{SU}(n)$ this leads to a multi-component spin
model with spins in $\RR^{n-1}\times(\ZZ_r)^{n-1}$ \cite{Yamazaki:2013nra}. Its star-star relation was proven in \cite{Kels:2017toi} as a consequence of
an $A_{n-1}$ elliptic hypergeometric sum/integral transformation, a lens generalization of the transformations of Rains \cite{Rains:2005lkc}, and for $r=1$ the model reduces to the multi-spin model of Bazhanov and Sergeev
\cite{Bazhanov:2011mz}.

Models with continuous spins contain models with discrete spins as limiting cases. Bazhanov and Sergeev showed that when an elliptic nome tends
to a root of unity, the quasi-classical expansion of the star-triangle relation of the elliptic gamma function solution gives the classical
star-triangle relation of the $Q4$ equation \cite{Adler:2002abs} at leading order, and a $\ZZ_N$ star-triangle relation at the next order
\cite{Bazhanov:2010kz}. The latter contains the Kashiwara--Miwa \cite{Kashiwara:1986tu}, Fateev--Zamolodchikov \cite{Fateev:1982wi} and chiral
Potts \cite{Au-Yang:1987syu,Baxter:1987eq} models as special cases. In \cite{Kels:2017vbc} the analysis was extended to the lens elliptic gamma function
solution (the case $n=2$ of the model above); there the $\ZZ_r$ and $\ZZ_N$ discrete variables combine into a single $\ZZ_{rN}$ spin.
Earlier connections between quantum dilogarithms at roots of unity and discrete integrable models were established in
\cite{Faddeev:1993rs,Bazhanov:1995qd}. Root-of-unity degenerations in which discrete variables emerge from continuous ones also appear for the non-compact quantum dilogarithm and
its identities \cite{Ip:2014ifa}.

For the multi-component models ($n\ge3$) the corresponding limit was anticipated by Bazhanov and Sergeev \cite[Sec.~5]{Bazhanov:2011mz}. They
argued that for $r=1$ it leads to a ``hybrid'' model, which couples a classical integrable system to discrete spins in $(\ZZ_N)^{n-1}$, and which
for a constant classical background should reduce to an $\mathfrak{sl}(n)$ generalization of the Kashiwara--Miwa model, and in the trigonometric
limit to the $\mathfrak{sl}(n)$ chiral Potts model. Discrete-spin $\mathfrak{sl}(n)$ models of both kinds are in fact known. In the trigonometric case,
the $(\ZZ_N)^{n-1}$ chiral Potts model of Bazhanov, Kashaev, Mangazeev and Stroganov (BKMS) \cite{Bazhanov:1990qk,Date:1990bs,Kashaev:1990wy}
was obtained from intertwiners of cyclic representations of $U_q(\widehat{\mathfrak{sl}}_n)$ at roots of unity, and is closely related to the
three-dimensional model of Bazhanov and Baxter \cite{Bazhanov:1992jqa,Bazhanov:1992jq}. In the elliptic case, $\mathfrak{sl}(n)$
generalizations of the Kashiwara--Miwa model, the $A^{(1)}_{n-1}$ broken $\ZZ_N$ models, were constructed in
\cite{Sergeev:1992ap,Quano:1992wc} from cyclic representations of the $L$-operators for Belavin's $R$-matrix (see also \cite{Hasegawa:1993cs}). More recent work on multi-component solutions of the star-star relation
\cite{Bazhanov:2011mz,Bazhanov:2013bh,Kels:2017toi,Gahramanov:2017idz,Catak:2021coz,Kels:2018qzx,Kels:2020zjn,Kels:2024afm} concerns models with continuous spins
and their quasi-classical limits. The root-of-unity limit leading to discrete spins, however, has been carried out only for $n=2$
\cite{Bazhanov:2010kz,Kels:2017vbc}, and not even for $r=1$ when $n\ge3$.

In this paper we consider the root-of-unity limit of the lens elliptic gamma function solution with spins in $\RR^{n-1}\times(\ZZ_r)^{n-1}$, for
general $n$ and $r$. As in \cite{Kels:2017vbc}, one of the elliptic nomes is taken to a root of unity, and the star-star relation is evaluated
by a saddle point method. The asymptotics of the lens elliptic gamma function in this limit were obtained in \cite{Kels:2017vbc} by rearranging
infinite products. Their $O(\hbar^0)$ term, through which the discrete variables enter, is a product of $rN$ fractional powers of theta functions
and of further theta functions whose number and sign depend on the discrete variable, and in this form its structure, in particular its
quasi-periodicity in the discrete variable, is not manifest. We note that only the ratio
of two such terms at the same value of the leading-order variable is needed, and that this ratio is determined by the leading asymptotics and an
exact first-order difference equation alone (Section~\ref{sec:asymptotics}). The asymptotics are then expressed through a ``cyclic elliptic gamma
function'' \eqref{omegadef}, an elliptic analogue of the Pochhammer symbol whose step is an $rN$-torsion point. This representation applies
directly to general $n$, and it also makes the modular transformation straightforward. The $O(\hbar^0)$ formula of \cite{Kels:2017vbc} is needed only for the
normalization of the star-star relation.

At leading order we obtain a classical star-star relation \eqref{cssr}, which coincides after a rescaling with that of
\cite{Bazhanov:2011mz} and whose saddle point equations \eqref{saddlemod} are multi-component analogues of the three-leg form of $Q4$. At the
next order we obtain a $(\ZZ_{rN})^{n-1}$-state solution \eqref{dssr} of the star-star relation on a classical background, the hybrid model of
\cite{Bazhanov:2011mz}, which, as for $n=2$ \cite{Kels:2017vbc}, depends on $r$ and $N$ only through $M=rN$ (Section~\ref{sec:special}).
Real $\ZZ_n$-symmetric configurations of the continuous spins give elliptic $\mathfrak{sl}(n)$ Kashiwara--Miwa type models
\eqref{KMnW}--\eqref{KMnS} labeled by an $n$-torsion point $\omega_0$. Extended to all primitive torsion points, they include at
$\omega_0=\pi/n$ the $A^{(1)}_{n-1}$ broken $\ZZ_M$ model of Quano and Fujii \cite{Quano:1992wc}, for which the star-star relation
\eqref{KMnssr} follows from their Yang--Baxter equation when $\gcd(n,M)=1$; for $\gcd(n,M)>1$ our weights extend their model to the state space
$Q/MQ\cong(\ZZ_M)^{n-1}$, where $Q$ is the root lattice of $\mathfrak{sl}(n)$. In the trigonometric limit the saddle point equations become the condition that the cross-ratio of the four rapidities be the same in all
$n$ components. This singles out the spectral curve of the $(\ZZ_M)^{n-1}$ chiral Potts model, and the Boltzmann weights become those of
\cite{Bazhanov:1990qk} (Section~\ref{sec:trigsec}), generalizing the appearance of the chiral Potts curve for $n=2$ \cite{Bazhanov:2010kz,Kels:2017vbc}.

The star-star relation \eqref{KMnssr} with the Quano--Fujii weights on the edges appears to be a new identity for that model, and the models at the
torsion points which are not modular images of the Quano--Fujii point, which exist when $\gcd(n,M)>1$, have not appeared before as far as we
are aware. The difference equations, the cyclic elliptic gamma function \eqref{omegadef} and its modular transformation are exact, while
\eqref{dssr} rests, as in \cite{Bazhanov:2010kz,Kels:2017vbc}, on assumptions on the saddle points (Section~\ref{sec:discrete}). It is supported
by numerical checks (Appendix~\ref{app:numerics}), and in a number of cases it reduces to relations proven in the literature
(Table~\ref{tab:status}). The relations between the models are summarized in Fig.~\ref{fig:models}.

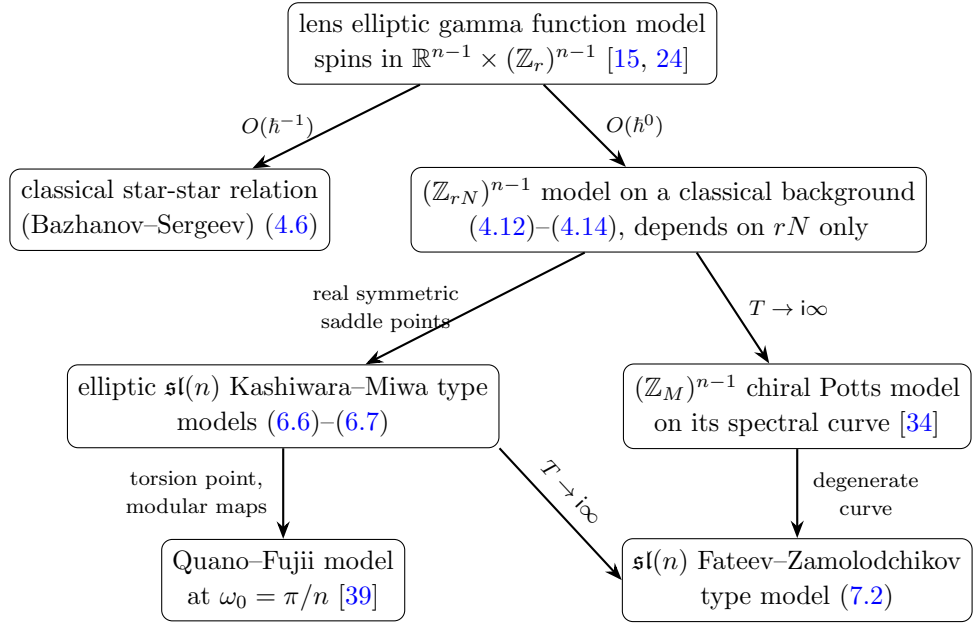
\begin{figure}[t]\centering
  \begin{tikzpicture}[box/.style={draw,rounded corners,align=center,font=\small,inner sep=4pt},arr/.style={-{Stealth},thick},lab/.style={font=\scriptsize,align=center}]
    \node[box] (lens) at (0,0) {lens elliptic gamma function model\\ spins in $\RR^{n-1}\times(\ZZ_r)^{n-1}$ \cite{Yamazaki:2013nra,Kels:2017toi}};
    \node[box] (cl) at (-4.4,-2.2) {classical star-star relation\\ (Bazhanov--Sergeev) \eqref{cssr}};
    \node[box] (dis) at (2.2,-2.2) {$(\ZZ_{rN})^{n-1}$ model on a classical background\\ \eqref{dW}--\eqref{dssr}, depends on $rN$ only};
    \node[box] (km) at (-2.9,-4.8) {elliptic $\mathfrak{sl}(n)$ Kashiwara--Miwa type\\ models \eqref{KMnW}--\eqref{KMnS}};
    \node[box] (qf) at (-2.9,-7.1) {Quano--Fujii model\\ at $\omega_0=\pi/n$ \cite{Quano:1992wc}};
    \node[box] (cp) at (3.9,-4.8) {$(\ZZ_M)^{n-1}$ chiral Potts model\\ on its spectral curve \cite{Bazhanov:1990qk}};
    \node[box] (fz) at (3.9,-7.1) {$\mathfrak{sl}(n)$ Fateev--Zamolodchikov\\ type model \eqref{trigW}};
    \draw[arr] (lens) -- node[lab,left=4pt] {$O(\hbar^{-1})$} (cl);
    \draw[arr] (lens) -- node[lab,right=4pt] {$O(\hbar^{0})$} (dis);
    \draw[arr] (dis) -- node[lab,left=2pt] {real symmetric\\ saddle points} (km);
    \draw[arr] (dis) -- node[lab,right=2pt] {$T\to\ii\infty$} (cp);
    \draw[arr] (km) -- node[lab,left=2pt] {torsion point,\\ modular maps} (qf);
    \draw[arr] (km.south east) -- node[lab,above=2pt,sloped,pos=0.45] {$T\to\ii\infty$} (fz.west);
    \draw[arr] (cp) -- node[lab,right=2pt] {degenerate\\ curve} (fz);
  \end{tikzpicture}
  \caption{The models obtained in this paper and their relations. Boxes in the lower part are discrete-spin models satisfying the star-star relation
    \eqref{dssr}; the cases in which the relation is proven are listed in Table~\ref{tab:status}.}\label{fig:models}
\end{figure}

The paper is organized as follows. In Section~\ref{sec:model} we define the model and derive the exact difference equations.
Section~\ref{sec:asymptotics} contains the root-of-unity asymptotics and the modular transformation. In Section~\ref{sec:expansion} we carry out
the quasi-classical expansion of the star-star relation, and in Section~\ref{sec:special} we discuss general properties of the resulting discrete
spin models. Section~\ref{sec:KMn} is devoted to the elliptic $\mathfrak{sl}(n)$ Kashiwara--Miwa type models and their relation to the Quano--Fujii model, and
Section~\ref{sec:trigsec} to the trigonometric limit and the $(\ZZ_M)^{n-1}$ chiral Potts model. We conclude in Section~\ref{sec:discussion}.
Appendix~\ref{app:theta} collects formulas for theta functions, and Appendix~\ref{app:dict}
gives the dictionary between the conventions of \cite{Kels:2017toi} and \cite{Kels:2017vbc}. Appendix~\ref{app:gcd} contains the details for $\gcd(n,M)>1$, Appendix~\ref{app:torsion} those of the
discussion of the general torsion points, Appendix~\ref{app:modular} the details of a computation there, and Appendix~\ref{app:numerics}
summarizes our numerical checks.

\section{The model}\label{sec:model}

\subsection{Lens elliptic gamma function}

In this section we define the model and collect the properties of the lens elliptic gamma function that will be needed below. We use the
conventions of \cite{Kels:2017vbc}. Let $r=1,2,\ldots$ be a positive integer, and let
\begin{equation}
  \p:=\EXP^{\pi\ii\tau_1},\qquad \q:=\EXP^{\pi\ii\tau_2},\qquad \im\tau_1,\im\tau_2>0\;.
\end{equation}
The lens elliptic gamma function is then defined by
\begin{equation}\label{Phidef}
  \Phi(z,m):=\EXP^{E(z,m)}\,
  \Phi_1\!\left(z+\left(\tfrac r2-m\right)\pi\tau_1;\p\q,\p^r\right)
  \Phi_1\!\left(z-\left(\tfrac r2-m\right)\pi\tau_2;\p\q,\q^r\right)\;,
\end{equation}
where
\begin{equation}\label{Edef}
  \begin{gathered}
    E(z,m):=\frac{\ii\, m(r-m)}{6r}\Big(6z+\pi(r-2m)(\tau_1-\tau_2-1)\Big)\;,\\
    \Phi_1(z;\p,\q):=\prod_{j,k=0}^\infty\frac{1-\EXP^{2\ii z}\p^{2j+1}\q^{2k+1}}{1-\EXP^{-2\ii z}\p^{2j+1}\q^{2k+1}}\;.
  \end{gathered}
\end{equation}
In \cite{Kels:2017vbc} the definition \eqref{Phidef} is used for $m\in\{0,\ldots,r-1\}$. Here it will be convenient to use the same formula for
all $m\in\ZZ$. The function defined in this way is $r$-periodic in $m$, and satisfies
\begin{equation}\label{Phiprop}
  \Phi(z,m+r)=\Phi(z,m),\qquad \Phi(z,m)\Phi(-z,-m)=1,\qquad
  \Phi(z+\pi,m)=\EXP^{\frac{\pi\ii m(r-m)}{r}}\,\Phi(z,m)\;.
\end{equation}
Throughout, the theta functions $\ovth_j(x\,|\,\tau)$ are normalized as in \cite[App.~A]{Kels:2017vbc}. Their definitions and the formulas
used in this paper are collected in Appendix~\ref{app:theta}.

\paragraph{Difference equations.}
The lens elliptic gamma function satisfies two first-order difference equations, one in each of the two periods, which will play a central
role in Section~\ref{sec:asymptotics}. Consider the simultaneous shift $z\to z+\pi\tau_2$, $m\to m-1$. Under this shift the argument of
the second factor $\Phi_1$ in \eqref{Phidef} is invariant, while that of the first factor is shifted by $\pi(\tau_1+\tau_2)$. With
$w:=z+(r/2-m)\pi\tau_1$ we have
$\Phi_1(w+\pi(\tau_1+\tau_2);\p\q,\p^r)/\Phi_1(w;\p\q,\p^r)=1/\big[(\EXP^{2\ii w}\p^{r+1}\q;\p^{2r})_\infty(\EXP^{-2\ii w}\p^{r-1}\q^{-1};\p^{2r})_\infty\big]$, and therefore
\begin{equation}\label{shift}
  \frac{\Phi(z+\pi\tau_2,m-1)}{\Phi(z,m)}
  =\frac{\EXP^{E(z+\pi\tau_2,m-1)-E(z,m)}}{\ovth_4\!\left(z+\frac{\pi}{2}(r-2m+1)\tau_1+\frac{\pi}{2}\tau_2\,\middle|\,r\tau_1\right)}\;.
\end{equation}
In the same way, exchanging the roles of the two factors (and of $\p^r$ and $\q^r$), we obtain the difference equation in the other period,
\begin{equation}\label{shift1}
  \frac{\Phi(z+\pi\tau_1,m+1)}{\Phi(z,m)}
  =\frac{\EXP^{E(z+\pi\tau_1,m+1)-E(z,m)}}{\ovth_4\!\left(z-\frac{\pi}{2}(r-2m-1)\tau_2+\frac{\pi}{2}\tau_1\,\middle|\,r\tau_2\right)}\;.
\end{equation}

\paragraph{Vertex factor.}
As a first application of the difference equations, let us compute the combination of lens elliptic gamma functions which appears in the
vertex weight below. Let $\eta:=-(\pi\ii/2)(\tau_1+\tau_2)$ be the crossing parameter. Using the inversion relation in \eqref{Phiprop}, we
can write
\begin{equation}
  \Phi(-\ii\eta+x,m)\,\Phi(-\ii\eta-x,-m)=\frac{\Phi(z,m)}{\Phi(z+\pi\tau_1+\pi\tau_2,m)},\qquad z:=x-\tfrac{\pi}{2}(\tau_1+\tau_2)\;.
\end{equation}
We then write the inverse of the ratio on the right-hand side as $[\Phi(z+\pi\tau_1+\pi\tau_2,m)/\Phi(z+\pi\tau_1,m+1)]\cdot[\Phi(z+\pi\tau_1,m+1)/\Phi(z,m)]$.
The two difference equations \eqref{shift} and \eqref{shift1} turn this into a product of two theta functions, and after a half-period shift
it is found that, for all $m\in\ZZ$,
\begin{equation}\label{pairexact}
  \Phi(-\ii\eta+x,m)\,\Phi(-\ii\eta-x,-m)=\EXP^{\frac{\eta r}{2}-\frac{2\eta m^2}{r}}\,
  \ovth_1(x-\pi m\tau_1\,|\,r\tau_1)\,\ovth_1(x+\pi m\tau_2\,|\,r\tau_2)\;.
\end{equation}
For even $m$ this reproduces the identity between the two expressions for the vertex weight given in \cite[(2.14)]{Kels:2017vbc}.
The relation \eqref{pairexact} is the lens analogue of $1/(\Gamma(z)\Gamma(z^{-1}))=\theta(z;p)\theta(z^{-1};q)$ for the ordinary elliptic gamma function.

\subsection{Boltzmann weights and the star-star relation}

We now define the integrable lattice model. Fix an integer $n\ge2$. A spin is a pair $\sigma=(\vx,\vm)$ with
\begin{equation}
  \vx=(x_1,\ldots,x_n)\in\RR^n,\quad \sum_{i=1}^n x_i=0,\qquad
  \vm=(m_1,\ldots,m_n)\in(\ZZ_r)^n,\quad \sum_{i=1}^n m_i\equiv0\ (\mathrm{mod}\ r)\;.
\end{equation}
The Boltzmann weights are \cite{Yamazaki:2013nra,Kels:2017toi}
\begin{align}\label{BW}
  W_\alpha(\sigma_a,\sigma_b) & :=\prod_{i,j=1}^n\Phi\big(x_{a,i}-x_{b,j}+\ii\alpha,\,m_{a,i}-m_{b,j}\big),\qquad
  \olW_\alpha(\sigma_a,\sigma_b):=W_{\eta-\alpha}(\sigma_a,\sigma_b)\;,                                           \\
  S(\sigma)                   & :=\prod_{i\neq j}\Phi\big(-\ii\eta+x_i-x_j,\,m_i-m_j\big)
  =\prod_{i<j}\Phi(-\ii\eta+x_{ij},m_{ij})\Phi(-\ii\eta-x_{ij},-m_{ij})\;,
  \label{Sdef}
\end{align}
with $x_{ij}:=x_i-x_j$, $m_{ij}:=m_i-m_j$, and the measure
\begin{equation}
  \int\!d\sigma:=\sum_{m_1,\ldots,m_{n-1}=0}^{r-1}\int_{[0,\pi]^{n-1}}dx_1\cdots dx_{n-1},\qquad
  x_n:=-\sum_{i<n}x_i,\quad m_n:=-\sum_{i<n}m_i\;.
\end{equation}
The integrand of the star-star relation below is invariant under $x_i\to x_i+\pi k_i$ with $\sum k_i=0$.
We call this model the $A_{n-1}$ model, after the sum/integral transformation of \cite{Kels:2017toi} for the $A_{n-1}$ root system which
implies its star-star relation. Accordingly, the spins have $n$ components, and the discrete models obtained below are associated with
$\mathfrak{sl}(n)$.
For $n=2$, where $\sigma=((x,-x),(m,-m))$, the weights \eqref{BW} reduce to those of \cite{Kels:2015bda,Kels:2017vbc}.

In \cite[Sec.~6]{Kels:2017toi} the weights are written in terms of $\Gamma(z,m;\sigma,\tau)$, and the weight $\overline{W}_\alpha$ is defined as
$\overline{W}_\alpha=\sqrt{S(\sigma_a)S(\sigma_b)}\,W_{\eta-\alpha}$. In the star-star relation the factors $\sqrt{S}$ of the four external spins
cancel between the two sides, so that one may equivalently put $S(\sigma_H)$ on the internal vertex, as we do here.
The dictionary between the conventions of \cite{Kels:2017toi} and those used in this paper is given in Appendix~\ref{app:dict}.

We next define the weights of the two four-edge stars shown in Fig.~\ref{ssfig},
\begin{align}
  W^{(1)}_{\bm u\bm v}\begin{pmatrix}\sigma_I&\sigma_J\\\sigma_K&\sigma_L\end{pmatrix} & :=
  \int d\sigma_H\, S(\sigma_H)\,\olW_{u-v}(\sigma_K,\sigma_H)\olW_{u'-v'}(\sigma_J,\sigma_H)\nonumber                                                                     \\
                                                                                       & \qquad\times W_{u'-v}(\sigma_H,\sigma_I)W_{u-v'}(\sigma_H,\sigma_L)\;,\label{W1} \\
  W^{(2)}_{\bm u\bm v}\begin{pmatrix}\sigma_I&\sigma_J\\\sigma_K&\sigma_L\end{pmatrix} & :=
  \int d\sigma_H\, S(\sigma_H)\,\olW_{u-v}(\sigma_H,\sigma_J)\olW_{u'-v'}(\sigma_H,\sigma_K)\nonumber                                                                     \\
                                                                                       & \qquad\times W_{u'-v}(\sigma_L,\sigma_H)W_{u-v'}(\sigma_I,\sigma_H)\;.\label{W2}
\end{align}
The star-star relation \cite{Yamazaki:2013nra,Kels:2017toi} then reads
\begin{multline}\label{ssr}
  W_{v'-v}(\sigma_L,\sigma_K)\,W_{u'-u}(\sigma_L,\sigma_J)\,W^{(1)}_{\bm u\bm v}\begin{pmatrix}\sigma_I&\sigma_J\\\sigma_K&\sigma_L\end{pmatrix}\\
  =W_{v'-v}(\sigma_J,\sigma_I)\,W_{u'-u}(\sigma_K,\sigma_I)\,W^{(2)}_{\bm u\bm v}\begin{pmatrix}\sigma_I&\sigma_J\\\sigma_K&\sigma_L\end{pmatrix}\;.
\end{multline}
Let us specify the region of validity. For real external spins and spectral parameters satisfying
\begin{equation}\label{ssrregion}
  0<\re(u-v),\ \re(u'-v'),\ \re(u'-v),\ \re(u-v')<\re\eta\;,
\end{equation}
the poles of the integrands of \eqref{W1} and \eqref{W2} do not meet the real integration domain, and \eqref{ssr} holds with the integration
over $[0,\pi]^{n-1}$ as written; this is the condition $\im(s_i),\im(t_i)>0$ for the real contour in the $A_{n-1}$ sum/integral transformation
of \cite[Sec.~3]{Kels:2017toi}, translated by the dictionary of Appendix~\ref{app:dict}. Outside this region, \eqref{ssr} is understood by
analytic continuation in the parameters, which is equivalent to deforming the contour so that it continues to separate the two families of
poles \cite{Kels:2017toi}. The conditions used in the saddle point analysis below are
stated separately in Sections~\ref{sec:asymptotics} and~\ref{sec:expansion}; for real
continuous spins, \eqref{ssrregion} corresponds there to $0<\theta<\pi/2$ for the four internal edges, cf.\ \eqref{thetadef}.

\begin{figure}[htb]\centering
\begin{tikzpicture}[scale=1.2]

\draw[-,thick] (-1,-1)--(1,1);
\draw[-,thick] (1,-1)--(-1,1);
\draw[-,thick] (-1,-1)--(1,-1);
\draw[-,thick] (1,-1)--(1,1);

\fill[black!] (-1,-1) circle (2.2pt)
node[below=1.5pt]{\color{black}\small $\sigma_K$};
\fill[black!] (1,-1) circle (2.2pt)
node[below=1.5pt]{\color{black}\small $\sigma_L$};
\fill[black!] (-1,1) circle (2.2pt)
node[above=1.5pt]{\color{black}\small $\sigma_I$};
\fill[black!] (1,1) circle (2.2pt)
node[above=1.5pt]{\color{black}\small $\sigma_J$};

\filldraw[fill=black!,draw=black!] (0,0) circle (2.2pt)
node[below=2.5pt]{\color{black}\small $\sigma_H$};

\draw[->,thick,dashed] (-1.7,0.5) -- (1.7,0.5);
\draw[white!] (-1.7,0.5) circle (0.01pt)
node[left=1.5pt]{\color{black}\small $u'$};
\draw[->] (-1.7,-0.5) -- (1.7,-0.5);
\draw[white!] (-1.8,-0.5) circle (0.01pt)
node[left=1.5pt]{\color{black}\small $u$};
\draw[->,thick,dashed] (-0.5,-1.7) .. controls (-0.5,-1.6) and (-0.4,-1.2) .. (0,-1) .. controls (0.2,-0.9) and (0.5,-0.6) .. (0.5,-0.5) .. controls (0.7,-0.4) and (0.7,0.4) .. (0.5,0.5) .. controls (0.5,0.6) and (0.2,0.9) .. (0,1) .. controls (-0.4,1.2) and (-0.5,1.6) .. (-0.5,1.7);
\draw[white!] (-0.5,-1.7) circle (0.01pt)
node[below=1.5pt]{\color{black}\small $v'$};
\draw[->] (0.5,-1.7) .. controls (0.5,-1.6) and (0.4,-1.2) .. (0,-1) .. controls (-0.2,-0.9) and (-0.5,-0.6) .. (-0.5,-0.5) .. controls (-0.7,-0.4) and (-0.7,0.4) .. (-0.5,0.5) .. controls (-0.5,0.6) and (-0.2,0.9) .. (0,1) .. controls (0.4,1.2) and (0.5,1.6) .. (0.5,1.7);
\draw[white!] (0.5,-1.77) circle (0.01pt)
node[below=1.5pt]{\color{black}\small $v$};

\draw[white!] (2.2,0) circle (0.01pt)
node[right=0.1pt]{\color{black}=};

\begin{scope}[xshift=150pt]

\draw[-,thick] (-1,-1)--(1,1);
\draw[-,thick] (1,-1)--(-1,1);
\draw[-,thick] (-1,1)--(1,1);
\draw[-,thick] (-1,-1)--(-1,1);

\fill[black!] (-1,-1) circle (2.2pt)
node[below=1.5pt]{\color{black}\small $\sigma_K$};
\fill[black!] (1,-1) circle (2.2pt)
node[below=1.5pt]{\color{black}\small $\sigma_L$};
\fill[black!] (-1,1) circle (2.2pt)
node[above=1.5pt]{\color{black}\small $\sigma_I$};
\fill[black!] (1,1) circle (2.2pt)
node[above=1.5pt]{\color{black}\small $\sigma_J$};

\filldraw[fill=black!,draw=black!] (0,0) circle (2.2pt)
node[below=2.5pt]{\color{black}\small $\sigma_H$};

\draw[->,thick,dashed] (-1.7,0.5) .. controls (-1.6,0.5) and (-1.2,0.4) .. (-1,0) .. controls (-0.9,-0.2) and (-0.6,-0.5) .. (-0.5,-0.5) .. controls (-0.4,-0.7) and (0.4,-0.7) .. (0.5,-0.5) .. controls (0.6,-0.5) and (0.9,-0.2) .. (1,0) .. controls (1.2,0.4) and (1.6,0.5) .. (1.7,0.5);
\draw[white!] (-1.7,0.5) circle (0.01pt)
node[left=1.5pt]{\color{black}\small $u'$};
\draw[->] (-1.7,-0.5) .. controls (-1.6,-0.5) and (-1.2,-0.4) .. (-1,0) .. controls (-0.9,0.2) and (-0.6,0.5) .. (-0.5,0.5) .. controls (-0.4,0.7) and (0.4,0.7) .. (0.5,0.5) .. controls (0.6,0.5) and (0.9,0.2) .. (1,0) .. controls (1.2,-0.4) and (1.6,-0.5) .. (1.7,-0.5);
\draw[white!] (-1.8,-0.5) circle (0.01pt)
node[left=1.5pt]{\color{black}\small $u$};
\draw[->,thick,dashed] (-0.5,-1.7) -- (-0.5,1.7);
\draw[white!] (-0.5,-1.7) circle (0.01pt)
node[below=1.5pt]{\color{black}\small $v'$};
\draw[->] (0.5,-1.7) -- (0.5,1.7);
\draw[white!] (0.5,-1.77) circle (0.01pt)
node[below=1.5pt]{\color{black}\small $v$};

\end{scope}
\end{tikzpicture}

\caption{The star-star relation \eqref{ssr}. The internal spin $\sigma_H$ is summed/integrated on both sides.}
\label{ssfig}
\end{figure}
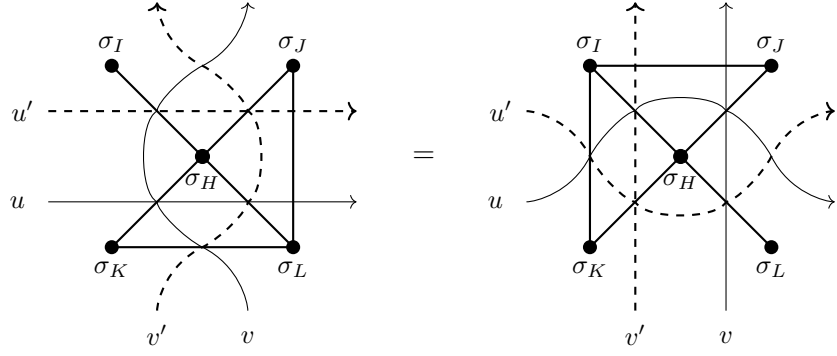

\section{Root-of-unity limit}\label{sec:asymptotics}

\subsection{The limit and the asymptotics of the lens elliptic gamma function}
In this section we study the asymptotics of the lens elliptic gamma function when one of the elliptic nomes tends to a root of unity.
Following \cite{Kels:2017vbc}, we take the limit
\begin{equation}\label{limit}
  \tau:=\tau_1,\qquad \q=\EXP^{\frac{\pi\ii}{rN}}\EXP^{-\frac{\hbar}{r}},\quad\text{i.e.}\quad
  \tau_2=\frac{1}{rN}+\frac{\ii\hbar}{\pi r},\qquad \hbar\to0^+\;,
\end{equation}
and introduce the shifted variables
\begin{equation}\label{hatdef}
  \hat z:=z+\frac{\pi m}{rN},\qquad \hat\tau:=N\tau+\frac1r\;.
\end{equation}
The asymptotics of $\Phi(z,m)$ in this limit were obtained in \cite[(3.3)--(3.6)]{Kels:2017vbc}. In our notation they read, for $z$ in the
strip $|\im z|<(\pi/2)\im\tau$,
\begin{equation}\label{KYasym}
  \log\Phi(z,m)=\hbar^{-1}\,\mathcal F(\hat z)+\log\Phi^{(0)}(z,m)+O(\hbar)\;,
\end{equation}
where the leading term is given by
\begin{equation}\label{lead}
  \mathcal F'(z)=\frac{\ii}{N}\log\ovth_4(Nz\,|\,\hat\tau),\qquad \mathcal F(0)=0\;,
\end{equation}
and the subleading term by
\begin{equation}\label{KYsub}
  \begin{aligned}
    \Phi^{(0)}(z,m) & =\EXP^{E_0(z,m)}\prod_{k_1=0}^{N-1}\prod_{k_2=0}^{r-1}
    \ovth_4\!\Big(\hat z+\frac{\pi(k_1+1)}{N}-\frac{\pi\hat\tau(2k_2+1)}{2N}+\frac{\pi r\tau}{2}\,\Big|\,r\tau\Big)^{e_{k_1k_2}(m)} \\
                    & \qquad\times\prod_{k=0}^{|m|-1}\ovth_4\!\Big(\hat z-\operatorname{sgn}(m)\frac{\pi\hat\tau(2k+1)}{2N}+\frac{\pi r\tau}{2}\,\Big|\,r\tau\Big)^{\operatorname{sgn}(m)}\;,
  \end{aligned}
\end{equation}
with the exponents $e_{k_1k_2}(m):=(rN-2r(k_1+1)+2k_2-2m+1)/(2rN)$, $\operatorname{sgn}(m):=1$ for $m\ge0$ and $\operatorname{sgn}(m):=-1$ for $m<0$, and $E_0(z,m):=E(z,m)|_{\tau_2=1/(rN)}$, cf.\ \eqref{Edef}.
The leading term \eqref{lead} comes from the singular factor $\Phi_1(\,\cdot\,;\p\q,\q^r)$ in \eqref{Phidef}, and \eqref{KYsub} was obtained in
\cite{Kels:2017vbc} by expanding the $O(1)$ part of the singular factor and the non-singular factor, and rearranging the resulting infinite
products into theta functions. The function $\mathcal F$ is odd, and it is periodic, $\mathcal F(z+\pi/N)=\mathcal F(z)$, since $\log\ovth_4$
has no constant Fourier mode. As emphasized in \cite{Kels:2017vbc}, the leading term depends on $z$ and $m$ only through $\hat z$, whereas
$\Phi^{(0)}$ depends on $z$ and $m$ also separately.

In \cite{Kels:2017vbc}, $\Phi(z,m)$ is defined, and \eqref{KYasym}--\eqref{KYsub} are derived, for $m\in\{0,\ldots,r-1\}$. With the
definition \eqref{Phidef} for all $m\in\ZZ$, both extend to all $m\in\ZZ$. First, by \eqref{Phiprop} we have $\Phi(z,m+r)=\Phi(z,m)$, while
$\hat z$ is shifted by $\pi/N$, under which $\mathcal F$ is invariant. Hence \eqref{KYasym} holds for all $m$, with
$\Phi^{(0)}(z,m+r)=\Phi^{(0)}(z,m)$, and $\Phi^{(0)}(\cdot,m)$ is continuous away from the zeros of the theta functions for every $m$; this
is all we shall use below. Second, under the step $(z,m)\to(z+\pi/(rN),m-1)$, which leaves $\hat z$ invariant, $\Phi^{(0)}$ changes as
dictated by the difference equation \eqref{shift} (Section~\ref{sec:omega}), and so does the right-hand side of \eqref{KYsub}: the exponents
$e_{k_1k_2}$ increase by $1/(rN)$, and the last product gains or loses the factor with argument $\hat z-(2m-1)\lambda/2+\pi r\tau/2$, with
$\lambda$ given in \eqref{lamdef}. Hence the explicit formula \eqref{KYsub} holds for all $m\in\ZZ$ as well.

In the saddle point analysis of Section~\ref{sec:expansion} the leading term fixes the hatted variables, and the discrete spins enter only at
$O(\hbar^0)$, through the dependence of $\Phi^{(0)}(z,m)$ on $m$ at fixed $\hat z$. Writing $z=\xi+\pi a/(rN)$ and $m=-a$ with $a\in\ZZ$, so that
$\hat z=\xi$, we thus need the ratio
\begin{multline}\label{KYratio}
  \frac{\Phi^{(0)}(\xi+\frac{\pi a}{rN},-a)}{\Phi^{(0)}(\xi,0)}=\EXP^{E_0(\xi+\frac{\pi a}{rN},-a)}
  \prod_{k_1=0}^{N-1}\prod_{k_2=0}^{r-1}\ovth_4\!\Big(\xi+\frac{\pi(k_1+1)}{N}-\frac{\pi\hat\tau(2k_2+1)}{2N}+\frac{\pi r\tau}{2}\,\Big|\,r\tau\Big)^{\frac{a}{rN}}\\
  \times\prod_{k=0}^{|a|-1}\ovth_4\!\Big(\xi-\operatorname{sgn}(-a)\frac{\pi\hat\tau(2k+1)}{2N}+\frac{\pi r\tau}{2}\,\Big|\,r\tau\Big)^{\operatorname{sgn}(-a)}\;,
\end{multline}
together with the $a$-independent term $\Phi^{(0)}(\xi,0)$, which only enters the normalization. In this form the dependence on $a$ is
not easy to control. It enters through the exponents of $rN$ fractional powers, through the number and the sign of the remaining factors, and
through the cubic polynomial $E_0$, and the two signs of $a$ have to be treated separately. The discrete spins will take values in $\ZZ_{rN}$,
which requires \eqref{KYratio} to be quasi-periodic under $a\to a+rN$; this is not visible in \eqref{KYratio}, and results from a cancellation
between the three types of factors. The phases, which matter for the star-star relation, are distributed over the branches of $rN$ fractional
powers. Finally, the arguments of the theta functions are not written in terms of a single torsion point, which obscures the modular
transformation of Section~\ref{sec:modular} and the fact that the discrete model depends on $r$ and $N$ only through $rN$. For general $n$,
every edge weight contains $n^2$ such factors.

\subsection{The cyclic elliptic gamma function}\label{sec:omega}
We therefore do not use \eqref{KYsub} for $m\ne0$, and compute the ratio in a different way. The step $(z,m)\to(z+\pi/(rN),m-1)$, which
leaves $\hat z$ invariant, is close to the shift $(z,m)\to(z+\pi\tau_2,m-1)$ in the exact difference equation \eqref{shift}, since in the limit
\eqref{limit}
\begin{equation}\label{keyobs}
  \frac{\pi}{rN}=\pi\tau_2-\frac{\ii\hbar}{r}\;.
\end{equation}
Up to a shift of $z$ by $O(\hbar)$, the step is thus given exactly by \eqref{shift}. The effect of the remaining shift on $\log\Phi$ is of order
$\hbar^{-1}\cdot\hbar=O(1)$ and is determined by $\mathcal F$, so that the ratio \eqref{KYratio} is fixed by \eqref{lead} and \eqref{shift},
without any manipulation of infinite products.

To make this explicit, we put $z_k:=\xi+\pi k/(rN)$ and $m=-k$, and use \eqref{keyobs} to split the ratio into two factors,
\begin{equation}
  \frac{\Phi(z_{k+1},-k-1)}{\Phi(z_k,-k)}=
  \frac{\Phi(z_k+\pi\tau_2-\frac{\ii\hbar}{r},-k-1)}{\Phi(z_k+\pi\tau_2,-k-1)}\cdot\frac{\Phi(z_k+\pi\tau_2,-k-1)}{\Phi(z_k,-k)}\;.
\end{equation}
The hatted variables of the numerator and the denominator of the first factor are $\xi$ and $\xi+\ii\hbar/r$ for every $k$, so that this factor
tends to $\exp\!\big(-(\ii/r)\mathcal F'(\xi)\big)=\ovth_4(N\xi|\hat\tau)^{1/(rN)}$. Here the $O(1)$ part of its logarithm, a difference of two values
of $\log\Phi^{(0)}(\cdot,-k-1)$ at points which are $O(\hbar)$ apart, vanishes as $\hbar\to0$, since $\Phi^{(0)}(\cdot,m)$ is continuous for every
$m\in\ZZ$, as shown above; as usual, we assume that the $O(\hbar)$ remainder in \eqref{KYasym} is uniform on compact subsets of
the strip. The second factor is given exactly by \eqref{shift}. At $\tau_2=1/(rN)$ the argument of the theta function in its denominator moves by
$\pi\tau+\pi/(rN)$ at each step $k\to k+1$, and we set
\begin{equation}\label{lamdef}
  \lambda:=\pi\Big(\tau+\frac{1}{rN}\Big)=\frac{\pi(1+rN\tau)}{rN}\;.
\end{equation}
By \eqref{hatdef} we have $\hat\tau=N\lambda/\pi$. Since $rN\lambda=\pi+N\cdot\pi r\tau$, $\lambda$ is an $rN$-torsion point of the elliptic curve $\CC/(\pi\ZZ+\pi r\tau\ZZ)$. In terms of $\lambda$, the
theta function in the denominator of \eqref{shift} becomes
\begin{equation}
  \begin{aligned}
    \ovth_4\big(z_k+\tfrac{\pi}{2}(r+2k+1)\tau+\tfrac{\pi}{2rN}\,\big|\,r\tau\big) & =\ovth_4\big(\xi+\tfrac{\pi r\tau}{2}+(k+\tfrac12)\lambda\,\big|\,r\tau\big)                                              \\
                                                                                   & =\ii\,\EXP^{-\ii(\xi+(k+\frac12)\lambda)-\frac{\pi\ii r\tau}{4}}\ovth_1\big(\xi+(k+\tfrac12)\lambda\,\big|\,r\tau\big)\;.
  \end{aligned}
\end{equation}
Taking the product over $k=0,\ldots,a-1$, the exponential factors involving $E$ telescope, and we obtain
\begin{equation}\label{omegalim}
  \lim_{\hbar\to0}\frac{\Phi\big(\xi+\frac{\pi a}{rN},-a\big)}{\Phi(\xi,0)}=\omega(\xi\,|\,a)\;,
\end{equation}
where
\begin{equation}\label{omegadef}
  \omega(\xi\,|\,a):=\ovth_4(N\xi\,|\,\hat\tau)^{\frac{a}{rN}}\;\EXP^{P_a(\xi)}\;
  \prod_{k=0}^{a-1}\frac{1}{\ovth_1\!\big(\xi+(k+\tfrac12)\lambda\,\big|\,r\tau\big)}\;,
\end{equation}
\begin{equation}\label{Pdef}
  P_a(\xi):=\frac{\ii a^2(\pi r-2\xi)}{2r}-\frac{\pi\ii a^3(rN\tau-rN+2)}{3r^2N}+\frac{\pi\ii a(rN\tau+2rN-6N+2)}{12N}\;.
\end{equation}
Here we use the convention $\prod_{k=0}^{a-1}g_k:=\prod_{k=a}^{-1}g_k^{-1}$ for $a<0$, and the fractional power is taken with the branch of
$\log\ovth_4$ continuous from the real axis. The exponent $P_a$ collects the exponential factors from \eqref{shift} and from the half-period
shift, $P_a=E(\xi+\pi a/(rN),-a)|_{\tau_2=1/(rN)}+\ii a\xi+\ii\lambda a^2/2+\ii a(\pi r\tau/4-\pi/2)$.
We call $\omega$ the \emph{cyclic elliptic gamma function}. It is an elliptic analogue of the Pochhammer symbol with step $\lambda$, and it satisfies
$\omega(\xi|0)=1$ and, as a consequence of the inversion relation, $\omega(-\xi|-a)=\omega(\xi|a)^{-1}$. Moreover, the last relation in
\eqref{Phiprop} implies
\begin{equation}\label{omegaper}
  \omega(\xi\,|\,a+rN)=\EXP^{-\frac{\pi\ii a(a+r)}{r}}\,\omega(\xi\,|\,a)\;.
\end{equation}
In particular $\omega(\xi|rN)=1$, which is equivalent to the multiplication formula
$\prod_{l=0}^{rN-1}\ovth_1(\xi+(l+1/2)\lambda|r\tau)=\ovth_4(N\xi|\hat\tau)\,\EXP^{P_{rN}(\xi)}$. Both sides of this formula have
zeros on $\lambda/2+(\pi/N)\ZZ+\lambda\ZZ$. We thus see that the fractional power in \eqref{omegadef} ensures that $\omega$ is
quasi-periodic in $a$, i.e.\ periodic up to the phase in \eqref{omegaper}, as in the $\ZZ_N$ models of \cite{Bazhanov:2010kz}. This quasi-periodicity, which allows
$a$ to be regarded as a $\ZZ_{rN}$ variable, is the reason for the adjective ``cyclic''.

By \eqref{KYasym} and \eqref{omegalim}, $\omega(\xi|a)$ is equal to \eqref{KYratio}; by the multiplication formula, the product of the $rN$
theta functions there is $\ovth_4(N\xi|\hat\tau)$ up to an exponential factor. In the form \eqref{omegadef}, however, the difficulties
mentioned above do not arise, and since the step is the torsion point $\lambda$, the modular transformation of Section~\ref{sec:modular} can be
carried out factor by factor. Since the Boltzmann weights \eqref{BW} are products of
$\Phi$'s, the discrete weights for general $n$ are products of $\omega$'s, see \eqref{dW}.

\paragraph{The $a$-independent term.}
The subleading term at $m=0$, $\Psi(z):=\log\Phi^{(0)}(z,0)$, will be needed only for the normalization in Section~\ref{sec:normalisation}. This
is the only place where we use \eqref{KYsub}, which gives
\begin{equation}\label{psidef}
  \begin{aligned}
    \log\Phi(z,0) & =\hbar^{-1}\mathcal F(z)+\Psi(z)+O(\hbar)\;,                                                                                  \\
    \Psi(z)       & =\sum_{k_1=0}^{N-1}\sum_{k_2=0}^{r-1}\frac{rN-2r(k_1+1)+2k_2+1}{2rN}                                                         \\
                  & \qquad\times\log\ovth_4\!\Big(z+\frac{\pi(k_1+1)}{N}-\frac{\pi\hat\tau(2k_2+1)}{2N}+\frac{\pi r\tau}{2}\,\Big|\,r\tau\Big)\;.
  \end{aligned}
\end{equation}
This expression can also be written in terms of the torsion point $\lambda$. Writing $\pi/N=r\lambda-\pi r\tau$, the arguments of the
theta functions become $z+(j-1/2)\lambda+\pi r\tau/2-\pi r\tau(k_1+1)$ with $j:=r(k_1+1)-k_2\in\{1,\ldots,rN\}$, and the
coefficients become $(rN-2j+1)/(2rN)$. The half-period shift and the quasi-periodicity \eqref{thetaper}, $\ovth_1(v-k\pi r\tau|r\tau)\propto\EXP^{2\ii kv}\ovth_1(v|r\tau)$,
then give
\begin{equation}\label{psitorsion}
  \Psi(z)=\sum_{j=1}^{rN}\frac{rN-2j+1}{2rN}\,\log\ovth_1\big(z+(j-\tfrac12)\lambda\,\big|\,r\tau\big)-\frac{\ii r(N^2-1)}{6}\,z+\text{const}\;,
\end{equation}
where the coefficient of the linear term is obtained from
$\ii\sum_{k_1,k_2}(rN-2r(k_1+1)+2k_2+1)/(2rN)\,(2k_1+1)=-\ii r(N^2-1)/6$.
This linear term drops out of the star-star relation, and apart from it $\Psi$ depends on $(r,N)$ only through $rN$, $\lambda$ and $r\tau$.

\subsection{The vertex factor}\label{sec:vertex}
Next we consider the vertex factor
\begin{equation}\label{Phipm}
  \Phi^{\pm}(x,m):=\Phi(-\ii\eta+x,m)\,\Phi(-\ii\eta-x,-m)\;,
\end{equation}
for which the exact expression \eqref{pairexact} is available. With $x=\xi+\pi a/(rN)$ and $m=-a$, we have $x-\pi m\tau_1=\xi+a\lambda$ and
$x+\pi m\tau_2=\xi-\ii a\hbar/r$, while $r\tau_2=1/N+\ii\hbar/\pi$. Applying the modular transformation $\tau\mapsto-\tau/(N\tau-1)$
to $\ovth_1(x|r\tau_2)$, cf.\ \eqref{thetamod}, we find, for $0<\re N\xi<\pi$,
\begin{equation}\label{pairasym}
  \log\Phi^{\pm}(\xi,0)=\hbar^{-1}\Big(\mathcal C(\xi)+\frac{\pi^2}{12N^2}\Big)+\log\ovth_1(\xi\,|\,r\tau)+c_0+O(\hbar),\qquad
  \mathcal C(\xi):=-\Big(\xi-\frac{\pi}{2N}\Big)^2\;,
\end{equation}
where $c_0$ is independent of $\xi$ and $\hbar$, and $\mathcal C$ is extended with period $\pi/N$. Here it is important that $\ovth_1=\vartheta_1/G$
(Appendix~\ref{app:theta}): the factor $\hbar^{-1/2}$ produced by the modular transformation of $\vartheta_1(x|r\tau_2)$ is cancelled by the same
factor in the asymptotics of $G(r\tau_2)$, which also produces the $\xi$-independent term $\pi^2/(12N^2\hbar)$. This term drops out of the
star-star relation, since $S$ appears once on each side, and we will omit it below. In \cite{Kels:2017vbc} the vertex weight contains in
addition the factor $G(r\tau_1)G(r\tau_2)$, whose asymptotics give the factor $1/\sqrt\hbar$ there. For the dependence on the discrete variable we find
\begin{equation}\label{pairlim}
  \lim_{\hbar\to0}\frac{\Phi^{\pm}(\xi+\frac{\pi a}{rN},-a)}{\Phi^{\pm}(\xi,0)}
  =\EXP^{\frac{\ii\lambda a^2}{r}+\frac{\ii a(2N\xi-\pi)}{rN}}\;\frac{\ovth_1(\xi+a\lambda\,|\,r\tau)}{\ovth_1(\xi\,|\,r\tau)}\;,
\end{equation}
where we used $\hbar\,\partial_x\log\ovth_1(x|r\tau_2)\to-(2N\xi-\pi)/N$. For other values of $\re N\xi$, the variable $\xi$ in the
exponential is to be replaced by $\xi-(\pi/N)\lfloor \re(N\xi)/\pi\rfloor$.

\subsection{Modular transformation}\label{sec:modular}
The asymptotic expressions above simplify after a modular transformation. As in \cite{Kels:2017vbc}, we define
\begin{equation}\label{taupdef}
  \tau':=\frac{rN\tau}{1+rN\tau},\qquad \phi:=\frac{\pi}{\lambda}\,\xi,\qquad M:=rN,\qquad T:=\frac{\tau'}{N}=\frac{r\tau}{1+rN\tau}\;.
\end{equation}
Then $r\tau'=MT=r-1/\hat\tau$. We apply the modular transformation $\big(\begin{smallmatrix}1&0\\N&1\end{smallmatrix}\big)$ to
$\ovth_1(\cdot|r\tau)$, which by \eqref{thetamod} gives
\begin{equation}
  \ovth_1(u|r\tau)=c_1\,\EXP^{-\frac{\ii N u^2}{\pi(1+rN\tau)}}\,\ovth_1\Big(\frac{u}{1+rN\tau}\,\Big|\,T\Big),\qquad \frac{u}{1+rN\tau}=\frac{\pi u}{M\lambda}\;,
\end{equation}
and we apply $\tau\mapsto-1/\tau$ followed by $\tau\mapsto\tau+r$ to $\ovth_4(\cdot|\hat\tau)$, which gives
$\ovth_4(N\xi|\hat\tau)=c_4\,\EXP^{-\ii N^2\xi^2/(\pi\hat\tau)}\ovth_2(\phi|r\tau')$. Here $c_1$ and $c_4$ are constants.
In \eqref{omegadef} the Gaussian factors combine into
\begin{equation}
  -\frac{\ii N\xi^2a}{\pi r\hat\tau}+\frac{\ii N}{\pi r\hat\tau}\sum_{k=0}^{a-1}\Big(\xi+\big(k+\tfrac12\big)\lambda\Big)^2
  =\frac{\ii a^2\xi}{r}+\frac{\ii\lambda}{r}\Big(\frac{a^3}{3}-\frac{a}{12}\Big)\;,
\end{equation}
and the $\xi$-dependent term cancels exactly against the term $-\ii a^2\xi/r$ in $P_a(\xi)$. We therefore obtain
\begin{equation}\label{omegamod}
  \omega(\xi\,|\,a)=\EXP^{\tilde P_a}\;\mathfrak w(\phi\,|\,a),\qquad
  \mathfrak w(\phi\,|\,a):=\frac{\ovth_2(\phi\,|\,MT)^{\frac{a}{M}}}{\prod_{k=0}^{a-1}\ovth_1\!\big(\frac{\phi+(k+\frac12)\pi}{M}\,\big|\,T\big)}\;,
\end{equation}
\begin{equation}
  \tilde P_a:=\frac{\pi\ii}{2}a^2+\frac{\pi\ii(rN-1)}{3r^2N}a^3+c_P\,a\;,
\end{equation}
with a constant $c_P$. The prefactor $\EXP^{\tilde P_a}$ does not contribute to the star-star relation. In an edge weight (see
\eqref{dW} below) we have
\begin{equation}
  \sum_{i,j=1}^n\tilde P_{a_{A,i}-a_{B,j}}=\frac{\pi\ii n}{2}\big(|\va_A|^2+|\va_B|^2\big)+\frac{\pi\ii n(rN-1)}{3r^2N}\Big(\sum_ia_{A,i}^3-\sum_ja_{B,j}^3\Big)\;,
\end{equation}
where the linear term has dropped out since $\sum_ia_i=0$. This is a product of one-vertex factors. Since $|\va|^2$ is even, the first term
only gives a sign, and counting the edges incident to each spin on the two sides of the star-star relation shows that all these factors cancel.
The vertex factor can be treated in the same way. For $0<\re N\xi<\pi$ we find
\begin{equation}\label{pairmod}
  \EXP^{\frac{\ii\lambda a^2}{r}+\frac{\ii a(2N\xi-\pi)}{rN}}\,\ovth_1(\xi+a\lambda\,|\,r\tau)
  =\mathcal N(\xi)\;\ovth_4\Big(\frac{\phi+\pi a}{M}+\frac{\pi T}{2}\,\Big|\,T\Big)
\end{equation}
with $\mathcal N(\xi)$ independent of $a$. Here again the Gaussian factor of the modular transformation cancels the term proportional to $a\xi$.

\section{Quasi-classical expansion of the star-star relation}\label{sec:expansion}

\subsection{Change of variables}
For every spin we write
\begin{equation}\label{spinparam}
  x_{c,i}=\xi_{c,i}+\frac{\pi a_{c,i}}{rN},\qquad m_{c,i}\equiv-a_{c,i}\ (\mathrm{mod}\ r),\qquad
  \sum_i\xi_{c,i}=0,\quad \sum_i a_{c,i}=0,\quad a_{c,i}\in\ZZ\;.
\end{equation}
This decomposition is not unique, since $(\xi_{c,i},a_{c,i})\to(\xi_{c,i}+(\pi/N)k_i,\;a_{c,i}-rk_i)$ leaves $(x_{c,i},m_{c,i})$ invariant.
By \eqref{KYasym} the leading asymptotics depend only on $\hat x_{c,i}:=x_{c,i}+\pi m_{c,i}/(rN)$, which equals $\xi_{c,i}$ modulo $(\pi/N)\ZZ$.
Since $\mathcal F$ is $\pi/N$-periodic, the leading term is therefore a function of the $\xi$'s only.
For the external spins \eqref{spinparam} is just a change of variables. As in \cite{Kels:2017vbc}, the continuous parameters
$\xi_{c,i}$ of the external spins will become additional parameters of the discrete model, which play a role similar to rapidities.
In the following, the upper-case letters $I,J,K,L,H$ label the vertices of the star-star relation (see Fig.~\ref{ssfig}), and the lower-case
letters label the components of the spins. For example, $\xi_{K,i}$ is the $i$-th component of the spin at the vertex $K$.

In the limit \eqref{limit} the crossing parameter becomes $\eta=-\ii\lambda/2+\hbar/(2r)$, so that $\ii\eta\to\lambda/2$.
Accordingly, we introduce the spectral variable
\begin{equation}\label{thetadef}
  \theta:=\frac{\pi}{\lambda}\,\ii\alpha\;,
\end{equation}
so that $\alpha\to\eta$ corresponds to $\theta\to\pi/2$. This is the same variable $\theta$ as in \cite[(3.16)]{Kels:2017vbc}.
It is convenient to measure the rapidities of the star-star relation in the same units, $\theta_x:=(\pi/\lambda)\,\ii x$ for
$x\in\{u,u',v,v'\}$. An edge with spectral parameter $x-y$ then has the spectral variable $\theta_{xy}:=\theta_x-\theta_y=(\pi/\lambda)\,\ii(x-y)$; the weights of the discrete models below carry these spectral variables as subscripts.

\subsection{Leading order: classical star-star relation}

The leading asymptotics of the weights are
$\log W_\alpha(\sigma_a,\sigma_b)=\hbar^{-1}\sum_{i,j}\mathcal F(\xi_{a,i}-\xi_{b,j}+\ii\alpha)+O(1)$ and
$\log S(\sigma)=\hbar^{-1}\sum_{i<j}\mathcal C(\xi_i-\xi_j)+O(1)$, where the constant $\pi^2/(12N^2\hbar)$ of \eqref{pairasym} is omitted. The ``action'' of the first star \eqref{W1} is
\begin{multline}
  \mathcal A^{(1)}(\vxi_H):=\sum_{i<j}\mathcal C(\xi_{H,i}-\xi_{H,j})+\sum_{i,j}\Big[
    \mathcal F\big(\xi_{K,i}-\xi_{H,j}+\tfrac{\lambda}{2}-\ii(u-v)\big)
    +\mathcal F\big(\xi_{J,i}-\xi_{H,j}+\tfrac{\lambda}{2}-\ii(u'-v')\big)\\
    +\mathcal F\big(\xi_{H,i}-\xi_{I,j}+\ii(u'-v)\big)
    +\mathcal F\big(\xi_{H,i}-\xi_{L,j}+\ii(u-v')\big)\Big]\;,
\end{multline}
and $\mathcal A^{(2)}$ is defined similarly from \eqref{W2}. We impose the constraint $\sum_i\xi_{H,i}=0$ with a Lagrange multiplier $\mu$, so that
the saddle point equations read $\partial\mathcal A^{(1)}/\partial\xi_{H,p}=\mu$ for $p=1,\ldots,n$. Explicitly, with $\mathcal L(z):=\log\ovth_4(Nz|\hat\tau)$, they are
\begin{multline}\label{saddleexplicit}
  \frac{\ii}{N}\sum_{i=1}^n\Big[-\mathcal L\big(\xi_{K,i}-\xi_{H,p}+\tfrac{\lambda}{2}-\ii(u-v)\big)
    -\mathcal L\big(\xi_{J,i}-\xi_{H,p}+\tfrac{\lambda}{2}-\ii(u'-v')\big)\\
    +\mathcal L\big(\xi_{H,p}-\xi_{I,i}+\ii(u'-v)\big)+\mathcal L\big(\xi_{H,p}-\xi_{L,i}+\ii(u-v')\big)\Big]
  +\sum_{q\neq p}\mathcal C'(\xi_{H,p}-\xi_{H,q})=\mu\;.
\end{multline}

These equations simplify considerably in the variables \eqref{taupdef}. To see this, we write $\mathcal L(z)=\log\ovth_2(\phi_z|MT)-\ii N^2z^2/(\pi\hat\tau)+\text{const}$
with $\phi_z:=(\pi/\lambda) z$, and multiply \eqref{saddleexplicit} by $N/\ii$. For a $W$-leg the Gaussian terms give
$\sum_i(\xi_{H,p}-\xi_{X,i}+y)^2=n(\xi_{H,p}+y)^2+|\vxi_X|^2$, where $y$ is the shift of the leg. A $\olW$-leg with shift $y'$ enters with the
opposite sign, and gives $\sum_i(\xi_{X,i}-\xi_{H,p}+y')^2=n(y'-\xi_{H,p})^2+|\vxi_X|^2$. Since there are two legs of each type, the terms $\propto\xi_{H,p}^2$ cancel.
Moreover, the four shifts add up to $(\lambda/\pi)\big[\theta_{u'v}+\theta_{uv'}+(\pi/2-\theta_{uv})+(\pi/2-\theta_{u'v'})\big]=\lambda$,
so that the Gaussian terms add up to $-(\ii N^2/(\pi\hat\tau))2n\lambda\,\xi_{H,p}+(\text{indep.\ of }p)=-2\ii nN\xi_{H,p}+\ldots$
On the other hand, using $\sum_q\xi_{H,q}=0$, the vertex term gives $(N/\ii)\sum_{q\ne p}\mathcal C'(\xi_{H,pq})=\ii\sum_{q\neq p}(2N\xi_{H,pq}-\pi)=2\ii nN\xi_{H,p}+\text{const}$,
where the constant does not depend on $p$ modulo $2\pi\ii$. Here we used that $\mathcal C'$ is piecewise linear with period $\pi/N$, so that on
each chamber it differs from $-(2N\xi-\pi)/N$ by a multiple of $2\pi/N$. The two contributions cancel exactly, and the saddle point equations
become, modulo $2\pi\ii$,
\begin{equation}\label{saddlemod}
  \sum_{i=1}^n\log\frac{\ovth_2(\phi_{H,p}-\phi_{I,i}+\theta_{u'v}\,|\,MT)\;\ovth_2(\phi_{H,p}-\phi_{L,i}+\theta_{uv'}\,|\,MT)}
  {\ovth_2(\phi_{K,i}-\phi_{H,p}+\frac\pi2-\theta_{uv}\,|\,MT)\;\ovth_2(\phi_{J,i}-\phi_{H,p}+\frac\pi2-\theta_{u'v'}\,|\,MT)}=\mu'\;,
\end{equation}
for $p=1,\ldots,n$, where $\mu'$ differs from $N\mu/\ii$ by a $p$-independent constant. Equivalently, \eqref{saddlemod} holds as an
equation for the exponentials of both sides. This is a multi-component analogue of the three-leg form of the $Q4$ equation \cite{Adler:2002abs}.

At order $O(\hbar^{-1})$ the star-star relation gives the classical star-star relation
\begin{equation}\label{cssr}
  \mathcal A^{(1)}(\vxi_H^{(1)})+\Lambda_{\rm ext}^{(1)}=\mathcal A^{(2)}(\vxi_H^{(2)})+\Lambda_{\rm ext}^{(2)}\;,
\end{equation}
where $\vxi_H^{(1,2)}$ are saddle points and $\Lambda^{(1,2)}_{\rm ext}$ are the leading terms of the external factors in \eqref{ssr},
\begin{equation}\label{Lext}
  \begin{aligned}
    \Lambda^{(1)}_{\rm ext} & :=\sum_{i,j}\Big[\mathcal F\big(\xi_{L,i}-\xi_{K,j}+\ii(v'-v)\big)+\mathcal F\big(\xi_{L,i}-\xi_{J,j}+\ii(u'-u)\big)\Big]\;, \\
    \Lambda^{(2)}_{\rm ext} & :=\sum_{i,j}\Big[\mathcal F\big(\xi_{J,i}-\xi_{I,j}+\ii(v'-v)\big)+\mathcal F\big(\xi_{K,i}-\xi_{I,j}+\ii(u'-u)\big)\Big]\;.
  \end{aligned}
\end{equation}
In general each star has several saddle points, in addition to their $n!$ Weyl images, and the $O(\hbar^{-1})$ part of \eqref{ssr} only implies
that the critical values of the dominant saddle points of the two sides agree. We shall call saddle points $\vxi_H^{(1)}$ and $\vxi_H^{(2)}$ of the two
stars {\it corresponding} if their critical values agree, i.e.\ if \eqref{cssr} holds. The assumptions on the saddle points under which the
relations below are derived are collected in Section~\ref{sec:discrete}.

The leading asymptotics \eqref{lead} and \eqref{pairasym} have the same form as those of the quasi-classical limit $\p\to1$ of the elliptic gamma
function in \cite[Sec.~4]{Bazhanov:2011mz}. Namely, we have $\mathcal F(z)=-N^{-2}\lambda_4(Nz|\hat\tau)$ with $\lambda_4(z|\tau):=-\ii\int_0^z\log\ovth_4(w|\tau)\,dw$,
and $\mathcal C$ is the asymptotics of $\log\ovth_1$ there. Consequently, after the rescaling $\vxi\to N\vxi$ (cf.\ \eqref{scaling}), \eqref{cssr} and
\eqref{saddlemod} are the classical star-star relation and the corresponding equations of motion of \cite[Secs.~4.2--4.3]{Bazhanov:2011mz}, with
the nome $\hat\tau$. In particular, they do not depend on $r$ and $N$ separately, and the classical star-star relation \eqref{cssr}, i.e.\ the equality of the
critical values of the two stars on solutions of \eqref{saddlemod}, follows from the results of \cite{Bazhanov:2011mz}. The hyperbolic counterparts of these equations have been studied in \cite{Kels:2024afm}, and related multi-component quad
equations have been studied in \cite{Kels:2018qzx}.

\paragraph{$n=2$: the $Q4$ equation.}
For $n=2$ the saddle point equations can be compared with those of \cite{Kels:2017vbc}. Here and below, the label KY refers to the conventions
of \cite{Kels:2017vbc}, in which the first component of a spin is $\phi^{\rm KY}:=\phi+\pi r\tau'/4$. For a leg with spectral variable
$\theta_e$ ($\theta_e:=\theta$ for a $W$-leg and $\pi/2-\theta$ for a $\olW$-leg) joining $H$ and $X$, the half-period shift
$\ovth_2(x-\pi\tau/2|\tau)=\EXP^{\ii x-\pi\ii\tau/4}\ovth_3(x|\tau)$ shows that its contribution to the difference of the $p=1$ and $p=2$
equations \eqref{saddlemod} is
\begin{equation}\label{q4leg}
  \Upsilon_{Q4}\big(\theta_e\,\big|\,\phi^{\rm KY}_H,\phi^{\rm KY}_X\big)+2\ii\theta_e\;,
\end{equation}
where $\Upsilon_{Q4}$ is the building block of the three-leg form of $Q4$ in \cite[(3.23)]{Kels:2017vbc}:
\begin{equation}\label{UpsQ4}
  \Upsilon_{Q4}(\theta|\phi_i,\phi_j):=\log\frac{\ovth_3(\theta+\phi_i+\phi_j|r\tau')}{\ovth_3(\theta-\phi_i-\phi_j|r\tau')}+\log\frac{\ovth_2(\theta+\phi_i-\phi_j|r\tau')}{\ovth_2(\theta-\phi_i+\phi_j|r\tau')}\;.
\end{equation}
Since the spectral variables of the
four legs add up to $\pi$, the saddle point equation is $\sum_e\Upsilon_{Q4}(\theta_e|\phi^{\rm KY}_H,\phi^{\rm KY}_{X_e})=0$ mod $2\pi\ii$, the
four-leg counterpart of \cite[(3.23)]{Kels:2017vbc}.

\subsection{Saddle points and discrete spins}
The action $\mathcal A^{(1)}$ is invariant under $\xi_{H,p}\to\xi_{H,p}+(\pi/N)k_p$, so that if $\vxi_H$ is a saddle point, so is every point
\begin{equation}\label{saddleset}
  \vx_H=\vxi_H+\frac{\pi}{rN}\va_H,\qquad \va_H\in\ZZ^n,\quad\sum_p a_{H,p}=0,\qquad \vm_H\equiv-\va_H\ (\mathrm{mod}\ r)\;.
\end{equation}
Here each $\va_H$ corresponds to a pair consisting of a value of $\vm_H$ and an image of $\vxi_H$ under the $(\pi/N)\ZZ^n$-shifts.
The integrand is invariant under $x_p\to x_p+\pi k_p$ with $\sum k_p=0$, and modulo this periodicity $\va_H$ ranges over
\begin{equation}
  Q/rNQ\;\cong\;(\ZZ_{rN})^{\,n-1}\;,\qquad Q:=\{\va\in\ZZ^n:\textstyle\sum_pa_p=0\}\;,
\end{equation}
where $Q$ is the root lattice of $\mathfrak{sl}(n)$. This accounts for $r^{n-1}$ values of $\vm_H$ times $N^{n-1}$ images of $\vxi_H$. The $n!$ Weyl
images of $\vxi_H$ give identical contributions. As in \cite{Kels:2017vbc}, we shall assume that all saddle points \eqref{saddleset} contribute with
equal weight, i.e.\ that the ``replica symmetry'' is unbroken (assumption (A2) of Section~\ref{sec:discrete}). The Gaussian integration around
each of the saddle points gives the same factor $\big[\det\big(-\partial^2\mathcal A^{(1)}/(2\pi\hbar)\big)\big]^{-1/2}$, since the integrand behaves
as $\EXP^{\mathcal A^{(1)}/\hbar}$ and the integration contour is deformed into the steepest descent contour through the saddle point. Here $\partial^2\mathcal A^{(1)}$
denotes the $(n-1)\times(n-1)$ Hessian with respect to $\xi_{H,1},\ldots,\xi_{H,n-1}$, with $\xi_{H,n}=-\sum_{p<n}\xi_{H,p}$, and the branch of the
square root is fixed by continuity in the parameters. Only the ratio of these factors for the two stars enters below, see \eqref{hessratio}.

\subsection{Subleading order: the discrete star-star relation}\label{sec:discrete}

Let us now collect the $O(\hbar^0)$ terms that depend on the discrete variables. For spins $s=(\vxi,\va)$, \eqref{omegalim} and
\eqref{pairlim} give the Boltzmann weights
\begin{align}
  \bW_\theta(s_a,s_b) & :=\prod_{i,j=1}^n\omega\Big(\xi_{a,i}-\xi_{b,j}+\tfrac{\lambda}{\pi}\theta\;\Big|\;a_{a,i}-a_{b,j}\Big)\;,
  \qquad \obW_\theta(s_a,s_b):=\bW_{\frac{\pi}{2}-\theta}(s_a,s_b)\;,\label{dW}                                                   \\
  \bS(s)              & :=\prod_{i<j}\EXP^{\frac{\ii\lambda a_{ij}^2}{r}+\frac{\ii a_{ij}(2N\tilde\xi_{ij}-\pi)}{rN}}\;
  \ovth_1\big(\xi_{ij}+a_{ij}\lambda\,\big|\,r\tau\big)\;,\label{dS}
\end{align}
where $\xi_{ij}:=\xi_i-\xi_j$, $a_{ij}:=a_i-a_j$, and $\tilde\xi_{ij}:=\xi_{ij}-(\pi/N)\lfloor\re(N\xi_{ij})/\pi\rfloor$.
Since $a_{H,p}\to a_{H,p}+rNk_p$ with $\sum_pk_p=0$ corresponds to $x_{H,p}\to x_{H,p}+\pi k_p$, under which the integrand of \eqref{ssr} is
invariant, the summand below is periodic in the internal discrete spin. In particular, the phases in \eqref{omegaper} of the individual factors
$\omega$ cancel in the summand, and the sum over $\va_H$ may be taken over any set of representatives. The individual edge weights, on the
other hand, depend on the lifts $\va\in\ZZ^n$ of the external discrete spins through such phases. Since a change of the lift of an external
spin by $rNQ$ at fixed $\vxi$ corresponds to a shift of its continuous components by $\pi Q$, and \eqref{ssr} holds at
every point, the relation below holds for any choice of lifts of the external spins, provided that the same lifts are used on both sides.
For definiteness, one may take the representatives $0\le a_1,\ldots,a_{n-1}<rN$, $a_n=-\sum_{i<n}a_i$.
Comparing the $O(\hbar^0)$ terms of the two sides of \eqref{ssr}, we obtain the $(\ZZ_{rN})^{n-1}$ star-star relation
\begin{multline}\label{dssr}
  \bW_{\theta_{v'v}}(s_L,s_K)\,\bW_{\theta_{u'u}}(s_L,s_J)\\
  \times\sum_{\va_H\in(\ZZ_{rN})^{n-1}}\bS(s_H)\,\obW_{\theta_{uv}}(s_K,s_H)\,\obW_{\theta_{u'v'}}(s_J,s_H)\,
  \bW_{\theta_{u'v}}(s_H,s_I)\,\bW_{\theta_{uv'}}(s_H,s_L)\Big|_{\vxi_H=\vxi_H^{(1)}}\\
  =\mathcal R\;\bW_{\theta_{v'v}}(s_J,s_I)\,\bW_{\theta_{u'u}}(s_K,s_I)\\
  \times\sum_{\va_H\in(\ZZ_{rN})^{n-1}}\bS(s_H)\,\obW_{\theta_{uv}}(s_H,s_J)\,\obW_{\theta_{u'v'}}(s_H,s_K)\,
  \bW_{\theta_{u'v}}(s_L,s_H)\,\bW_{\theta_{uv'}}(s_I,s_H)\Big|_{\vxi_H=\vxi_H^{(2)}}\;,
\end{multline}
for arbitrary discrete spins $\va_I,\va_J,\va_K,\va_L\in(\ZZ_{rN})^{n-1}$. The derivation applies directly to external continuous spins
$\vxi_I,\vxi_J,\vxi_K,\vxi_L$ and spectral parameters for which all edge arguments lie in the strip of validity of \eqref{KYasym} and the
corresponding saddle points are non-degenerate; outside this region \eqref{dssr} is understood by analytic continuation. Among the solutions of
the exponentiated saddle point equations \eqref{saddlemod}, the saddle points are those for which the logarithms in \eqref{saddleexplicit} are on
the branch of the action $\mathcal A$; in practice we select them by continuity from a known saddle point. Here
$\vxi_H^{(1)},\vxi_H^{(2)}$ are corresponding solutions of the saddle point equations of the two stars, and $\mathcal R$ is a factor independent
of the discrete spins, which will be given in \eqref{Kdef} below.

\paragraph{Assumptions.}
The derivation of \eqref{dssr} rests on the following assumptions on the saddle points, which are the same as in \cite{Bazhanov:2010kz,Kels:2017vbc}.
\begin{itemize}
  \item[(A1)] The integrals \eqref{W1} and \eqref{W2} are dominated by their saddle points, i.e.\ the integration contour can be deformed into
        steepest descent contours through the saddle points, and the corrections are of relative order $O(\hbar)$.
  \item[(A2)] All saddle points \eqref{saddleset} in the orbit of $\vxi_H$ under the shifts by $(\pi/N)\ZZ^n$ contribute with equal weight.
  \item[(A3)] The saddle points of the two stars come in corresponding pairs in the sense of Section~\ref{sec:expansion}, and the parameters are
        such that each critical value is attained by a single orbit of saddle points under the shifts \eqref{saddleset} and the Weyl group, so that the
        contributions of the two sides can be compared individually; the relation is then extended by analytic continuation. The same condition is used
        for the ratio of the Hessians in Section~\ref{sec:normalisation}.
\end{itemize}
In addition, we use the uniformity of the $O(\hbar)$ remainder in \eqref{KYasym} on compact subsets of the strip
(Section~\ref{sec:asymptotics}), which is standard but not proven either. None of these assumptions is proven. They are supported by the numerical checks of Appendix~\ref{app:numerics}, in which the sum over $\va_H$ is
taken with equal weights and every critical value found on one side was also found on the other. Under these assumptions, \eqref{dssr} holds for all discrete spins whenever $\vxi_H^{(1)}$, $\vxi_H^{(2)}$ are corresponding non-degenerate
saddle points with all edge arguments in the strip of validity of \eqref{KYasym}, with $\mathcal R=\EXP^{\mathcal K^{(2)}-\mathcal K^{(1)}}$
given by \eqref{Kdef}. The cases in which \eqref{dssr} reduces to relations proven in the literature are listed in Section~\ref{sec:proofs}.

\subsection{The normalization factor}\label{sec:normalisation}

The factor $\mathcal R$ collects the $O(\hbar^0)$ contributions which do not depend on the discrete spins. There are four such contributions.
The first one is $\Psi$ of \eqref{psidef}, for every factor $\Phi$ of every edge weight. The second one arises for every factor of an edge weight
$\olW$ from the $O(\hbar)$ part of $\ii\eta=\lambda/2+\ii\hbar/(2r)$, and by \eqref{lead} it is given by
$\exp((\ii/(2r))\mathcal F'(z_0))=\ovth_4(Nz_0|\hat\tau)^{-1/(2rN)}$. The third one is the Gaussian factor, and the fourth one is the
vertex term $\log\ovth_1(\xi_{H,ij}|r\tau)$ in \eqref{pairasym}, which is already included in $\bS$. The remaining factors, namely constants,
powers of $\hbar$, the number of saddle points and the Weyl factor, are the same on both sides.

The Hessian can be eliminated as in \cite[Sec.~3.4]{Kels:2017vbc}. To this end we compare models with different parameters. We refer to
the root-of-unity limit with lens index $r$, root of unity $\q\to\EXP^{\pi\ii/(rN)}$ and $\tau_1=\tau$ in \eqref{limit} as the
$(r,N,\tau)$ model, and indicate this dependence by subscripts when necessary. The leading action has the scaling property
\begin{equation}\label{scaling}
  \mathcal A_{r,N,\tau}(\vxi)=N^{-2}\,\mathcal A_{1,1,\check\tau}(N\vxi),\qquad \check\tau:=N\tau+\frac1r-1\;.
\end{equation}
Here $\check\tau$ is chosen such that the $(1,1,\check\tau)$ model, i.e.\ the model with $r=N=1$ and $\tau$ replaced by $\check\tau$, has the
same $\hat\tau$ as the $(r,N,\tau)$ model: by \eqref{hatdef}, $\hat\tau=\check\tau+1$ for the former and $\hat\tau=N\tau+1/r$ for the latter.
The relation \eqref{scaling} holds since $\mathcal F_{r,N,\tau}(z)=N^{-2}\mathcal F_{1,1,\check\tau}(Nz)$, $\mathcal C_N(\xi)=N^{-2}\mathcal C_1(N\xi)$, and $N\lambda=\pi\hat\tau$ is the $\lambda$ of the
$(1,1,\check\tau)$ model. The saddle points of the two models therefore correspond to each other, and their Hessians agree. For $r=N=1$ the discrete
sums are trivial, and the $O(\hbar^0)$ part of the star-star relation of that model therefore determines the ratio of the Hessian determinants as
\begin{equation}\label{hessratio}
  \left(\frac{\det\partial^2\mathcal A^{(1)}(\vxi^{(1)}_H)}{\det\partial^2\mathcal A^{(2)}(\vxi^{(2)}_H)}\right)^{1/2}
  =\frac{\mathcal H^{(1)}}{\mathcal H^{(2)}},\qquad
  \mathcal H^{(k)}:=\prod_{\olW\text{ edges}}\prod_{i,j}\ovth_4(Nz_{0}\,|\,\hat\tau)^{-\frac12}\;\prod_{i<j}\ovth_1\big(N\xi^{(k)}_{H,ij}\,\big|\,\hat\tau\big)\;,
\end{equation}
where the products run over the edges of the star $k=1,2$, and $z_0$ are the arguments of the $\olW$-factors. Here we used the fact that
$\ovth_1(x|\check\tau)$ and $\ovth_1(x|\hat\tau)$ differ by a constant. Note that this step uses the $O(\hbar^0)$ part of the star-star
relation of the $(1,1,\check\tau)$ model, and hence the assumptions (A1)--(A3) for that model, as well as the correspondence of the saddle points
under \eqref{scaling}. The normalization $\mathcal R$ obtained below is therefore conditional in the same sense as \eqref{dssr}. The independence
of the ratio of the two sides of \eqref{dssr} from the discrete spins, on the other hand, does not involve the Hessians, which are the same for
all terms of each sum. Putting everything together, we find
$\mathcal R=\EXP^{\mathcal K^{(2)}-\mathcal K^{(1)}}$ with
\begin{equation}\label{Kdef}
  \mathcal K^{(k)}:=\sum_{W\text{ edges}}\sum_{i,j}\Psi(z)
  +\sum_{\olW\text{ edges}}\sum_{i,j}\Big[\Psi(z_0)+\frac{rN-1}{2rN}\log\ovth_4(Nz_0\,|\,\hat\tau)\Big]
  -\sum_{i<j}\log\ovth_1\big(N\xi^{(k)}_{H,ij}\,\big|\,\hat\tau\big)\;,
\end{equation}
where the sums run over all six edges of side $k$ of \eqref{dssr}, $z:=\xi_{A,i}-\xi_{B,j}+(\lambda/\pi)\theta$ for $W_\theta(s_A,s_B)$, and
$z_0:=\xi_{A,i}-\xi_{B,j}+(\lambda/\pi)(\pi/2-\theta)$ for $\olW_\theta(s_A,s_B)$. For $r=N=1$ we have $\Psi\equiv0$, and \eqref{dssr} reduces to the identity $\bS(s_H^{(1)})=\mathcal R\,\bS(s_H^{(2)})$ for the saddle points.
The linear term in \eqref{psitorsion} cancels between the two sides, since the same six spectral parameters appear on both sides.

\section{Discrete spin solutions of the star-star relation}\label{sec:special}

\subsection{\texorpdfstring{$n=2$: comparison with Kels--Yamazaki}{n=2: comparison with Kels-Yamazaki}}\label{sec:n2}
We first compare our results with those of \cite{Kels:2017vbc} for $n=2$. In this case a spin is $s=((\zeta,-\zeta),(a,-a))$, i.e.\ $x=\zeta+\pi a/(rN)$
and $m\equiv-a$. This is the parametrization used in \cite{Kels:2017vbc}, where $\hat x=x+\pi m/(rN)=\zeta+\pi k/N$ with $k\in\ZZ$, and
the $\ZZ_{rN}$ spin is $a=rk-m$. The weights of \cite[(4.7)--(4.9), (4.13)]{Kels:2017vbc} are
\begin{gather}
  W^{\rm KY}_\theta\big((\phi_i,a_i),(\phi_j,a_j)\big):=w_{4,3}(\theta|\phi_i+\phi_j,a_i+a_j)\,w_{1,2}(\theta|\phi_i-\phi_j,a_i-a_j)\;,\\
  w_{j_1,j_2}(\theta|\phi,a):=\left(\frac{\ovth_{j_2}(\phi+\theta|r\tau')}{\ovth_{j_2}(\phi-\theta|r\tau')}\right)^{\frac{a}{rN}}
  \prod_{k=1}^{a}\frac{\ovth_{j_1}\big(\frac{\phi-\theta}{rN}+\frac{\pi(2k-1)}{2rN}\big|\frac{\tau'}{N}\big)}{\ovth_{j_1}\big(\frac{\phi+\theta}{rN}+\frac{\pi(2k-1)}{2rN}\big|\frac{\tau'}{N}\big)}\;,\\
  S^{\rm KY}(\phi,a)\propto\vartheta_4\Big(\frac{2\pi a}{rN}+\frac{2\phi}{rN}\,\Big|\,\frac{\tau'}{N}\Big)\;,
\end{gather}
with $\phi_i:=(\zeta_i+r\pi\tau/4)\pi/\lambda$ and $\theta^{\rm KY}:=(\pi/\lambda)\ii\alpha$, which coincides with our $\theta$.
On the other hand, for $n=2$ our edge weight can be written, using $\omega(-z|-a)=\omega(z|a)^{-1}$, as
\begin{equation}
  \bW_\theta(s_i,s_j)=\frac{\omega(\zeta_-+\frac{\lambda}{\pi}\theta\,|\,a_-)}{\omega(\zeta_--\frac{\lambda}{\pi}\theta\,|\,a_-)}\cdot\frac{\omega(\zeta_++\frac{\lambda}{\pi}\theta\,|\,a_+)}{\omega(\zeta_+-\frac{\lambda}{\pi}\theta\,|\,a_+)}\;,
  \qquad \zeta_\pm:=\zeta_i\pm\zeta_j,\ a_\pm:=a_i\pm a_j\;.
\end{equation}
By \eqref{omegamod}, the factor $\EXP^{\tilde P_a}$ cancels in each ratio, and we have $\omega(\zeta+\lambda\theta/\pi|a)/\omega(\zeta-\lambda\theta/\pi|a)=w_{1,2}(\theta|(\pi/\lambda)\zeta,a)$.
Since $\phi_i-\phi_j=(\pi/\lambda)\zeta_-$, the first ratio is the factor $w_{1,2}$. For the other factor we have $\phi_i+\phi_j=(\pi/\lambda)\zeta_++\pi r\tau'/2$.
In this case the half-period shifts $\ovth_1(x-\pi T/2|T)\propto\EXP^{\ii x}\ovth_4(x|T)$ and $\ovth_2(y-\pi r\tau'/2|r\tau')\propto\EXP^{\ii y}\ovth_3(y|r\tau')$
produce factors $\EXP^{\mp2\ii\theta a/(rN)}$, which cancel against each other. The vertex weight is given by \eqref{pairmod} with $\xi_{12}=2\zeta$ and
$a_{12}=2a$. Therefore, for a spin $s=((\zeta,-\zeta),(a,-a))$ with $0<\re(2N\zeta)<\pi$ and $\phi:=(\zeta+r\pi\tau/4)\pi/\lambda$, we find
\begin{equation}
  \bW_\theta(s_i,s_j)=W^{\rm KY}_\theta\big((\phi_i,a_i),(\phi_j,a_j)\big),\qquad \bS(s)=\mathcal N(2\zeta)\,S^{\rm KY}(\phi,a)\;,
\end{equation}
and we recover exactly the $\ZZ_{rN}$ model of \cite[(4.17)]{Kels:2017vbc}.

For general $n$, the weights \eqref{dW}--\eqref{dS} take the following form in the modular frame:
\begin{equation}\label{dWmod}
  \bW_\theta(s_A,s_B)\simeq\prod_{i,j=1}^n\mathfrak w\big(\phi_{A,i}-\phi_{B,j}+\theta\,\big|\,a_{A,i}-a_{B,j}\big),\qquad
  \bS(s)\simeq\prod_{i<j}\ovth_1\Big(\frac{\phi_{ij}+\pi a_{ij}}{M}\,\Big|\,T\Big)\,\EXP^{-\frac{\pi\ii}{M}\varepsilon_{ij}a_{ij}}\;,
\end{equation}
where $\varepsilon_{ij}:=\operatorname{sign}\re(N\xi_{ij})$ for $|\re N\xi_{ij}|<\pi$, and $\simeq$ denotes equality up to factors which either cancel term by term in \eqref{dssr}, or depend only on the continuous spins.
Factors of the second kind which belong to the internal vertex do not cancel in general, since $\vxi_H^{(1)}\ne\vxi_H^{(2)}$, and they change the
normalization $\mathcal R$. For example, if $\bS$ is replaced by $C(\vxi)\bS$, then $\mathcal R$ is replaced by $\mathcal R\,C(\vxi_H^{(1)})/C(\vxi_H^{(2)})$.
For $\varepsilon_{ij}=+1$ the vertex weight reduces to \eqref{pairmod}. Since $\ovth_1$ is odd, each factor of the vertex weight changes sign under
$(\phi_{ij},a_{ij},\varepsilon_{ij})\to-(\phi_{ij},a_{ij},\varepsilon_{ij})$, so that the vertex weight is antisymmetric under transpositions of the
components, like the Weyl denominator, of which it can be regarded as a theta-function analogue. The same holds for \eqref{dS}.

Besides the continuous parameters $\phi$, the discrete model thus carries the Weyl-chamber data $\varepsilon_{ij}$ of the saddle point.
Under simultaneous permutations of the components of $(\vxi_H,\va_H)$, which permute the $\varepsilon_{ij}$ accordingly, the summand
of \eqref{dssr} changes only by the sign of the permutation, which does not depend on the discrete spins and can be absorbed into $\mathcal R$.
This sign is compensated by the $a$-independent factors, since the full integrand of \eqref{ssr} is symmetric under permutations of the
components of $\sigma_H$. The chamber data are therefore defined only up to this action, and Weyl images of a saddle point give the same
model, so that they are not additional parameters. However, they must be kept track of when the saddle point is continued in the external
parameters. Indeed, $\varepsilon_{ij}$ changes when $\re(N\xi_{H,ij})$ crosses $\pi\ZZ$, which is the boundary of the region where \eqref{pairasym}
applies. For real configurations two components then coincide modulo $\pi/N$ (cf.\ Section~\ref{sec:KMreal}), while for complex saddle points
this is only a boundary of the asymptotic region.

\paragraph{Branches of the fractional powers.}
The weights \eqref{dW} contain the fractional powers $\ovth_4(N\xi|\hat\tau)^{a/(rN)}$ of \eqref{omegadef}, equivalently
$\ovth_2(\phi|MT)^{a/M}$ in \eqref{omegamod}. They are defined with the branch of $\log\ovth_4$ continuous from the real axis
(Section~\ref{sec:asymptotics}), and the modular form \eqref{dWmod} inherits this branch. This choice is important for the real symmetric configurations of Section~\ref{sec:KMreal}
and in the trigonometric limit of Section~\ref{sec:trig-general}.

\subsection{Dependence on \texorpdfstring{$r$}{r}}
For $n=2$ it was observed in \cite{Kels:2017vbc} that the $\ZZ_{rN}$ model obtained for general $r$ coincides with the $r=1$ model of
\cite{Bazhanov:2010kz} with $N\to rN$. Using the variables of Section~\ref{sec:modular}, this can be understood as follows, and it turns out that
the statement holds for all $n$. Namely, in the variables $(\phi,\theta,\va)$ the discrete star-star relation \eqref{dssr} depends on $(r,N,\tau)$
only through $(M,T)=(rN,r\tau/(1+rN\tau))$, up to factors which cancel or are absorbed into $\mathcal R$. To see this, note first that $\mathfrak w(\phi|a)$ in \eqref{omegamod}
contains only $a/M$, arguments divided by $M$, and the nomes $MT$ and $T$. The prefactor $\EXP^{\sum\tilde P}$ cancels in \eqref{dssr},
and its cubic coefficient $\pi\ii(rN-1)/(3r^2N)$ is the only place where $r$ appears separately. Next, the vertex weight in \eqref{dWmod}
depends only on $(M,T)$, and the saddle point equations \eqref{saddlemod} contain only $MT$. Finally, the summation range is $(\ZZ_M)^{n-1}$.
Consequently the model with parameters $(r,N,\tau)$ is equivalent to the model with parameters
\begin{equation}\label{tau1def}
  (1,\;rN,\;\tilde\tau),\qquad \tilde\tau:=\frac{r\tau}{1-(r-1)rN\tau}\;,
\end{equation}
where the spins are identified by equal $\phi$ and equal $\va$, i.e.\ $\xi$ is mapped to $(\tilde\lambda/\lambda)\xi$, where $\tilde\lambda$ is the torsion point \eqref{lamdef} of the model $(1,rN,\tilde\tau)$. In this identification the
Weyl-chamber data $\varepsilon_{ij}$ and the branches of the fractional powers are carried along with $\phi$. Since $N\xi=\phi/(r-MT)$, the
chambers defined by the signs of $\re(N\xi_{ij})$ for fixed $(\phi,T)$ may differ for different $r$; the equivalence is therefore an equivalence of
the models with the same chamber data and branches, obtained by analytic continuation, rather than of their definitions in terms of $\re(N\xi)$. The map $r\tau\mapsto\tilde\tau$ is the modular
transformation $\big(\begin{smallmatrix}1&0\\-(r-1)N&1\end{smallmatrix}\big)$, so that $\im\tilde\tau>0$. Note that this is a statement about the
discrete model only. The lens elliptic gamma function models with different $r$ are different, and so are the asymptotics of $\Phi(z,m)$ themselves
(cf.\ \eqref{psidef}). In particular, the $(\ZZ_r)^{n-1}$ model obtained for $N=1$ coincides with the one obtained from the $r=1$ model of
Bazhanov and Sergeev in the limit $\q\to\EXP^{\pi\ii/r}$. In the following we will work with $(M,T)$.

\subsection{Status of the discrete star-star relation}\label{sec:proofs}

The relation \eqref{dssr} was obtained from the quasi-classical expansion under the assumptions (A1)--(A3) (Section~\ref{sec:discrete}), and as far as we are aware a direct proof is not available in general;
already for $n=2$ the corresponding $\ZZ_N$ star-triangle relation of \cite[Sec.~4]{Bazhanov:2010kz} for general solutions of $Q4$ has not been
proven directly. The unconditional statement which we expect to hold is the following.
\begin{conjecture}\label{conj:dssr}
  Let $\vxi_H^{(1)}$, $\vxi_H^{(2)}$ be corresponding non-degenerate solutions of the saddle point equations \eqref{saddlemod} of the two stars.
  Then the $(\ZZ_{rN})^{n-1}$ star-star relation \eqref{dssr} holds for all discrete spins $\va_I,\va_J,\va_K,\va_L$, with the normalization
  $\mathcal R$ given by \eqref{Kdef}.
\end{conjecture}
In a number of special cases \eqref{dssr} reduces to relations which are established in the literature. These cases, and the status of
\eqref{dssr} in the other cases discussed in this paper, are listed below and summarized in Table~\ref{tab:status}, to which the
Introduction, Fig.~\ref{fig:models} and Section~\ref{sec:torsion} refer.

\begin{table}[t]\centering\small
  \renewcommand{\arraystretch}{1.2}
  \begin{tabular}{>{\raggedright\arraybackslash}p{4.6cm}>{\raggedright\arraybackslash}p{4.9cm}>{\raggedright\arraybackslash}p{4.6cm}}
    \hline
    model and parameters                                                                                                                                                               & status of \eqref{dssr}                                                                                                      & where                                                   \\
    \hline
    general background, any $n$, $r$, $N$                                                                                                                                              & conditional (Section~\ref{sec:discrete}); numerical checks                                                                  & Section~\ref{sec:discrete}, Appendix~\ref{app:numerics} \\
    $n=2$: Kashiwara--Miwa points \eqref{w0}, Fateev--Zamolodchikov point, chiral Potts curve                                                                                          & proven \cite{Kashiwara:1986tu,Hasegawa:1990du,Fateev:1982wi,Baxter:1987eq,Au-Yang:1989mk}                                   & Section~\ref{sec:n2}                                    \\
    $n\ge3$, elliptic Kashiwara--Miwa type, torsion points in the orbit of the Quano--Fujii point, $\gcd(n,M)=1$ (all primitive points if $\gcd(n,2M)=1$; the real points iff $n$ odd) & proven from the Yang--Baxter equation of \cite{Quano:1992wc}                                                                & Sections~\ref{sec:QFproof}, \ref{sec:torsion}           \\
    $n\ge3$, elliptic Kashiwara--Miwa type, real points with even $n$ (odd $\kappa$)                                                                                                   & conditional; numerical checks                                                                                               & Section~\ref{sec:torsion}                               \\
    $n\ge3$, elliptic Kashiwara--Miwa type, $\gcd(n,M)>1$                                                                                                                              & extension of the Quano--Fujii model; conditional if the point is in the orbit of a real point, otherwise numerical only     & Sections~\ref{sec:gcd}, \ref{sec:torsion}               \\
    $n\ge3$, elliptic Kashiwara--Miwa type, points in the orbit of neither the Quano--Fujii point nor a real point                                                                     & numerical only                                                                                                              & Section~\ref{sec:torsion}                               \\
    trigonometric limit, Z-invariant configurations \eqref{zinvspins}, any non-degenerate solution of the saddle point equations                                                       & proven, equal to the star-star relation of the $\mathfrak{sl}(n)$ chiral Potts model \cite{Bazhanov:1990qk,Bazhanov:1992jq} & Sections~\ref{sec:trig-general}, \ref{sec:BKMSproof}    \\
    trigonometric limit, external spins not Z-invariant                                                                                                                                & conditional; numerical checks                                                                                               & Section~\ref{sec:trig-general}                          \\
    \hline
  \end{tabular}
  \caption{Status of the discrete star-star relation \eqref{dssr} in the cases discussed in this paper. ``Conditional'' means that the relation
    follows from the quasi-classical expansion under the assumptions (A1)--(A3) of Section~\ref{sec:discrete}.}\label{tab:status}
\end{table}

For $n=2$ the star-star relation follows from the star-triangle relation \cite{Baxter:1997tn,Bazhanov:2011mz}, and our weights coincide with
those of \cite{Kels:2017vbc} (Section~\ref{sec:n2}), so that the proofs for the Kashiwara--Miwa, Fateev--Zamolodchikov and chiral Potts models
apply. For $n\ge3$ the two families of proven cases in Table~\ref{tab:status} are treated in Sections~\ref{sec:QFproof}, \ref{sec:torsion} and
in Sections~\ref{sec:trig-general}, \ref{sec:BKMSproof}, respectively.

\paragraph{General case.} For general solutions of the saddle point equations, a possible route to a proof is suggested by
\cite{Bazhanov:2008np,Bazhanov:1995zg}. There the $R$-matrix of the quantum discrete sine-Gordon model at a root of unity splits into
a classical part and a finite-dimensional part, and the Yang--Baxter equation yields a $\ZZ_N$ star-triangle relation on a classical background
solving the equations of motion. For a constant background this is the chiral Potts model. An elliptic $\mathfrak{sl}(n)$ analogue of this
construction, with cyclic representations of the Sklyanin and Belavin algebras as in \cite{Hasegawa:1990du,Quano:1992wc} on a background solving
\eqref{saddlemod}, would give a direct proof of \eqref{dssr}. This is the picture of the ``hybrid'' model anticipated in \cite[Sec.~5]{Bazhanov:2011mz}.
Intertwining operators of the Sklyanin algebra which are relevant for such a construction were studied in \cite{Zabrodin:2010qm}. We leave this
problem for future work.

\section{\texorpdfstring{Elliptic $\mathfrak{sl}(n)$ Kashiwara--Miwa type models}{Elliptic sl(n) Kashiwara-Miwa type models}}\label{sec:KMn}

In this section we study the simplest solutions of the saddle point equations, at which all spins carry the same continuous components.
They lead to elliptic $\mathfrak{sl}(n)$ generalizations of the Kashiwara--Miwa model, which, extended to general torsion points, include the $A^{(1)}_{n-1}$
broken $\ZZ_M$ model of Quano and Fujii.

\subsection{Real symmetric configurations}\label{sec:KMreal}

Let us look for solutions of the saddle point equations at which all spins, both external and internal, have the same continuous components
$\vxi_c=\vxi^0$, in analogy with the Kashiwara--Miwa points of \cite[Sec.~4.1]{Kels:2017vbc}. We require that such a configuration be admissible for the same background $\vxi^0$ uniformly for the full
interval of spectral variables, i.e.\ that for all $\theta\in(0,\pi/2)$ all edge arguments $z=\xi^0_i-\xi^0_j+(\lambda/\pi)\theta$ lie in the strip $|\im z|<(1/2)\im\lambda$, where \eqref{KYasym} holds and where
$\log\ovth_4(Nz|\hat\tau)$ is analytic and single-valued. We have $\im z=\im(\xi^0_{ij})+(\theta/\pi)\im\lambda$, and the spectral variables of the
$W$- and $\olW$-edges, $\theta$ and $\pi/2-\theta$, range over $(0,\pi/2)$. Requiring this for all $\theta$ in this range and for both orientations $ij$ and $ji$ of the edges gives
$-(1/2)\im\lambda<\pm\im\xi^0_{ij}\le0$, i.e.\ $\im\xi^0_{ij}=0$, so that the configuration has to be real. (For fixed rapidities,
or for spectral variables in a smaller interval, small non-zero $\im\xi^0_{ij}$ would also be admissible.)

We take
\begin{equation}\label{xi0}
  \xi^0_p:=-\rho_p\,\frac{\kappa\pi}{nN},\qquad \rho_p:=p-\frac{n+1}{2},\qquad p=1,\ldots,n,\qquad \kappa\in\ZZ\;,
\end{equation}
which is an $n$-torsion point of $\RR/(\pi/N)\ZZ$ in each difference. It is invariant, modulo $(\pi/N)\ZZ^n$ and an
overall translation, under the cyclic permutation of the components; this is what we mean by a $\ZZ_n$-symmetric configuration.

Let us first check that this configuration is a saddle point. In \eqref{saddleexplicit} the edge terms involve only $\mathcal L(z)=\log\ovth_4(Nz|\hat\tau)$, which is $\pi/N$-periodic.
Moreover, for every $p$ the multiset $\{\xi^0_p-\xi^0_i\}_i$ modulo $\pi/N$ is $\{s \kappa\pi/(nN)\}_{s\in\ZZ_n}$, and the edge terms do not
depend on $p$. The vertex term is $\sum_{q\neq p}\mathcal C'(\xi^0_{pq})=-(1/N)\sum_{q\neq p}(2N\tilde\xi^0_{pq}-\pi)$. If $\gcd(\kappa,n)=1$, then
$N\tilde\xi^0_{pq}$ runs over $\{s\pi/n\}_{s=1}^{n-1}$ for every $p$, and the sum vanishes. Thus, for $\gcd(\kappa,n)=1$, $\vxi_H=\vxi^0$ solves
the saddle point equations of both stars for arbitrary spectral parameters, with $\mu$ independent of $p$. For $\gcd(\kappa,n)>1$, on the other
hand, two components coincide modulo $\pi/N$. Such a configuration lies on a wall of the Weyl chambers, where \eqref{pairasym} does not
apply, and we exclude it.

Next consider the fractional powers in \eqref{dW}, which combine into
$\exp\big[\sum_{i,j}(a_{A,i}-a_{B,j})\,\mathcal L(\xi^0_{ij}+(\lambda/\pi)\theta)/(rN)\big]$.
As we have just seen, $\sum_j\mathcal L(\xi^0_{ij}+(\lambda/\pi)\theta)$ does not
depend on $i$, and the same is true for the sum over $i$. Since $\sum_ia_{A,i}=\sum_ja_{B,j}=0$, the fractional powers cancel identically.
Because $\mathcal L$ is single-valued in the strip, there is no branch ambiguity, and the weights become pure products of theta functions.

Finally, the Weyl-chamber signs of the vertex weight are $\varepsilon_{ij}=\operatorname{sign}\re(N\xi^0_{ij})$, where $N\xi^0_{ij}$ is
reduced to $(-\pi,\pi)$ as in \eqref{dS}. For $\kappa=1$ they are $\varepsilon_{ij}=+1$ for $i<j$, and for general $\kappa$ they are determined by
$\kappa(j-i)$ mod $n$. Below we write the resulting weights for $\kappa=1$. For general $\kappa$ with $\gcd(\kappa,n)=1$ we have verified
numerically that the resulting relation coincides with \eqref{KMnssr} at the corresponding $\omega_0$ (Appendix~\ref{app:numerics}).

In the modular frame \eqref{taupdef}, using $\pi/N=r\lambda-\pi r\tau$, which is mapped to $-\pi(MT-r)$ by $\xi\mapsto\phi=(\pi/\lambda)\xi$, the configuration \eqref{xi0} becomes
\begin{equation}\label{phi0}
  \phi^0_p=\rho_p\,\omega_0,\qquad \omega_0:=\frac{\kappa\pi(MT-r)}{n}\;,
\end{equation}
where $\omega_0$ is a primitive $n$-torsion point of $\CC/(\pi\ZZ+\pi MT\ZZ)$. For $\kappa=1$ this is the point $\nu\equiv r$ of the family
\begin{equation}\label{w0}
  \omega_0=\frac{\pi(MT-\nu)}{n},\qquad \nu\in\ZZ\;,
\end{equation}
with $\varepsilon_{ij}=+1$, while $\kappa=-1$ gives the mirror images $\omega_0\to-\omega_0$, $\varepsilon\to-\varepsilon$. The other values of $\nu$ in
\eqref{w0} correspond to complex configurations $\vxi^0$, for which some of the edge arguments leave the strip. Nevertheless, we find numerically that the
weights obtained below satisfy the star-star relation for all $\nu$; on the orbit of the Quano--Fujii point this is proven in
Sections~\ref{sec:QFproof} and~\ref{sec:torsion}. For $n=2$ we have $\phi_{12}=-\omega_0$, and \eqref{w0} with $\nu=0,1$ are
exactly the Kashiwara--Miwa points $\nu^{\rm KY}=0,1/2$ of \cite[Sec.~4.1]{Kels:2017vbc}, recalling that $\phi^{\rm KY}_1+\phi^{\rm KY}_2=\phi_{12}+\pi MT/2$.

In the modular frame the cancellation of the fractional powers reads
\begin{equation}
  \prod_{i,j=1}^n\ovth_2\big((i-j)\omega_0+\theta\,\big|\,MT\big)^{\frac{a_{A,i}-a_{B,j}}{M}}\simeq\EXP^{-\frac{2\ii\theta}{M}\langle\rho,\va_A+\va_B\rangle}\;,
\end{equation}
with $\rho=(\rho_1,\ldots,\rho_n)$ as in \eqref{xi0}. Since $\sum_ia_i=0$, only $\langle\rho,\va\rangle$ enters, and a constant shift of $\rho$ is
immaterial. The phase comes from the Gaussian factors of the modular transformation, which give
\begin{equation}
  \sum_{i,j=1}^n(a_{A,i}-a_{B,j})\big(\xi^0_{ij}+\tfrac{\lambda}{\pi}\theta\big)^2=n\big(\langle\va_A,(\vxi^0)^2\rangle-\langle\va_B,(\vxi^0)^2\rangle\big)
  +\tfrac{2n\lambda}{\pi}\theta\big(\langle\va_A,\vxi^0\rangle+\langle\va_B,\vxi^0\rangle\big)\;;
\end{equation}
the first term is a gauge factor and the second one is the phase. In the modular frame the same phase follows from the
quasi-periodicity of $\ovth_2$ with the branch of Section~\ref{sec:special}; with the principal branch of each factor, additional $\ZZ_M$
characters appear, which spoil \eqref{dssr}. Around the internal vertex the four spectral shifts add up to $\pi$, which produces $\EXP^{-2\pi\ii\langle\rho,\va_H\rangle/M}$.
Since $\sum_{i<j}a_{ij}=-2\langle\rho,\va\rangle$, this factor exactly cancels the phases in the vertex weight \eqref{dWmod}, while on the external
spins these factors cancel between the two sides. We thus obtain the weights
\begin{align}
  \mathsf W_\theta(\va,\boldsymbol b) & :=\prod_{i,j=1}^n\ \prod_{k=0}^{a_i-b_j-1}\ovth_1\Big(\frac{(i-j)\omega_0+\theta+(k+\frac12)\pi}{M}\,\Big|\,T\Big)^{-1}\;,
  \qquad \overline{\mathsf W}_\theta:=\mathsf W_{\frac\pi2-\theta}\;,\label{KMnW}                                                                                  \\
  \mathsf S(\va)                      & :=\prod_{1\le i<j\le n}\ovth_1\Big(\frac{\pi(a_i-a_j)-(j-i)\omega_0}{M}\,\Big|\,T\Big)\;,\label{KMnS}
\end{align}
for $\va\in Q/MQ\cong(\ZZ_M)^{n-1}$. Here the right-hand sides are evaluated on lifts $\va,\boldsymbol b\in Q$. As in Section~\ref{sec:discrete}, the individual weights depend on the lifts, but the
factors produced by a change of lift by $MQ$ cancel in the star-star relation (at the internal vertex between $\mathsf S$ and the four
edges, and on the external spins between the two sides), so that the relation below holds for any lifts of the external spins, used on both
sides, and with the sum over $\va_H$ over any set of representatives. For definiteness we take $0\le a_1,\ldots,a_{n-1}<M$,
$a_n=-\sum_{i<n}a_i$. With this convention the weights \eqref{KMnW}--\eqref{KMnS} satisfy the star-star relation \eqref{dssr} with $\mathcal R=1$, namely
\begin{multline}\label{KMnssr}
  \mathsf W_{\theta_{v'v}}(\va_L,\va_K)\mathsf W_{\theta_{u'u}}(\va_L,\va_J)\\
  \times\sum_{\va_H}\mathsf S(\va_H)\overline{\mathsf W}_{\theta_{uv}}(\va_K,\va_H)\overline{\mathsf W}_{\theta_{u'v'}}(\va_J,\va_H)
  \mathsf W_{\theta_{u'v}}(\va_H,\va_I)\mathsf W_{\theta_{uv'}}(\va_H,\va_L)\\
  =\mathsf W_{\theta_{v'v}}(\va_J,\va_I)\mathsf W_{\theta_{u'u}}(\va_K,\va_I)\\
  \times\sum_{\va_H}\mathsf S(\va_H)\overline{\mathsf W}_{\theta_{uv}}(\va_H,\va_J)\overline{\mathsf W}_{\theta_{u'v'}}(\va_H,\va_K)
  \mathsf W_{\theta_{u'v}}(\va_L,\va_H)\mathsf W_{\theta_{uv'}}(\va_I,\va_H)\;.
\end{multline}
For $n=2$ these are the Kashiwara--Miwa weights \cite[(4.18)--(4.20)]{Kels:2017vbc}. Indeed, the factors with $i=j$ give $w_1(\theta|0,a-b)$, where $a$ and $b$ are the discrete spins at the two ends of the edge, and the
factors with $i\neq j$ give $w_4$, since $\omega_0/M=\pi(MT-\nu)/(2M)$ contains the half period $\pi T/2$. For $n\ge3$ the weights involve theta functions at
arguments shifted by multiples of $\omega_0/M$, where $\omega_0$ is an $n$-torsion point of $\CC/(\pi\ZZ+\pi MT\ZZ)$, as is common for $\mathfrak{sl}(n)$ models. The weights
\eqref{KMnW}--\eqref{KMnS} thus provide a discrete-spin elliptic $\mathfrak{sl}(n)$ analogue of the Kashiwara--Miwa model which satisfies the star-star relation.

\subsection{Identification with the Quano--Fujii model}\label{sec:QF}
In \cite{Quano:1992wc}, Quano and Fujii constructed an $A^{(1)}_{n-1}$ generalization of the broken $\ZZ_N$ model of
\cite{Kashiwara:1986tu,Hasegawa:1990du} from vertex--face intertwiners of the Belavin $\ZZ_n\otimes\ZZ_n$ $R$-matrix. In the following, the label QF
refers to their paper.\footnote{Equation, theorem and section numbers of \cite{Quano:1992wc} refer to the arXiv version
hep-th/9210063v2 (v1 has been withdrawn).} With $h(z):=\vartheta_1(\pi z|\tau_{\rm QF})$, the crossing-type parameter $w:=1/N_{\rm QF}$ and the background $a_0=0$, the
Boltzmann weights are $\overline W^{\rm QF}_{ab}(u_{\rm QF}):=\sqrt{\Delta(a)\Delta(b)}\,\bar w_{ab}(u_{\rm QF})$. Here
$\Delta(a):=\prod_{l<m}h(\bar a_{lm}w)$, and $\bar w_{ab}$ is defined by the recursion \cite[Thm.~4.4]{Quano:1992wc}
\begin{equation}\label{QFrec}
  \frac{\bar w_{a+\bar\epsilon_i,b+\bar\epsilon_k}(u_{\rm QF})}{\bar w_{ab}(u_{\rm QF})}
  =\frac{\prod_{m<i}h\big(u_{\rm QF}+(\bar a_m-\bar b_k)w\big)\prod_{m>i}h\big(-u_{\rm QF}-(\bar a_m-\bar b_k)w\big)}
  {\prod_{m<k}h\big(-u_{\rm QF}-(\bar a_i-\bar b_m+1)w\big)\prod_{m>k}h\big(u_{\rm QF}+(\bar a_i-\bar b_m+1)w\big)}\;,
\end{equation}
with $\bar a_j:=(a+\rho_{\rm QF}|\epsilon_j)+(2j+1-n)/(2n)$, $j=0,\ldots,n-1$, where $\rho_{\rm QF}$ is the Weyl vector of \cite{Quano:1992wc}.
Here, as in \cite{Quano:1992wc}, $\epsilon_0,\ldots,\epsilon_{n-1}$ is an orthonormal basis of $\RR^n$ with respect to the bilinear form
$(\,\cdot\,|\,\cdot\,)$, $\bar\epsilon_j:=\epsilon_j-(1/n)\sum_{k=0}^{n-1}\epsilon_k$, and the indices $i,k,m$ in \eqref{QFrec} run over
$0,\ldots,n-1$.
The local state $a$ is a weight of level $N_{\rm QF}$, i.e.\ a vector of $n$ integers modulo
simultaneous shifts, and a step $a\to a+\bar\epsilon_i$ raises its $i$-th component by one.
We now take $\omega_0=\pi/n$ in \eqref{KMnW}--\eqref{KMnS} and set
\begin{equation}\label{QFdict}
  \tau_{\rm QF}=T,\qquad N_{\rm QF}=M,\qquad (a|\epsilon_j)=a_{j+1}+j\ \ (j=0,\ldots,n-1),\qquad u_{\rm QF}=\frac{\theta-\frac\pi2}{\pi M}\;,
\end{equation}
where $a$ is the local state of \cite{Quano:1992wc} and $\va=(a_1,\ldots,a_n)$ is our discrete spin, and similarly for $b$ and $\boldsymbol b$.
In \eqref{KMnW} a simultaneous step $a_{i}\to a_{i}+1$, $b_k\to b_k+1$ adds the $n-1$ factors with $(i,j)$, $j\neq k$, to the denominator and
removes the $n-1$ factors with $(m,k)$, $m\neq i$, while the factor $(i,k)$ is unchanged. This is the same structure as in \eqref{QFrec}. With
\eqref{QFdict} one has $\pi\big(u_{\rm QF}+(\bar a_{i-1}-\bar b_{j-1})w\big)=\big((i-j)\pi/n+\theta-\pi/2+\pi(a_i-b_j)\big)/M$, since the shifts by
$\rho_{\rm QF}$ and by $j$ in \eqref{QFdict} cancel. Therefore the factors agree one by one, up to signs coming from $h(-z)=-h(z)$.

Keeping track of these signs, one finds that the ratio of $\mathsf W_\theta(\va,\boldsymbol b)$ and $\bar w_{ab}(u_{\rm QF})$ changes by
$(-1)^{n-i+k-1}$ under the step $a_i\to a_i+1$, $b_k\to b_k+1$. The factor $(-1)^{n-1}$ does not contribute, since the net number of steps
between two pairs of spins with $\sum_ia_i=\sum_ib_i=0$ vanishes, and $(-1)^{k-i}$ is reproduced by $\varsigma(\va)/\varsigma(\boldsymbol b)$.
We thus obtain
\begin{equation}\label{QFid}
  \mathsf W_\theta(\va,\boldsymbol b)=\varrho_\theta\,\frac{\varsigma(\va)}{\varsigma(\boldsymbol b)}\,\bar w_{ab}(u_{\rm QF}),\qquad \mathsf S(\va)=\varrho\,\Delta(a),\qquad
  \varsigma(\va):=(-1)^{\sum_{i=1}^n i\,a_i}\;,
\end{equation}
with constants $\varrho_\theta$, $\varrho$ independent of the spins, where $\va,\boldsymbol b\in\ZZ^n$ are any lifts with $\sum_ia_i=\sum_ib_i=0$ and
$\mathsf W_\theta$ is evaluated by \eqref{KMnW} on the same lifts. The relation for $\mathsf S$ follows in the same way, since $\mathsf S(\va)$
and $\Delta(a)$ consist of the same factors. Thus, at $\omega_0=\pi/n$, the local Boltzmann weights of the elliptic $\mathfrak{sl}(n)$ Kashiwara--Miwa type model coincide with those of
the Quano--Fujii model on corresponding lifts. For $\gcd(n,M)=1$ the state spaces are identified as well, and the two models coincide
(see Section~\ref{sec:gcd} for $\gcd(n,M)>1$). We also have
$\overline W^{\rm QF}\propto\sqrt{\mathsf S(\va)\mathsf S(\boldsymbol b)}\,\mathsf W_\theta$, in the same way as the weights of \cite{Kels:2017toi} are related
to \eqref{BW} (Section~\ref{sec:model}). In \cite{Quano:1992wc} the Yang--Baxter equation is proven for the face weights
$\overline W^{\rm QF}_{ab}\overline W^{\rm QF}_{bc}\overline W^{\rm QF}_{cd}/\overline W^{\rm QF}_{ad}$, in which the signs $\varsigma$ in \eqref{QFid} cancel.
The relation \eqref{KMnssr}, which is the star-star relation with $\bar w$ on the edges and $\Delta$ on the internal vertex, is a further identity
for the Quano--Fujii weights, and in Section~\ref{sec:QFproof} we will derive it from the Yang--Baxter equation of \cite{Quano:1992wc}.

\subsection{Star-star relation from the Yang--Baxter equation}\label{sec:QFproof}
At $\omega_0=\pi/n$ the weights \eqref{KMnW}--\eqref{KMnS} are those of the
Quano--Fujii model (Section~\ref{sec:QF}). We now show that, for $\gcd(n,M)=1$, the star-star relation \eqref{KMnssr} is a consequence of the
Yang--Baxter equation proven in \cite{Quano:1992wc}; the case $\gcd(n,M)>1$ is discussed in Section~\ref{sec:gcd}. In \cite{Quano:1992wc} the weights $\overline W^{\rm QF}_{ab}(z,z'):=\overline W^{\rm QF}_{ab}((z-z')/n)$ are
assembled into the ``box'' weights
\begin{equation}\label{QFS}
  \langle a,d|\mathcal S(z_1,z_1';z_2,z_2')|b,c\rangle:=\frac{\overline W^{\rm QF}_{ab}(z_1,z_2)\,\overline W^{\rm QF}_{bc}(z_2,z_1')\,\overline W^{\rm QF}_{cd}(z_1',z_2')}
  {\overline W^{\rm QF}_{ad}(z_1,z_2')}\;,
\end{equation}
which define an operator on $\mathcal H\otimes\mathcal H$, where $\mathcal H$ is spanned by the local states (Fig.~\ref{fig:box}). It is shown in
\cite{Quano:1992wc} that $\mathcal S$ intertwines the factorized $L$-operators $L(z_1,z_1')\otimes L(z_2,z_2')$ and $L(z_2,z_2')\otimes L(z_1,z_1')$.
Since the commutant of $L\otimes L\otimes L$ is trivial for generic parameters, it follows that $\mathcal S$ satisfies the Yang--Baxter equation
\cite[Thm.~4.6]{Quano:1992wc}
\begin{equation}\label{QFYBE}
  \begin{aligned}
     & \mathcal S_{23}(z_1,z_1';z_2,z_2')\,\mathcal S_{12}(z_1,z_1';z_3,z_3')\,\mathcal S_{23}(z_2,z_2';z_3,z_3')           \\
     & \qquad=\mathcal S_{12}(z_2,z_2';z_3,z_3')\,\mathcal S_{23}(z_1,z_1';z_3,z_3')\,\mathcal S_{12}(z_1,z_1';z_2,z_2')\;,
  \end{aligned}
\end{equation}
with $\mathcal S_{12}:=\mathcal S\otimes1$, $\mathcal S_{23}:=1\otimes\mathcal S$. Moreover, $\overline W^{\rm QF}_{ab}(z,z)\propto\delta_{ab}$ with an $a$-independent
constant \cite[Sec.~4.2]{Quano:1992wc}.

The box \eqref{QFS} contains the four edge types of our star. With
\begin{equation}\label{QFlines}
  \begin{gathered}
    (z_1,z_1')=(\zeta_{v'},\zeta_v),\qquad (z_2,z_2')=(\zeta_{u'},\zeta_u)\;,\\
    \zeta_u:=\frac{n\theta_u}{\pi M},\quad \zeta_{u'}:=\frac{n(\theta_{u'}-\frac\pi2)}{\pi M},\quad \zeta_v:=\frac{n\theta_v}{\pi M},\quad \zeta_{v'}:=\frac{n(\theta_{v'}-\frac\pi2)}{\pi M}\;,
  \end{gathered}
\end{equation}
and \eqref{QFdict}, the edges $(z_1,z_2)$, $(z_2,z_1')$, $(z_1',z_2')$, $(z_1,z_2')$ of \eqref{QFS} are $\overline{\mathsf W}_{\theta_{u'v'}}$, $\mathsf W_{\theta_{u'v}}$,
$\overline{\mathsf W}_{\theta_{uv}}$ and $\mathsf W_{\theta_{v'u}}$, while the edges $(z_1,z_1')$ and $(z_2,z_2')$ are $\mathsf W_{\theta_{v'v}}$ and $\mathsf W_{\theta_{u'u}}$. The remaining
edge of the star is given by $\mathsf W_{\theta_{uv'}}(a,b)=\mathsf W_{\theta_{v'u}}(b,a)^{-1}$, as follows from the inversion relation
$\mathsf W_\theta(\va,\boldsymbol b)\mathsf W_{-\theta}(\boldsymbol b,\va)=1$, which in turn follows from \eqref{KMnW} and $\sum_ia_i=\sum_ib_i=0$.
We insert \eqref{QFid}, i.e.\ $\mathsf W\propto\bar w=\overline W^{\rm QF}/\sqrt{\Delta\Delta}$ and $\mathsf S\propto\Delta$, into \eqref{KMnssr}.
The factors $\Delta(\va_H)$ cancel exactly, since three edges in the numerator and one in the denominator meet at $H$. The external factors
$\Delta(\va_I)^{-1/2}\Delta(\va_J)^{-1}\Delta(\va_K)^{-1}\Delta(\va_L)^{-1/2}$ are the same on both sides, and the signs $\varsigma$ and the constants
$\varrho_\theta$ cancel as well. With the matrices $A:=\overline W^{\rm QF}(z_1,z_2)$, $B:=\overline W^{\rm QF}(z_2,z_1')$, $C:=\overline W^{\rm QF}(z_1',z_2')$,
$D:=\overline W^{\rm QF}(z_1,z_2')$, $U:=\overline W^{\rm QF}(z_1,z_1')$, $V:=\overline W^{\rm QF}(z_2,z_2')$ on $\mathcal H$, the relation \eqref{KMnssr} becomes
\begin{equation}\label{SSRQF}
  U_{LK}\,V_{LJ}\sum_{H}\frac{A_{JH}\,B_{HI}\,C_{KH}}{D_{LH}}=U_{JI}\,V_{KI}\sum_{H}\frac{C_{HJ}\,A_{HK}\,B_{LH}}{D_{HI}}\;.
\end{equation}
Here the condition $\gcd(n,M)=1$ is used for the first time: it ensures that the dictionary \eqref{QFdict} is a bijection between our
discrete spins $Q/MQ$ and the local states of \cite{Quano:1992wc} (Appendix~\ref{app:gcd}),\footnote{The local states of
  \cite{Quano:1992wc} with the cyclic condition form $\ZZ^n/(\ZZ(1,\ldots,1)+M\ZZ^n)$, and \eqref{QFdict} maps $\va\in Q$ to the class of
  $\va+(0,1,\ldots,n-1)$. Both sets have $M^{n-1}$ elements. If $\va\in Q$ is mapped to the class of $(0,1,\ldots,n-1)$, then
  $\va=c(1,\ldots,1)+M\boldsymbol w$ with $c\in\ZZ$, $\boldsymbol w\in\ZZ^n$, and $\sum_ia_i=0$ gives $nc\equiv0$ (mod $M$). For
  $\gcd(n,M)=1$ this implies $M\,|\,c$ and hence $\va\in MQ$, so that the map is a bijection; for $\gcd(n,M)>1$ the choice $c=M/\gcd(n,M)$ gives
  the kernel \eqref{kerQF}.} so that the sum over $\va_H$ in \eqref{KMnssr}
becomes the sum over the basis of $\mathcal H$ in \eqref{SSRQF}. It is used a second time below, since the Yang--Baxter equation
\eqref{QFYBE} is proven in \cite{Quano:1992wc} under the assumption that $n$ and $N_{\rm QF}=M$ are coprime.
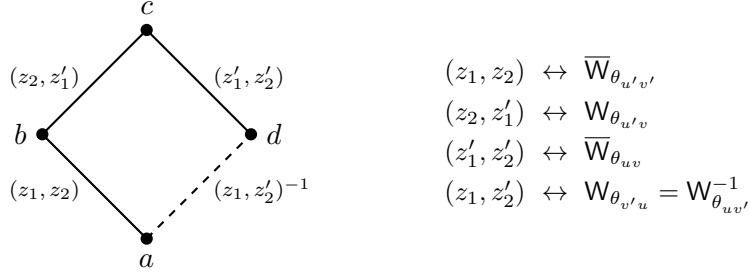
\begin{figure}[t]\centering
  \begin{tikzpicture}[scale=1.15]
    \coordinate (a) at (0,-1.2); \coordinate (b) at (-1.2,0); \coordinate (c) at (0,1.2); \coordinate (d) at (1.2,0);
    \draw[thick] (a)--(b)--(c)--(d); \draw[thick,dashed] (a)--(d);
    \foreach \p/\l/\pos in {a/a/below,b/b/left,c/c/above,d/d/right} \fill (\p) circle (2pt) node[\pos=2pt] {$\l$};
    \node[font=\scriptsize,left] at (-0.65,-0.65) {$(z_1,z_2)$};
    \node[font=\scriptsize,left] at (-0.65,0.65) {$(z_2,z_1')$};
    \node[font=\scriptsize,right] at (0.65,0.65) {$(z_1',z_2')$};
    \node[font=\scriptsize,right] at (0.65,-0.65) {$(z_1,z_2')^{-1}$};
    \begin{scope}[xshift=5.2cm]
      \node[font=\small,align=left] at (0,0) {$(z_1,z_2)\ \leftrightarrow\ \overline{\mathsf W}_{\theta_{u'v'}}$\\[2pt]
        $(z_2,z_1')\ \leftrightarrow\ \mathsf W_{\theta_{u'v}}$\\[2pt] $(z_1',z_2')\ \leftrightarrow\ \overline{\mathsf W}_{\theta_{uv}}$\\[2pt]
        $(z_1,z_2')\ \leftrightarrow\ \mathsf W_{\theta_{v'u}}=\mathsf W_{\theta_{uv'}}^{-1}$};
    \end{scope}
  \end{tikzpicture}
  \caption{The box weight \eqref{QFS} of \cite{Quano:1992wc}: three edges in the numerator and one (dashed) in the denominator, labeled by the pairs
    of rapidities of the edge weights $\overline W^{\rm QF}$. With \eqref{QFlines} they are the four edge types of the stars in \eqref{KMnssr} (right).}\label{fig:box}
\end{figure}

\begin{figure}[t]\centering
\begin{tikzpicture}[
    x=1cm,y=1cm,
    v/.style={circle,fill=black,inner sep=1.3pt},
    hv/.style={circle,fill=blue!70!black,inner sep=1.9pt},
    e/.style={thick},
    den/.style={thick,dashed},
    del/.style={very thick,red!80!black},
    idl/.style={densely dotted,gray},
    st/.style={very thick,blue!70!black},
    stden/.style={very thick,dashed,blue!70!black},
    lab/.style={font=\scriptsize,inner sep=1pt},
    slab/.style={font=\scriptsize,inner sep=1pt,fill=white,text=blue!70!black},
    vl/.style={font=\scriptsize,inner sep=1.5pt},
    idlab/.style={font=\scriptsize,inner sep=1pt,text=gray!60!black},
  ]
  \begin{scope}
    \node[font=\small] at (1.3,4.25) {LHS of \eqref{QFYBE}};
    \coordinate (d1) at (0,0); \coordinate (d2) at (1.3,0); \coordinate (d3) at (2.6,0);
    \coordinate (x1) at (0,1.1); \coordinate (x2) at (1.3,1.1); \coordinate (x3) at (2.6,1.1);
    \coordinate (y1) at (0,2.2); \coordinate (y2) at (1.3,2.2); \coordinate (y3) at (2.6,2.2);
    \coordinate (a1) at (0,3.3); \coordinate (a2) at (1.3,3.3); \coordinate (a3) at (2.6,3.3);
    \draw[e] (x2)--node[lab,left]{$B$}(d2); \draw[e] (d2)--node[lab,below]{$C$}(d3);
    \draw[del] (d3)--node[lab,right]{$\delta$}(x3); \draw[den] (x2)--node[lab,above]{$V$}(x3);
    \draw[idl] (d1)--(x1);
    \draw[e] (y1)--node[lab,left]{$U$}(x1); \draw[del] (x1)--node[lab,below]{$\delta$}(x2);
    \draw[e] (x2)--node[lab,right]{$C$}(y2); \draw[den] (y1)--node[lab,above]{$D$}(y2);
    \draw[idl] (x3)--(y3);
    \draw[e] (a2)--node[lab,left]{$A$}(y2); \draw[e] (y2)--node[lab,below]{$B$}(y3);
    \draw[e] (y3)--node[lab,right]{$C$}(a3); \draw[den] (a2)--node[lab,above]{$D$}(a3);
    \draw[idl] (y1)--(a1);
    \foreach \p in {d1,d2,d3,x1,x2,x3,y1,y3,a1,a2,a3} \node[v] at (\p) {};
    \node[hv] at (y2) {};
    \node[idlab,left=2pt] at (x1) {$=d_1$}; \node[idlab,below right=1pt] at (x2) {$=d_1$};
    \node[idlab,right=2pt] at (x3) {$=d_3$}; \node[idlab,right=2pt] at (y3) {$=d_3$};
    \node[idlab,left=2pt] at (y1) {$=a_1$};
    \node[vl,below] at (d1) {$d_1$}; \node[vl,below] at (d2) {$d_2$}; \node[vl,below] at (d3) {$d_3$};
    \node[vl,above] at (a1) {$a_1$}; \node[vl,above left] at (a2) {$a_2$}; \node[vl,above] at (a3) {$a_3$};
    \node[vl,below right] at (y2) {$y$};
    \node[font=\scriptsize,anchor=west,align=left] at (3.1,0.55) {$\mathcal S(2;3)$};
    \node[font=\scriptsize,anchor=east,align=right] at (-0.25,1.65) {$\mathcal S(1;3)$};
    \node[font=\scriptsize,anchor=west,align=left] at (3.1,2.75) {$\mathcal S(1;2)$};
  \end{scope}
  \draw[-{Stealth},thick] (4.4,1.65)--node[above,font=\scriptsize,align=center]{contract\\ $\delta$-edges}(5.8,1.65);
  \begin{scope}[xshift=7.9cm,yshift=1.65cm]
    \node[font=\small] at (0,2.6) {LHS of \eqref{SSRQF}};
    \coordinate (I) at (-1.25,1.25); \coordinate (J) at (1.25,1.25); \coordinate (K) at (-1.25,-1.25); \coordinate (L) at (1.25,-1.25);
    \coordinate (H) at (0,0);
    \draw[st] (H)--node[slab,pos=0.6,above right=1pt]{$A$}(J); \draw[st] (H)--node[slab,pos=0.6,above left=1pt]{$B$}(I);
    \draw[st] (H)--node[slab,pos=0.6,below left=1pt]{$C$}(K); \draw[stden] (H)--node[slab,pos=0.6,below right=1pt]{$D$}(L);
    \draw[e] (L)--node[lab,below]{$U$}(K); \draw[den] (K)--node[lab,left]{$V$}(I);
    \foreach \p in {I,J,K,L} \node[v] at (\p) {};
    \node[hv] at (H) {};
    \node[vl,above left] at (I) {$I=d_3$}; \node[vl,above right] at (J) {$J=a_2$};
    \node[vl,below left] at (K) {$K=d_1$}; \node[vl,below right] at (L) {$L=a_1$};
    \node[vl,left=5pt] at (H) {$H=y$};
  \end{scope}
  \begin{scope}[yshift=-6.0cm]
    \node[font=\small] at (1.3,4.25) {RHS of \eqref{QFYBE}};
    \coordinate (d1) at (0,0); \coordinate (d2) at (1.3,0); \coordinate (d3) at (2.6,0);
    \coordinate (x1) at (0,1.1); \coordinate (x2) at (1.3,1.1); \coordinate (x3) at (2.6,1.1);
    \coordinate (y1) at (0,2.2); \coordinate (y2) at (1.3,2.2); \coordinate (y3) at (2.6,2.2);
    \coordinate (a1) at (0,3.3); \coordinate (a2) at (1.3,3.3); \coordinate (a3) at (2.6,3.3);
    \draw[e] (x1)--node[lab,left]{$A$}(d1); \draw[e] (d1)--node[lab,below]{$B$}(d2);
    \draw[e] (d2)--node[lab,right]{$C$}(x2); \draw[den] (x1)--node[lab,above]{$D$}(x2);
    \draw[idl] (d3)--(x3);
    \draw[e] (y2)--node[lab,left]{$U$}(x2); \draw[del] (x2)--node[lab,below]{$\delta$}(x3);
    \draw[e] (x3)--node[lab,right]{$C$}(y3); \draw[den] (y2)--node[lab,above]{$D$}(y3);
    \draw[idl] (x1)--(y1);
    \draw[e] (a1)--node[lab,left]{$B$}(y1); \draw[e] (y1)--node[lab,below]{$C$}(y2);
    \draw[del] (y2)--node[lab,right]{$\delta$}(a2); \draw[den] (a1)--node[lab,above]{$V$}(a2);
    \draw[idl] (y3)--(a3);
    \foreach \p in {d1,d2,d3,x2,x3,y1,y2,y3,a1,a2,a3} \node[v] at (\p) {};
    \node[hv] at (x1) {};
    \node[idlab,left=2pt] at (y1) {$=p$}; \node[idlab,above right=1pt] at (x2) {$=d_3$};
    \node[idlab,right=2pt] at (x3) {$=d_3$}; \node[idlab,below right=1pt] at (y2) {$=a_2$};
    \node[idlab,right=2pt] at (y3) {$=a_3$};
    \node[vl,below] at (d1) {$d_1$}; \node[vl,below] at (d2) {$d_2$}; \node[vl,below] at (d3) {$d_3$};
    \node[vl,above] at (a1) {$a_1$}; \node[vl,above] at (a2) {$a_2$}; \node[vl,above] at (a3) {$a_3$};
    \node[vl,below left] at (x1) {$p$};
    \node[font=\scriptsize,anchor=east,align=right] at (-0.45,0.55) {$\mathcal S(1;2)$};
    \node[font=\scriptsize,anchor=west,align=left] at (3.05,1.65) {$\mathcal S(1;3)$};
    \node[font=\scriptsize,anchor=east,align=right] at (-0.25,2.75) {$\mathcal S(2;3)$};
  \end{scope}
  \draw[-{Stealth},thick] (4.4,-4.35)--node[above,font=\scriptsize,align=center]{contract\\ $\delta$-edges}(5.8,-4.35);
  \begin{scope}[xshift=7.9cm,yshift=-4.35cm]
    \node[font=\small] at (0,2.6) {RHS of \eqref{SSRQF}};
    \coordinate (I) at (-1.25,1.25); \coordinate (J) at (1.25,1.25); \coordinate (K) at (-1.25,-1.25); \coordinate (L) at (1.25,-1.25);
    \coordinate (H) at (0,0);
    \draw[st] (H)--node[slab,pos=0.6,above right=1pt]{$C$}(J); \draw[stden] (H)--node[slab,pos=0.6,above left=1pt]{$D$}(I);
    \draw[st] (H)--node[slab,pos=0.6,below left=1pt]{$A$}(K); \draw[st] (H)--node[slab,pos=0.6,below right=1pt]{$B$}(L);
    \draw[e] (J)--node[lab,above]{$U$}(I); \draw[den] (L)--node[lab,right]{$V$}(J);
    \foreach \p in {I,J,K,L} \node[v] at (\p) {};
    \node[hv] at (H) {};
    \node[vl,above left] at (I) {$I=d_3$}; \node[vl,above right] at (J) {$J=a_2$};
    \node[vl,below left] at (K) {$K=d_1$}; \node[vl,below right] at (L) {$L=a_1$};
    \node[vl,left=5pt] at (H) {$H=p$};
  \end{scope}
\end{tikzpicture}
\caption{The star-star relation as a specialization of the Yang--Baxter equation \eqref{QFYBE}. Left: the two sides of \eqref{QFYBE}, read from
  bottom ($d_1d_2d_3$) to top ($a_1a_2a_3$); each box $\langle a,d|\mathcal S|b,c\rangle$ of \eqref{QFS} is drawn with its input states $b,c$ at the
  bottom and its output states $a,d$ at the top, and dotted lines are states which are not acted upon. Solid edges are factors
  $A,B,C,U$ of the numerator, dashed edges factors $D,V$ of the denominator ($\mathcal S(i;j):=\mathcal S(z_i,z_i';z_j,z_j')$, and $A,\ldots,V$ are as in
  \eqref{SSRQF}). At $z_3=z_1'$, $z_3'=z_2'$ the edges $\overline W^{\rm QF}(z_3,z_1')$ and $\overline W^{\rm QF}(z_2',z_3')$ become
  $\overline W^{\rm QF}(0)\propto\delta$ (red) and identify the states at their ends; the dotted lines identify the states at their ends as well. The gray labels give the external
  state with which each intermediate vertex is thereby identified; for example, on the left-hand side $d_3$ is joined to $y$ by the edge $B$
  of $\mathcal S(1;2)$, through the red edge and the dotted line above it. The only remaining internal state ($y$, resp.\ $p$, blue) is the
  center of a four-edge star (right). The factors $B_{d_1d_2}C_{d_2d_3}$ and $C_{d_3a_3}/D_{a_2a_3}$ involving $d_2$ and $a_3$ appear on both sides
  and are not shown. The two sides are thus $U_{LK}V_{KI}^{-1}$ and $U_{JI}V_{LJ}^{-1}$ times the two stars, and multiplying by $V_{LJ}V_{KI}$ gives \eqref{SSRQF}.}\label{fig:ybe-ssr}
\end{figure}
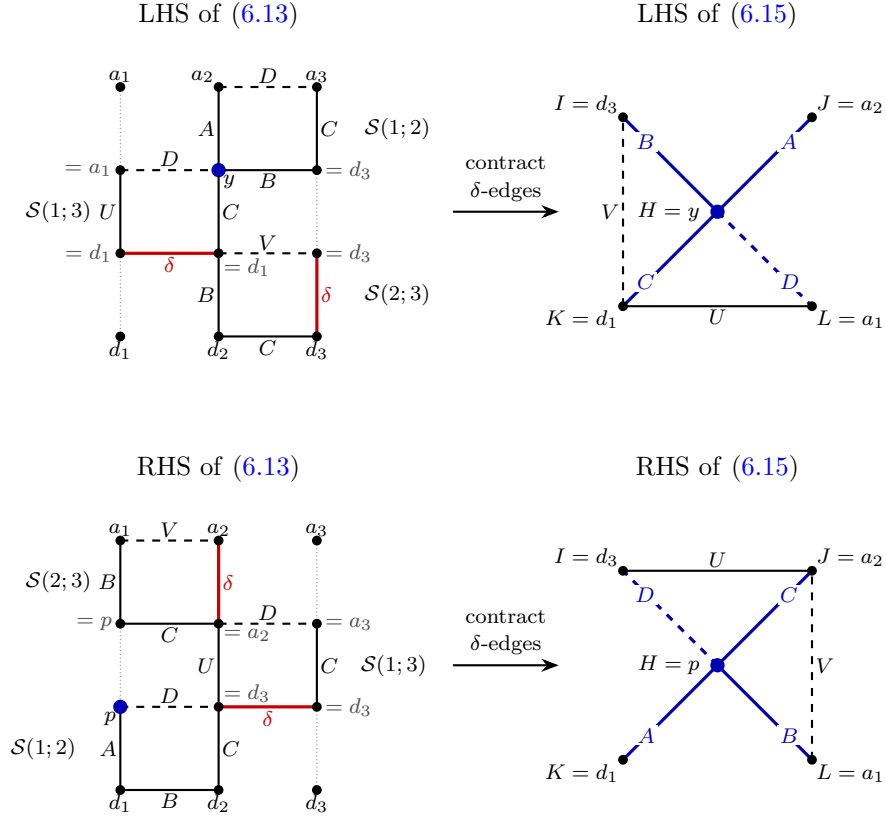

The reduction is summarized in Fig.~\ref{fig:ybe-ssr}. Take the matrix element $\langle a_1a_2a_3|\cdots|d_1d_2d_3\rangle$ of \eqref{QFYBE} at
\begin{equation}
  z_3=z_1',\qquad z_3'=z_2'\;.
\end{equation}
Then $\mathcal S(z_2,z_2';z_3,z_3')$ contains the edge $\overline W^{\rm QF}(z_2',z_3')=\overline W^{\rm QF}(0)$, and $\mathcal S(z_1,z_1';z_3,z_3')$ contains
the edge $\overline W^{\rm QF}(z_3,z_1')=\overline W^{\rm QF}(0)$. Both are proportional to $\delta$'s. On the left-hand side of \eqref{QFYBE} they fix two
of the three summed states, and the remaining sum is over the state $y$ at which $\mathcal S_{23}(z_1,z_1';z_2,z_2')$ and $\mathcal S_{12}(z_1,z_1';z_3,z_3')$
are joined. Omitting the constant $\overline W^{\rm QF}_{aa}(0)^2$, which is common to both sides, we find
\begin{equation}
  \langle a_1a_2a_3|\,\text{LHS}\,|d_1d_2d_3\rangle=U_{a_1d_1}\,\frac{C_{d_3a_3}}{D_{a_2a_3}}\,\frac{B_{d_1d_2}C_{d_2d_3}}{V_{d_1d_3}}
  \sum_{y}\frac{A_{a_2y}B_{yd_3}C_{d_1y}}{D_{a_1y}}\;,
\end{equation}
and similarly on the right-hand side the sum is over the state $p$ joining $\mathcal S_{12}(z_1,z_1';z_2,z_2')$ and $\mathcal S_{23}(z_2,z_2';z_3,z_3')$:
\begin{equation}
  \langle a_1a_2a_3|\,\text{RHS}\,|d_1d_2d_3\rangle=U_{a_2d_3}\,\frac{C_{d_3a_3}}{D_{a_2a_3}}\,\frac{B_{d_1d_2}C_{d_2d_3}}{V_{a_1a_2}}
  \sum_{p}\frac{C_{pa_2}A_{pd_1}B_{a_1p}}{D_{pd_3}}\;.
\end{equation}
The factors involving $a_3$ and $d_2$ are common to both sides and cancel. With $(I,J,K,L)=(d_3,a_2,d_1,a_1)$ the two sums are the two stars of
\eqref{SSRQF}, and the remaining factors are $U_{LK}/V_{KI}$ and $U_{JI}/V_{LJ}$. Thus \eqref{QFYBE} at $z_3=z_1'$, $z_3'=z_2'$ is
\eqref{SSRQF}. In other words, the star-star relation \eqref{KMnssr} at $\omega_0=\pi/n$ is the Yang--Baxter equation of \cite{Quano:1992wc}, with the
third pair of rapidities $(z_3,z_3')=(\zeta_v,\zeta_u)$ taken from the first two. The Yang--Baxter equation is proven in \cite{Quano:1992wc} for
generic parameters, and it extends to this point by analyticity. We have also checked \eqref{QFYBE}, \eqref{SSRQF} and this reduction numerically
(Appendix~\ref{app:numerics}).

Note that the specialization is regular, since the denominators of the boxes involved are $D=\overline W^{\rm QF}(z_1,z_2')$ and
$V=\overline W^{\rm QF}(z_2,z_2')$, which do not vanish for generic rapidities. The star-star relation at other torsion points is obtained from
the one at $\omega_0=\pi/n$ by the symmetries of \eqref{KMnssr} discussed in Section~\ref{sec:torsion} and Appendix~\ref{app:torsion}, and not by repeating the above argument
with modular images of the Quano--Fujii weights, which in general do not map \eqref{KMnssr} to itself.

The implication from the star-star relation to the Yang--Baxter equation of the box (vertex-type) weights is standard
\cite{Baxter:1986phd,Bazhanov:1992jqa,Baxter:1997tn}. The derivation above goes in the opposite direction: for the Quano--Fujii model the
Yang--Baxter equation of the box weights is established independently, from the uniqueness of the intertwiner of the factorized $L$-operators
\cite{Quano:1992wc} (as for the box $R$-matrix of \cite[Sec.~3]{Bazhanov:2011mz}), and the star-star relation follows from it by a
specialization of the rapidities.

\subsection{Extension to \texorpdfstring{$\gcd(n,M)>1$}{gcd(n,M)>1}}\label{sec:gcd}
The argument of Section~\ref{sec:QFproof} requires $\gcd(n,M)=1$: \cite{Quano:1992wc} assume that $n$ and $N_{\rm QF}$ are coprime
(below their (4.17)), so that their Yang--Baxter equation is available only in this case. For $d:=\gcd(n,M)>1$ the situation is as follows;
the details are given in Appendix~\ref{app:gcd}. Our discrete spins are the elements of $Q/MQ$,
while the local states of \cite{Quano:1992wc} are classes in $\overline P/M\overline P$, where $\overline P$ is the weight lattice of
$\mathfrak{sl}(n)$. The dictionary \eqref{QFdict} induces a map $Q/MQ\to\overline P/M\overline P$, whose kernel is $\ZZ_d$, generated by the class of
\begin{equation}\label{kerQF}
  \va_{\rm ker}:=\tfrac{M}{d}\,(1-n,1,\ldots,1)\;.
\end{equation}
The recursion \eqref{QFrec} assigns different weights to $\va$ and $\va+\va_{\rm ker}$, so that for $d>1$ the cyclic condition of
\cite{Quano:1992wc}, which identifies them, is incompatible with the recursion; this is where their assumption enters. The weights
\eqref{KMnW}--\eqref{KMnS} on $(\ZZ_M)^{n-1}\cong Q/MQ$, on the other hand, are well-defined for all $n$ and $M$. For $d>1$ the model at
$\omega_0=\pi/n$ is therefore an extension of the Quano--Fujii model to the state space $Q/MQ$, and \eqref{KMnssr} is a star-star relation for
this extension, which we have verified numerically (Appendix~\ref{app:numerics}). A proof along the lines of \cite{Quano:1992wc} would require
the Yang--Baxter equation for this extension; a candidate for the corresponding box operator, and the open points in its construction, are
described in Appendix~\ref{app:gcd}.

\subsection{General torsion points}\label{sec:torsion}
In \eqref{KMnW}--\eqref{KMnS} the shift $\omega_0\to\omega_0+\pi$ is equivalent to the relabeling $a_i\to a_i+i$ of all spins, so that only $\nu$
modulo $n$ for odd $n$, or modulo $2n$ for even $n$, matters in \eqref{w0} (Appendix~\ref{app:torsion}). The real symmetric configurations \eqref{xi0} give the points
$\omega_0=\pi(\kappa MT-\kappa r)/n$ with $\gcd(\kappa,n)=1$, for all factorizations $M=rN$, but not, for example, the Quano--Fujii point
$\omega_0=\pi/n$. Nevertheless, we find numerically (Appendix~\ref{app:numerics}) that \eqref{KMnssr} holds with $\mathcal R=1$ at every primitive
$n$-torsion point $\omega_0=\pi(\kappa MT-\nu)/n$, $\gcd(\kappa,\nu,n)=1$, of the lattice $\pi\ZZ+\pi MT\ZZ$. At the non-primitive points it fails in all non-degenerate cases we have checked; in some cases, such as $\omega_0=0$ for $(n,M)=(3,2)$, $\mathsf S$ vanishes identically and both sides are zero.

This is explained in Appendix~\ref{app:torsion} by symmetries of \eqref{KMnssr}: (a) the relabelings, the shift $T\to T+1$ and the mirror symmetry
$\omega_0\to-\omega_0$; (b) the shift $\omega_0\to\omega_0+\pi MT$; and (c) the modular transformations
$g=\big(\begin{smallmatrix}\alpha&\beta\\\gamma&\delta\end{smallmatrix}\big)\in\mathrm{SL}(2,\ZZ)$ with $\alpha\equiv1$ and $\gamma\equiv0$ (mod $2M$),
acting as $(T,\omega_0,\theta)\to(gT,\omega_0/\mathcal J,\theta/\mathcal J)$ with $\mathcal J:=\gamma T+\delta$, which preserve the step of the discrete
spins and the crossing parameter up to periods (Appendix~\ref{app:modular}). The orbits of the primitive torsion points under these symmetries are
determined in Proposition~\ref{prop:orbits}. They act transitively if and only if $\gcd(n,2M)=1$, and the orbit of the Quano--Fujii point consists of the
points with $\nu\equiv\pm1$ modulo $\gcd(n,M)$ (modulo $2\gcd(n,M)$ and with even $\kappa$ for even $n$). As a result, \eqref{KMnssr} is proven on the
orbit of the Quano--Fujii point for $\gcd(n,M)=1$, it follows from \eqref{dssr} under the assumptions (A1)--(A3) on the orbits of the real points,
and at the remaining points, which exist only for $\gcd(n,M)>1$ (for example $\omega_0=\pi(2MT-3)/6$ for $(n,M)=(6,3)$), it is supported by numerical
evidence. Appendix~\ref{app:torsion} also discusses which torsion points are modular images of the Quano--Fujii model at the level of the Boltzmann
weights, see \eqref{QFcond}, and the role of the parity of $\kappa$ for even $n$.

\subsection{Reality and positivity}\label{sec:reality}
For a statistical-mechanical interpretation one would like the Boltzmann weights to be real and positive. Let the nome be real, $T\in\ii\RR_{>0}$,
and take $\omega_0=\pi MT/n$, i.e.\ $\nu=0$, so that $\bar\omega_0=-\omega_0$. Since then $\overline{\ovth_1(z|T)}=\ovth_1(\bar z|T)$, the weights
\eqref{KMnW}--\eqref{KMnS} satisfy
\begin{equation}\label{KMconj}
  \overline{\mathsf W_\theta(\va,\boldsymbol b)}=\mathsf W_{\bar\theta}(-\boldsymbol b,-\va),\qquad \overline{\mathsf S(\va)}=(-1)^{\frac{n(n-1)}2}\,\mathsf S(-\va)\;.
\end{equation}
In other words, complex conjugation reverses both the spins and the orientation of the edges. We apply this to the IRF weight given by the
left-hand side of \eqref{KMnssr},
\begin{multline}\label{IRFV}
  \mathbb V_{\bm u\bm v}(\va_I,\va_J,\va_K,\va_L):=\mathsf W_{\theta_{v'v}}(\va_L,\va_K)\mathsf W_{\theta_{u'u}}(\va_L,\va_J)\\
  \times\sum_{\va_H}\mathsf S(\va_H)\overline{\mathsf W}_{\theta_{uv}}(\va_K,\va_H)\overline{\mathsf W}_{\theta_{u'v'}}(\va_J,\va_H)
  \mathsf W_{\theta_{u'v}}(\va_H,\va_I)\mathsf W_{\theta_{uv'}}(\va_H,\va_L)\;,
\end{multline}
which by \eqref{KMnssr} is also equal to the right-hand side of \eqref{KMnssr}. By \eqref{KMconj}, $\overline{\mathbb V_{\bm u\bm v}}$ is turned into the right-hand side
of \eqref{KMnssr} with the spins $(-\va_L,-\va_K,-\va_J,-\va_I)$ and conjugate rapidities, up to the sign $(-1)^{n(n-1)/2}$ coming from
$\mathsf S$. This sign is removed by the normalization $\widehat{\mathsf S}(\va):=\ii^{\,n(n-1)/2}\,\mathsf S(\va)$, which satisfies
$\overline{\widehat{\mathsf S}(\va)}=\widehat{\mathsf S}(-\va)$ and does not affect \eqref{KMnssr}, since $\mathsf S$ appears once on each side.
Denoting by $\widehat{\mathbb V}_{\bm u\bm v}$ the IRF weight \eqref{IRFV} with $\mathsf S$ replaced by $\widehat{\mathsf S}$ and using the star-star relation, we
obtain for real rapidities
\begin{equation}
  \overline{\widehat{\mathbb V}_{\bm u\bm v}(\va_I,\va_J,\va_K,\va_L)}=\widehat{\mathbb V}_{\bm u\bm v}(-\va_L,-\va_K,-\va_J,-\va_I)\;.
\end{equation}
Here both sides are evaluated on lifts, the lifts on the right-hand side being the negatives of those on the left-hand side, and the
internal sum does not depend on the choice of representatives (Section~\ref{sec:KMreal}). Thus, locally, complex conjugation is equivalent
to the reversal of all spins combined with the rotation of the face by $\pi$.

Consequently the partition function of a finite domain with
real rapidities and fixed boundary spins is real whenever the domain, the rapidity configuration and the boundary spins (with their lifts) are
invariant under the rotation by $\pi$ combined with the reversal of the spins.\footnote{With $\mathsf S$ instead of $\widehat{\mathsf S}$ the
  partition function is real up to the phase $\ii^{-n(n-1)/2}$ per internal vertex; for $n=2$ and vanishing external spins, for example,
  $\mathbb V_{\bm u\bm v}$ itself is purely imaginary. For a general domain this also requires that the summand does not depend on the
  representatives of the summed spins, which we have verified only for the internal spin of the star.} This is the analogue of the reality property of the continuous-spin
model discussed in \cite[Sec.~3]{Bazhanov:2011mz}, where the two forms of the IRF weight are exchanged by complex conjugation.

For $n=2$ the gauge-invariant ratios
$\mathsf W_\theta(\va,\boldsymbol b)\mathsf W_\theta(\boldsymbol 0,\boldsymbol 0)/\big(\mathsf W_\theta(\va,\boldsymbol 0)\mathsf W_\theta(\boldsymbol 0,\boldsymbol b)\big)$
are real and positive for $0<\theta<\pi/2$, in agreement with the positivity of the Kashiwara--Miwa model in its physical regime. For $n\ge3$ the
weights are not symmetric and these ratios are not real. A regime with positive weights is not known for $n\ge3$; for the continuous-spin model,
whose two-spin weights are not positive either, the IRF weights are positive in a neighborhood of $\p=\q=0$ \cite[Sec.~3]{Bazhanov:2011mz}.

\section{\texorpdfstring{Trigonometric limit and the $(\ZZ_M)^{n-1}$ chiral Potts model}{Trigonometric limit and the chiral Potts model}}\label{sec:trigsec}

\subsection{An \texorpdfstring{$\mathfrak{sl}(n)$}{sl(n)} Fateev--Zamolodchikov type model}\label{sec:trig}

In this section we consider the trigonometric limit $T\to\ii\infty$. Let us start with the elliptic $\mathfrak{sl}(n)$ Kashiwara--Miwa type model
\eqref{KMnW}--\eqref{KMnS}. For $0<|y|<1$ we have
\begin{equation}
  \ovth_1(x+y\pi T\,|\,T)\;\longrightarrow\;\pm\ii\,\EXP^{\frac{\pi\ii T}{4}}\,\EXP^{\mp\pi\ii yT}\,\EXP^{\mp\ii x}\qquad(y\gtrless0)\;,
\end{equation}
while for $y=0$ we have $\ovth_1(x|T)\to2\EXP^{\pi\ii T/4}\sin x$. The factors with $i\neq j$ in \eqref{KMnW} have $y=(i-j)/n$, and
therefore become pure exponentials. Apart from the phase $\EXP^{\pm(\pi\ii/(2M))(a_i-b_j)^2}$, each of them depends on $\theta$ through a
factor $\EXP^{\pm\ii\theta(a_i-b_j)/M}$. After the product over $i\neq j$ these factors combine into a phase linear in $\theta\langle\rho,\va+\boldsymbol b\rangle$,
which cancels in \eqref{KMnssr} in the same way as the phase discussed in Section~\ref{sec:KMreal}. The divergent constants combine into one-vertex factors, which cancel in \eqref{KMnssr}.
Similarly, every factor of $\mathsf S$ becomes an exponential. Collecting all the terms, we find that the dependence on $\nu$ disappears, and we
obtain the weights, for $\va,\boldsymbol b$ with $\sum_ia_i=\sum_ib_i=0$,
\begin{equation}\label{trigW}
  \mathsf W^{\rm trig}_\theta(\va,\boldsymbol b):=\EXP^{-\frac{\pi\ii}{M}\sum_{i>j}(a_ib_j-a_jb_i)}\;
  \prod_{i=1}^n\prod_{k=0}^{a_i-b_i-1}\frac{1}{\sin\frac{\theta+(k+\frac12)\pi}{M}}\;,
\end{equation}
\begin{equation*}
  \overline{\mathsf W}^{\rm trig}_\theta:=\mathsf W^{\rm trig}_{\frac\pi2-\theta},\qquad \mathsf S^{\rm trig}:=1\;,
\end{equation*}
which satisfy \eqref{KMnssr} with $\mathcal R=1$. The chiral phase is essential for $n\ge3$. For $n=2$ it vanishes identically, and
\eqref{trigW} reduces to the Fateev--Zamolodchikov $\ZZ_M$ model \cite{Fateev:1982wi}. Both signs of the phase give solutions, since they are
related by $a_i\to a_{n+1-i}$. If we write $1/\sin y\propto\EXP^{\ii y}/(1-\EXP^{2\ii y})$, the weight \eqref{trigW} takes the form, up to
factors which cancel,
\begin{equation}\label{trigBKMS}
  \mathsf W^{\rm trig}_\theta(\va,\boldsymbol b)\simeq\varpi^{\,\mathsf Q(\va,\boldsymbol b)}\prod_{i=1}^n\prod_{k=0}^{a_i-b_i-1}\frac{1}{1-t\,\varpi^{k+\frac12}}\;,
\end{equation}
\begin{equation*}
  \varpi:=\EXP^{\frac{2\pi\ii}{M}},\quad t:=\EXP^{\frac{2\ii\theta}{M}},\quad \mathsf Q(\va,\boldsymbol b):=\sum_{i=1}^{n-1}a_i\sum_{j=1}^i(a_j-b_j)\;.
\end{equation*}
This is the form of the weights of the $(\ZZ_M)^{n-1}$ chiral Potts model \cite{Bazhanov:1990qk} when all the $n$ ratios
$z_i^+(q)z_i^-(p)/z_i^+(p)z_i^-(q)$ are equal, namely at the degenerate point $K_{ij}\propto\mathbf 1$ of its spectral curve, which will be
discussed in Section~\ref{sec:trig-general}. This model is therefore an $\mathfrak{sl}(n)$ analogue of the
Fateev--Zamolodchikov point of the chiral Potts model, and generalizes to $\mathfrak{sl}(n)$ the reduction from the Kashiwara--Miwa model to the
Fateev--Zamolodchikov model in \cite[Sec.~4.3]{Kels:2017vbc}.

\subsection{General trigonometric limit and the spectral curve}\label{sec:trig-general}

\paragraph{The trigonometric model.}
For $n=2$ it was found in \cite[Sec.~4.4]{Kels:2017vbc} that the chiral Potts model arises from non-symmetric saddle points in the trigonometric
limit. In our variables the corresponding regime is
\begin{equation}\label{trigansatz}
  \phi_{c,p}=\psi_{c,p}+\rho_p\,\frac{\pi MT}{n},\qquad \rho_p=p-\tfrac{n+1}{2},\qquad \psi_{c,p}\ \text{finite},\qquad T\to\ii\infty\;,
\end{equation}
which is a finite deformation of the Kashiwara--Miwa point $\nu=0$. For $n=2$ this means that $\phi^{\rm KY}=\psi_1$ is finite.
Using the asymptotics above, we find that the factors of \eqref{dWmod} with $i\neq j$ lose all dependence on $\psi$ and $\theta$ and contribute
only the chiral phase, and that the vertex weight becomes independent of the discrete spins. We obtain
\begin{equation}\label{triggenW}
  \bW_\theta(s_A,s_B)\simeq\EXP^{-\frac{\pi\ii}{M}\sum_{i>j}(a_{A,i}a_{B,j}-a_{A,j}a_{B,i})}\prod_{i=1}^n
  \mathsf w\big(\psi_{A,i}-\psi_{B,i}+\theta\,\big|\,a_{A,i}-a_{B,i}\big)\;,
\end{equation}
\begin{equation}
  \mathsf w(x|d):=\frac{(2\cos x)^{\frac{d}{M}}}{\prod_{k=0}^{d-1}2\sin\frac{x+(k+\frac12)\pi}{M}}\;,
\end{equation}
and $\bS\simeq1$. Each diagonal component is thus a copy of the trigonometric weight for $n=2$ \cite[(4.34)]{Kels:2017vbc}. In the same way, the
saddle point equations \eqref{saddlemod} become, modulo $2\pi\ii$,
\begin{equation}\label{triggensaddle}
  \begin{aligned}
    G_p:= & \sum_{W}\log2\cos(\psi_{H,p}-\psi_{X,p}+\theta)-\sum_{\olW}\log2\cos(\psi_{X,p}-\psi_{H,p}+\tfrac\pi2-\theta) \\
          & +\ii\Big(\sum_{\olW}-\sum_{W}\Big)\Big(\sum_{i>p}\psi_{X,i}-\sum_{i<p}\psi_{X,i}\Big)=\mu\;.
  \end{aligned}
\end{equation}
These are $n$ copies of the three-leg form of $Q3_{(\delta=0)}$ \cite{Adler:2002abs}, which are coupled through the Lagrange multiplier and through the
``cumulative'' terms in the second line.

\paragraph{Rapidity lines.}
Following \cite[(4.49)]{Kels:2017vbc}, we attach to each rapidity line $\ell$ the variables $R_{\ell,i}:=\EXP^{2\ii\varphi_{\ell,i}}$, $i=1,\ldots,n$, and put
\begin{equation}\label{linedict}
  \alpha_\ell:=\frac1n\sum_{i=1}^n\varphi_{\ell,i}+\frac\pi2\,e_\ell,\qquad \beta_{\ell,i}:=\varphi_{\ell,i}-\frac1n\sum_{k}\varphi_{\ell,k}\;,
\end{equation}
where $e_\ell:=0$ for the solid lines $u,v$ and $e_\ell:=1$ for the dashed lines $u',v'$ of Fig.~\ref{ssfig}. The spectral variables are
differences of the $\alpha$'s, for example $\theta_{uv}=\alpha_u-\alpha_v$. The continuous spins are placed in the Z-invariant way, namely for an edge
crossed by the lines $\ell_1,\ell_2$ with $\theta=\alpha_{\ell_1}-\alpha_{\ell_2}$ we set $\psi_A-\psi_B=\beta_{\ell_1}-\beta_{\ell_2}$. For the two
stars of \eqref{dssr} this gives
\begin{equation}\label{zinvspins}
  \begin{gathered}
    \psi_I=\beta_v-\beta_{u'},\quad \psi_J=\beta_{v'}-\beta_{u'},\quad \psi_K=\beta_v-\beta_u,\quad \psi_L=\beta_{v'}-\beta_u\;,\\
    \psi_H^{(1)}=0,\qquad\psi_H^{(2)}=\psi_J-\beta_u+\beta_v\;,
  \end{gathered}
\end{equation}
and all spins satisfy $\sum_i\psi_i=0$. The $i$-th component of every edge weight then depends only on $\pm R_{\ell_1,i}/R_{\ell_2,i}$.
The shift of the dashed lines by $\pi/2$ accounts for the difference between the $W$- and $\olW$-type edges, since a $W$-edge is crossed by
one solid and one dashed line, while a $\olW$-edge is crossed by two lines of the same type.

\paragraph{The saddle point equations as a cross-ratio condition.}
We insert \eqref{zinvspins} and \eqref{linedict} into \eqref{triggensaddle}. All the terms linear in the $\varphi_{\ell,i}$, i.e.\ the cumulative
terms and the terms coming from $\alpha_\ell$, cancel up to spin-independent multiples of $2\pi\ii$. Only the diagonal factors survive, and for the
first star we find (for the second star one obtains the inverse)
\begin{equation}
  \EXP^{G_p}\propto\frac{\cos(\varphi_{u',p}-\varphi_{v,p}+\frac\pi2)\cos(\varphi_{u,p}-\varphi_{v',p}-\frac\pi2)}{\cos(\varphi_{v,p}-\varphi_{u,p}+\frac\pi2)\cos(\varphi_{v',p}-\varphi_{u',p}+\frac\pi2)}\;.
\end{equation}
Since $2\ii\sin(\varphi_a-\varphi_b)=(R_a-R_b)/\sqrt{R_aR_b}$, the square roots cancel and $\EXP^{G_p}\propto\chi_p$,
the cross-ratio
\begin{equation}\label{crcond}
  \chi_p:=-\frac{(R_{u',p}-R_{v,p})(R_{u,p}-R_{v',p})}{(R_{v,p}-R_{u,p})(R_{v',p}-R_{u',p})}\;.
\end{equation}
Thus, for arbitrary line variables, the saddle point equations of both stars are equivalent to the single condition that the cross-ratio
$\chi_p$ of the four rapidity lines is the same in every component $p=1,\ldots,n$. The Lagrange multiplier is then given by $\EXP^\mu\propto\chi_p$.

\paragraph{The spectral curve.}
Because of the M\"obius invariance of the cross-ratio, this condition is solved if all rapidities lie on a common curve
\begin{equation}\label{bkmscurve}
  R_{\ell,i}=g_i(s_\ell),\qquad g_i=\begin{pmatrix}\mathsf a_i&\mathsf b_i\\\mathsf c_i&\mathsf d_i\end{pmatrix}\in \mathrm{PGL}(2,\CC)\ \text{fixed},\qquad i=1,\ldots,n\;.
\end{equation}
If we write $R_{\ell,i}=(z^+_i(\ell)/z^-_i(\ell))^M$, the curve \eqref{bkmscurve} is the spectral curve
$(z_i^{+M},z_i^{-M})^{\!\top}\propto K_{ij}(z_j^{+M},z_j^{-M})^{\!\top}$, $K_{ij}:=g_ig_j^{-1}$, of the $(\ZZ_M)^{n-1}$ chiral Potts model \cite{Bazhanov:1990qk}.

For a single star-star relation the converse is elementary, and requires no assumption on the lines. If the four lines $u,u',v,v'$ of \eqref{dssr}
have equal cross-ratios $\chi_p$ for all $p$, then for generic lines there are M\"obius maps $g_i$ with $R_{\ell,i}=g_i(R_{\ell,1})$ for these four
lines, namely the maps sending $(R_{u,1},R_{u',1},R_{v,1})$ to $(R_{u,i},R_{u',i},R_{v,i})$, since a fourth point of $\mathbb P^1$ is determined by its
cross-ratio with three given points. The derivation of the weights \eqref{BKMSform} below only uses the curve through the four lines involved. Consequently, for
Z-invariant configurations, i.e.\ external spins of the form \eqref{zinvspins} with the saddle points given there, every non-degenerate solution of the saddle point
equations gives the $\mathfrak{sl}(n)$ chiral Potts weights on a curve through the four rapidities, and \eqref{dssr} holds in the trigonometric limit by
the results of Section~\ref{sec:BKMSproof}. Here non-degenerate means that the four rapidities are distinct in every component, so that the
cross-ratios and the weights are defined; degenerate configurations are reached by continuity on the chosen branches. A curve which is common to all lines is needed only for the Z-invariance of the whole
lattice. For external spins which are not of the form \eqref{zinvspins}, i.e.\ $\psi_I-\psi_J\neq\psi_K-\psi_L$, the relation \eqref{dssr} in the
trigonometric limit is a generalization of the chiral Potts star-star relation in which the weights \eqref{triggenW} depend on the continuous spins
$\psi_A$ of the vertices rather than on rapidities attached to lines, with the internal spin determined by \eqref{triggensaddle}; in this case we have
verified \eqref{dssr} only numerically (Appendix~\ref{app:numerics}).

We remark that a global version of this converse is also expected. Suppose that the lines of each of the four types lie on holomorphic
curves, and that \eqref{crcond} is independent of $p$ for all choices of the four lines. A comparison of the zeros, the poles and the points
where $\chi_p=-1$ as functions of one rapidity suggests that all rapidities then lie on a single curve of the form \eqref{bkmscurve}, although
we have not worked out a complete argument. For $n=2$ and generic coefficients this is, after rescalings of $X:=R_{\ell,1}$ and
$Y:=R_{\ell,2}$, the chiral Potts curve $X+Y=k(1+XY)$, while degenerate maps such as $Y=X+c$ give degenerate curves, such as the point
$K_{ij}\propto\mathbf 1$ of Section~\ref{sec:trig}. Such a statement would generalize the ``homogeneous'' and ``inhomogeneous'' uniqueness
statements of \cite[Secs.~4.4.3--4.4.4]{Kels:2017vbc}.

\paragraph{The Boltzmann weights on the curve.}
Let $z^\pm_i(\ell)$ be the $M$-th roots of $z_i^+(\ell)^M=\mathsf a_is_\ell+\mathsf b_i$, $z_i^-(\ell)^M=\mathsf c_is_\ell+\mathsf d_i$. On the curve we have
\begin{equation}
  R_{\ell_1,i}-R_{\ell_2,i}=\frac{\det g_i\,(s_{\ell_1}-s_{\ell_2})}{(\mathsf c_is_{\ell_1}+\mathsf d_i)(\mathsf c_is_{\ell_2}+\mathsf d_i)}\;.
\end{equation}
Since $\sum_id_i=0$ for $d_i:=a_{A,i}-a_{B,i}$, the factors $(s_{\ell_1}-s_{\ell_2})^{d_i/M}$ cancel between the components. The fractional
powers in \eqref{triggenW} therefore factorize into functions on the lines, $\upsilon_i(\ell):=(\mathsf c_is_\ell+\mathsf d_i)^{-1/M}=1/z^-_i(\ell)$, which
are multi-valued in the $s$-plane but single-valued on the covering curve with coordinates $z_i^\pm$,
which are the analogues of $\mu_p$ in the chiral Potts model. The factors $\prod_i[\upsilon_i(\ell_1)\upsilon_i(\ell_2)]^{d_i}$ then cancel around every vertex
of the star-star relation. For an edge crossed by the ordered pair of lines $(\ell_1,\ell_2)$, what remains is
\begin{equation}\label{BKMSform}
  \mathsf B^{(\iota)}_{\ell_1\ell_2}(\va,\boldsymbol b):=\varpi^{\mathsf Q(\va,\boldsymbol b)}\prod_{i=1}^n\prod_{k=0}^{a_i-b_i-1}
  \frac{1}{z_i^+(\ell_1)z_i^-(\ell_2)\,\varpi^{k+\iota}-z_i^+(\ell_2)z_i^-(\ell_1)}\;,
\end{equation}
with $\prod_{k=0}^{m-1}f_k:=\prod_{k=m}^{-1}f_k^{-1}$ for $m<0$. Here $\iota=1$ for all edges except $\bW_{\theta_{uv'}}$, for which $\iota=0$, and we have
$\mathsf B^{(0)}_{u v'}(\va,\boldsymbol b)\propto\mathsf B^{(1)}_{v' u}(\boldsymbol b,\va)^{-1}$. These weights depend only on the points of the curve and
on the discrete spins, and they satisfy the star-star relation \eqref{dssr} with $\mathsf S=1$ and $\mathcal R=1$.

\paragraph{Identification with the BKMS weights.}
Let us now compare these weights with the BKMS weights. For spins $l,m\in(\ZZ_M)^n$, the Boltzmann weights of
\cite{Bazhanov:1990qk} are defined by \cite[(0.4)--(0.7)]{Bazhanov:1990qk}
\begin{equation}
  \begin{gathered}
    \frac{\overline W_{pq}(l,m)}{\overline W_{pq}(0,0)}=\varpi^{\mathsf Q_{\rm BKMS}(l,m)}g_{pq}(0,l-m)\;,\\
    g_{pq}(m,m+\delta_\alpha):=\frac{h^+_{\alpha-1}(p)h^-_{\alpha-1}(q)-h^-_{\alpha-1}(p)h^+_{\alpha-1}(q)\,\varpi^{m_{\alpha-1,\alpha}}}{h^+_\alpha(p)h^-_\alpha(q)-h^-_\alpha(p)h^+_\alpha(q)\,\varpi^{1+m_{\alpha,\alpha+1}}}\;,
  \end{gathered}
\end{equation}
with $\mathsf Q_{\rm BKMS}(l,m):=\sum_\alpha m_{\alpha-1,\alpha}(l_\alpha-m_\alpha)$ and $m_{\alpha\beta}:=m_\alpha-m_\beta$, where the indices are in $\ZZ_n$. The
reversed orientation is given by $1/\overline W_{qp}$ \cite[Fig.~2]{Bazhanov:1990qk}. We put $a_\alpha:=m_\alpha-m_{\alpha+1}$, so that $\sum_\alpha a_\alpha=0$
automatically, and identify $(p,q)=(\ell_1,\ell_2)$ and $h^\pm_\alpha=z^\mp_\alpha$.

A step $m\to m+\delta_\alpha$ changes $a_\alpha\to a_\alpha+1$ and $a_{\alpha-1}\to a_{\alpha-1}-1$. In \eqref{BKMSform} with $\iota=1$ this produces one
new denominator factor $z^+_\alpha(\ell_1)z^-_\alpha(\ell_2)\varpi^{1+d_\alpha}-z^+_\alpha(\ell_2)z^-_\alpha(\ell_1)$ and removes one factor
$z^+_{\alpha-1}(\ell_1)z^-_{\alpha-1}(\ell_2)\varpi^{d_{\alpha-1}}-z^+_{\alpha-1}(\ell_2)z^-_{\alpha-1}(\ell_1)$, which are (minus) the denominator and numerator of
$g_{pq}(m,m+\delta_\alpha)$ above; since both carry the same sign, the ratios agree exactly. With the representatives $l,m\in\ZZ^n$ fixed by
$l_\alpha-l_{\alpha+1}=a_\alpha$, $m_\alpha-m_{\alpha+1}=b_\alpha$ and $l_1=m_1=0$, we have
$\mathsf B^{(1)}_{\ell_1\ell_2}(\va,\boldsymbol b)=\varpi^{\mathsf Q(\va,\boldsymbol b)-\mathsf Q_{\rm BKMS}(l,m)}\,\overline W_{pq}(l,m)/\overline W_{pq}(0,0)$, and it remains to
compare the quadratic forms. Using $\sum_ia_i=\sum_ib_i=0$ one finds
\begin{equation}\label{Qcompare}
  \begin{aligned}
    \mathsf Q(\va,\boldsymbol b) & =\tfrac12|\va|^2-\tfrac12\textstyle\sum_ia_ib_i+\tfrac12\sum_{i<j}(a_ib_j-a_jb_i)\;,           \\
    \mathsf Q_{\rm BKMS}(l,m)    & =\tfrac12|\boldsymbol b|^2-\tfrac12\textstyle\sum_ia_ib_i-\tfrac12\sum_{i<j}(a_ib_j-a_jb_i)\;,
  \end{aligned}
\end{equation}
so that
\begin{equation}\label{Qrel}
  \begin{aligned}
    \mathsf Q(\va,\boldsymbol b)-\mathsf Q_{\rm BKMS}(l,m) & =\tfrac12\big(|\va|^2-|\boldsymbol b|^2\big)+\textstyle\sum_{i<j}(a_ib_j-a_jb_i)\;,        \\
    \mathsf Q_{\rm BKMS}(l,m)                              & =\mathsf Q(\tilde{\va},\tilde{\boldsymbol b})-\tfrac12|\va|^2+\tfrac12|\boldsymbol b|^2\;,
  \end{aligned}
\end{equation}
where $\tilde a_i:=a_{n+1-i}$ denotes the reversal of the components. The first term of \eqref{Qrel} is a one-vertex factor $\varpi^{(|\va|^2-|\boldsymbol b|^2)/2}$: at the internal vertex of \eqref{dssr} the two $\overline{\mathsf W}$-edges contribute $\varpi^{-|\va_H|^2/2}$ each and the two
$\mathsf W$-edges $\varpi^{+|\va_H|^2/2}$ each, and on each external spin the contributions are the same on both sides. The second term is the reversal of the
components, which is a symmetry of the star-star relation, since it maps the curve \eqref{bkmscurve} with the maps $g_i$ to the curve with the maps $g_{n+1-i}$ and
relabels the spins. Therefore the weights with the phase $\varpi^{\mathsf Q}$ and those with $\varpi^{\mathsf Q_{\rm BKMS}}$ satisfy \eqref{dssr} simultaneously, and
$\mathsf B^{(1)}_{\ell_1\ell_2}$ is the BKMS weight $\overline W_{pq}$ with the components in reversed order, up to the one-vertex factor. In particular both signs
of the skew-symmetric part give solutions. 
The special edge with $\iota=0$ is the reversed-orientation weight $1/\overline W_{qp}$. Therefore, in the trigonometric regime with rapidities on the curve
\eqref{bkmscurve}, the discrete model obtained from the lens elliptic gamma function solution is the $(\ZZ_M)^{n-1}$ chiral Potts model of
Bazhanov--Kashaev--Mangazeev--Stroganov.\footnote{In the compressed form of these weights in \cite[Sec.~2]{Moosavian:2025nyv}, the numerator product must carry $\varpi^{-k}$
  relative to the denominators $\varpi^{1+k}$; with $\varpi^{+k}$, as printed there, the star-star relation fails for $M\ge3$ (checked numerically).}

\subsection{Star-star relation from the \texorpdfstring{$\mathfrak{sl}(n)$}{sl(n)} chiral Potts model}\label{sec:BKMSproof}
As we have seen in Section~\ref{sec:trig-general}, with rapidities on the curve \eqref{bkmscurve} the weights coincide with those of the
$(\ZZ_M)^{n-1}$ chiral Potts model \cite{Bazhanov:1990qk}, up to factors which cancel term by term in \eqref{dssr}. The relation \eqref{dssr} is therefore the
two-dimensional star-star relation of the $\mathfrak{sl}(n)$ chiral Potts model with the weights \cite[(3.19)]{Bazhanov:1990qk}. In \cite{Bazhanov:1990qk}
the Yang--Baxter equation is formulated for the face weights $\overline W_{qp}\overline W_{p'q}\overline W_{q'p'}/\overline W_{q'p}$ \cite[(3.23)]{Bazhanov:1990qk} and
proven in \cite{Kashaev:1990wy} (see also \cite{Date:1990bs}). The star-star relation itself was derived in \cite[footnote~5]{Bazhanov:1992jq}. There the argument used to derive the intertwining relation
for the layer-to-layer transfer matrices \cite[(3.10)]{Bazhanov:1992jq} is applied to a single front-to-back line of the three-dimensional
Bazhanov--Baxter model \cite{Bazhanov:1992jqa}. Since the conventions there differ from ours, we have also checked numerically that the original
weights \cite[(3.19)--(3.22)]{Bazhanov:1990qk} satisfy \eqref{dssr} in the form used here. This includes the degenerate point $K_{ij}\propto\mathbf 1$ of the
curve, i.e.\ the $\mathfrak{sl}(n)$ Fateev--Zamolodchikov type model \eqref{trigW}.

\section{Discussion}\label{sec:discussion}

We have studied the root-of-unity limit of the $A_{n-1}$ lens elliptic gamma function solution of the star-star relation
\cite{Yamazaki:2013nra,Kels:2017toi}. The quasi-classical expansion gives a classical star-star relation and, at the next order, a
$(\ZZ_{rN})^{n-1}$ solution of the star-star relation which depends on $r$ and $N$ only through $rN$. This confirms the picture anticipated in
\cite{Bazhanov:2011mz}. For constant backgrounds one obtains an elliptic $\mathfrak{sl}(n)$ Kashiwara--Miwa type model, which includes the
Quano--Fujii model, and in the trigonometric limit the $(\ZZ_M)^{n-1}$ chiral Potts model of \cite{Bazhanov:1990qk} together with its spectral curve. We
conclude with some comments and open problems.

\paragraph{Proof of the discrete star-star relation.}
Apart from the cases of Table~\ref{tab:status}, \eqref{dssr} rests on assumptions on the saddle points, and a possible route to a general proof
was outlined in Section~\ref{sec:proofs}. For the constant background the mechanism is visible in Section~\ref{sec:QFproof}, where the star-star
relation follows from the Yang--Baxter equation of the box $R$-matrix of \cite{Quano:1992wc}. For general backgrounds the box $R$-matrix would
depend on the saddle point, and the corresponding $L$-operators remain to be constructed.

\paragraph{Torsion points and new models.}
For the elliptic $\mathfrak{sl}(n)$ Kashiwara--Miwa type models a proof is missing in particular at the real points for even $n\ge4$, the
$\mathfrak{sl}(n)$ analogues of the original Kashiwara--Miwa model (Appendix~\ref{app:torsion}); for $\gcd(n,M)>1$, where the Yang--Baxter equation
for the extension of the Quano--Fujii model to $Q/MQ$ would be needed (Section~\ref{sec:gcd}); and at the points violating \eqref{QFcond}, whose
algebraic origin is not clear to us. For non-symmetric saddle points \eqref{dssr} is an elliptic deformation of the $\mathfrak{sl}(n)$ chiral Potts
model, analogous to the $Q4$-type $\ZZ_N$ model of \cite{Bazhanov:2010kz} for $n=2$, and deserves further study, as does the question of a regime
with positive Boltzmann weights for $n\ge3$ (Section~\ref{sec:reality}).

\paragraph{Classical integrability.}
For $n=2$ the saddle point equations are the three-leg form of $Q4$ \cite{Adler:2002abs,Bazhanov:2016ajm}. For general $n$ they coincide with the equations of
motion of the classical star-star relation of \cite{Bazhanov:2011mz}. Their hyperbolic counterparts give $(n-1)$-component extensions of integrable
5-point equations \cite{Kels:2024afm}, and related multi-component quad equations were proposed in \cite{Kels:2018qzx}. While the classical star-star relation \eqref{cssr}
follows from \cite{Bazhanov:2011mz}, it is an open problem to establish multidimensional consistency properties of \eqref{saddlemod},
analogous to the 3D consistency of $Q4$, directly. In the trigonometric limit the equations reduce to a condition on
cross-ratios, which suggests a description in terms of $n$ copies of $Q3_{(\delta=0)}$ coupled through the spectral curve. Lens elliptic
hypergeometric sum/integral expressions also give $\tau$-functions of a lens generalization of the elliptic discrete Painlev\'e equation, with both
continuous and discrete variables \cite{Kels:2018tau}; the relation between these $\tau$-functions and the classical equations found here
remains to be determined.

\paragraph{Gauge theory.}
In the Gauge/YBE correspondence the model studied here corresponds to the lens index of four-dimensional $\mathcal N=1$ $\mathrm{SU}(n)$ gauge theories
\cite{Yamazaki:2013nra,Kels:2017toi}, and the root-of-unity limit is a limit of the fugacities. For $n=2$ and $r=N=1$ it was proposed in
\cite{Kels:2017vbc}, in parallel with the Bethe/gauge correspondence of Nekrasov and Shatashvili \cite{Nekrasov:2009uh,Nekrasov:2009ui}, that the
leading term of the quasi-classical expansion is the effective twisted superpotential of the two-dimensional $\mathcal N=(2,2)$ theory obtained
by compactifying the four-dimensional theory on a torus, so that the saddle point equations are its vacuum equations. This suggests that the
action $\mathcal A$ and the saddle point equations \eqref{saddlemod} play the same role for the $\mathrm{SU}(n)$ theories, and that the discrete
spins which appear at $O(\hbar^0)$ encode additional data of the vacua at a root of unity. Making this precise requires an understanding of the discrete models directly in
gauge theory, including the appearance of the $rN$-torsion points and the dependence on $rN$ only, possibly through a reduction to
lower-dimensional theories (cf.\ \cite{Kels:2017vbc}; see also \cite{Yamazaki:2015voa} for $R$-matrices given by theta functions,
obtained from two-dimensional $\mathcal N=(2,2)$ theories, and their reduction to one dimension). For ordinary superconformal indices, surface defects acting as transfer matrices were constructed in
\cite{Maruyoshi:2016caf,Yagi:2017hmj} (see \cite{Yagi:2017hmj} for their relation to elliptic quantum groups).
Surface defects in the lens index, in particular in the $\mathcal N=2$ lens index \cite{Razamat:2013jxa}, give rise to transfer matrices and
difference operators of the corresponding integrable models \cite{Maruyoshi:2026nrl,Yamazaki:2026kws}. Their root-of-unity limit,
which may give transfer matrices of the discrete models found here, has not been studied. Generalized chiral Potts models have recently been realized \cite{Moosavian:2025nyv} in terms of
(semi-)holomorphic Chern--Simons theories \cite{Costello:2013zra,Costello:2017dso} (see \cite{Yamazaki:2025ocv} for an overview). Whether the elliptic models found here admit a similar realization is
not known.

\paragraph{Other generalizations.}
Several generalizations of this work suggest themselves. A $BC_n$-type counterpart of the $A_{n-1}$ sum/integral transformation is proven in
\cite{Kels:2017toi}, and one may study the root-of-unity limit of the corresponding integrals and of the associated lattice models.
Related multivariate rarefied elliptic beta integrals and symmetry transformations were studied in \cite{Spiridonov:2016uae}. Star-star
relations built from hyperbolic gamma functions and their lens versions \cite{Gahramanov:2016ilb,Catak:2021coz,Gahramanov:2022jxz} have quasi-classical expansions of a
similar form \cite{Kels:2024afm}, while the analogue of the discrete models found here is not known in the hyperbolic case. One could also
consider other roots of unity, for example limits in which both nomes approach roots of unity, or non-primitive roots of unity.

We hope to return to some of these questions in the future.

\section*{Acknowledgments}\addcontentsline{toc}{section}{Acknowledgments}
This research was supported in part by the World Premier International Research Center Initiative (WPI), MEXT, Japan; JSPS KAKENHI Grant No.~23K25865; JST, Japan (CREST Grant No.~JPMJCR26XA, Moonshot R\&D Grant No.~JPMJMS256E); and IBM-UTokyo-sponsored research. The author was assisted by Anthropic Claude and OpenAI ChatGPT in writing and editing the manuscript, in particular in numerical checks.

\appendix
\section{Theta functions}\label{app:theta}
\begin{align}
  \ovth_1(x|\tau) & :=2\EXP^{\frac{\pi\ii\tau}{4}}\sin x\prod_{k\ge1}(1-\EXP^{2\ii x}\EXP^{2\pi\ii\tau k})(1-\EXP^{-2\ii x}\EXP^{2\pi\ii\tau k})\;, \\
  \ovth_4(x|\tau) & :=\prod_{k\ge1}(1-\EXP^{2\ii x}\EXP^{\pi\ii\tau(2k-1)})(1-\EXP^{-2\ii x}\EXP^{\pi\ii\tau(2k-1)})\;,
\end{align}
and $\ovth_2(x|\tau):=\ovth_1(x+\pi/2|\tau)$, $\ovth_3(x|\tau):=\ovth_4(x+\pi/2|\tau)$,
$\ovth_1(x|\tau)=-\ii\,\EXP^{\ii x+\pi\ii\tau/4}\ovth_4(x+\pi\tau/2|\tau)$, and $\vartheta_j:=G(\tau)\ovth_j$ with $G(\tau):=\prod_{k\ge1}(1-\EXP^{2\pi\ii\tau k})$.
The function $\ovth_1$ is odd, and under the periods it transforms as
\begin{equation}\label{thetaper}
  \ovth_1(x+\pi|\tau)=-\ovth_1(x|\tau),\qquad \ovth_1(x+\pi\tau|\tau)=-\EXP^{-2\ii x-\pi\ii\tau}\,\ovth_1(x|\tau)\;,
\end{equation}
and correspondingly for $\ovth_{2,3,4}$; its zeros are at $x\in\pi\ZZ+\pi\tau\ZZ$.
Under the modular group, $\vartheta_1(z|\tau+1)=\EXP^{\pi\ii/4}\vartheta_1(z|\tau)$, $\vartheta_2(z|\tau+1)=\EXP^{\pi\ii/4}\vartheta_2(z|\tau)$,
$\vartheta_3(z|\tau+1)=\vartheta_4(z|\tau)$, and
\begin{equation}
  \vartheta_1\Big(\frac{z}{\tau}\Big|-\frac1\tau\Big)=-\ii\sqrt{\frac{\tau}{\ii}}\,\EXP^{\frac{\ii z^2}{\pi\tau}}\vartheta_1(z|\tau),\qquad
  \vartheta_4\Big(\frac{z}{\tau}\Big|-\frac1\tau\Big)=\sqrt{\frac{\tau}{\ii}}\,\EXP^{\frac{\ii z^2}{\pi\tau}}\vartheta_2(z|\tau)\;;
\end{equation}
more generally, for $g=\big(\begin{smallmatrix}\alpha&\beta\\\gamma&\delta\end{smallmatrix}\big)\in \mathrm{SL}(2,\ZZ)$,
\begin{equation}\label{thetamod}
  \vartheta_1\Big(\frac{z}{\gamma\tau+\delta}\,\Big|\,\frac{\alpha\tau+\beta}{\gamma\tau+\delta}\Big)=\epsilon_g\,(\gamma\tau+\delta)^{1/2}\,\EXP^{\frac{\ii\gamma z^2}{\pi(\gamma\tau+\delta)}}\,\vartheta_1(z|\tau)\;,\qquad \epsilon_g^8=1\;.
\end{equation}
We shall only need the Gaussian factor and the factor $(\gamma\tau+\delta)^{1/2}$; the root of unity $\epsilon_g$ and the factors $G$ in
$\ovth_1=\vartheta_1/G$ only give $z$-independent constants.

\section{Conventions for the lens elliptic gamma function}\label{app:dict}

The star-star relation of \cite{Yamazaki:2013nra} is written in \cite[Sec.~6]{Kels:2017toi} in terms of the lens elliptic gamma function
(in this appendix $\sigma$ and $\tau$ denote the two modular parameters of \cite{Kels:2017toi}, not a spin and the nome of Section~\ref{sec:asymptotics})
\begin{equation}\label{Gamma1704}
  \begin{aligned}
    \Gamma(z,m;\sigma,\tau) & :=\EXP^{\varphi_e(z,m;\sigma,\tau)}\,\Gamma_1(z+\sigma m;r\sigma,\sigma+\tau)\,
    \Gamma_1(z+\tau(r-m);r\tau,\sigma+\tau)\;,                                                                            \\
    \Gamma_1(z;\sigma,\tau) & :=\prod_{j,k=0}^\infty\frac{1-\EXP^{-2\pi\ii z}p^{j+1}q^{k+1}}{1-\EXP^{2\pi\ii z}p^jq^k}\;,
  \end{aligned}
\end{equation}
with $p:=\EXP^{2\pi\ii\sigma}$, $q:=\EXP^{2\pi\ii\tau}$, $m\in\{0,\ldots,r-1\}$, and
\begin{equation}
  \varphi_e(z,m;\sigma,\tau):=2\pi\ii\big[R_2(z,0;\sigma,\tau)+R_2(0,m;\tfrac12,-\tfrac12)-R_2(z,m;\sigma,\tau)\big]\;,
\end{equation}
where
\begin{multline*}
  R_2(z,m;\sigma,\tau):=\frac{(\sigma+\tau-2z)\big(2z^2-2z(\sigma+\tau)+\sigma\tau(r^2+6(m-r)m)+1\big)}{24r\sigma\tau}\\-\frac{(\sigma-\tau)(2m-r)(m-r)m}{12r}
\end{multline*}
\cite[(2.2)--(2.5)]{Kels:2017toi}. The function $\Phi$ of \eqref{Phidef}, which is the one used in \cite{Kels:2017vbc}, is related to it by
\begin{equation}\label{PhiGamma}
  \Phi(z,m)=\Gamma\Big(\frac{\tau_1+\tau_2}{2}-\frac{z}{\pi},\,m;\,\tau_1,\tau_2\Big),\qquad\text{i.e.}\qquad \sigma=\tau_1,\quad\tau=\tau_2\;.
\end{equation}
For the infinite products this follows from
$\Phi_1(z;\EXP^{\pi\ii s},\EXP^{\pi\ii t})=\Gamma_1\big((s+t)/2-z/\pi;s,t\big)$, applied with $(s,t)=(\tau_1+\tau_2,r\tau_1)$ and $(\tau_1+\tau_2,r\tau_2)$:
the first factor of \eqref{Phidef} becomes $\Gamma_1(z'+\tau_1m;r\tau_1,\tau_1+\tau_2)$ and the second one $\Gamma_1(z'+\tau_2(r-m);r\tau_2,\tau_1+\tau_2)$, with
$z':=(\tau_1+\tau_2)/2-z/\pi$. The exponential prefactors $\EXP^{E(z,m)}$ and $\EXP^{\varphi_e(z',m)}$ then agree as well, so that \eqref{PhiGamma} holds
exactly, without additional normalization factors.

In \cite[Sec.~6]{Kels:2017toi} the Boltzmann weights are expressed through
$\Phi_{\text{\cite{Kels:2017toi}}}(z,m):=\Gamma\big((\sigma+\tau)/2-z,-m;\sigma,\tau\big)$, which by \eqref{PhiGamma} is
\begin{equation}
  \Phi_{\text{\cite{Kels:2017toi}}}(z,m)=\Phi(\pi z,-m)\;.
\end{equation}
The weights $W_\alpha$, $S$ and the crossing parameter $\eta=-(\ii/2)(\sigma+\tau)$ of \cite[(6.3)]{Kels:2017toi} are therefore mapped to
\eqref{BW}--\eqref{Sdef} by
\begin{equation}
  x^{\text{\cite{Kels:2017toi}}}_{i}=\frac{x_i}{\pi},\qquad m^{\text{\cite{Kels:2017toi}}}_{i}=-m_i,\qquad \alpha^{\text{\cite{Kels:2017toi}}}=\frac{\alpha}{\pi},\qquad
  \eta^{\text{\cite{Kels:2017toi}}}=\frac{\eta}{\pi},\qquad (\sigma,\tau)=(\tau_1,\tau_2)\;.
\end{equation}
The sign of the discrete components is immaterial, since the measure and the constraint $\sum_im_i\equiv0$ are invariant under $\vm\to-\vm$. The integration range
$x^{\text{\cite{Kels:2017toi}}}\in[0,1)$ becomes $x\in[0,\pi)$. Finally, $\overline W_\alpha=\sqrt{S(\sigma_a)S(\sigma_b)}\,W_{\eta-\alpha}$ in \cite{Kels:2017toi} differs from
\eqref{BW} by the factors $\sqrt S$, which cancel in the star-star relation as explained in Section~\ref{sec:model}.

\section{\texorpdfstring{The Quano--Fujii model for $\gcd(n,M)>1$}{The Quano-Fujii model for gcd(n,M)>1}}\label{app:gcd}

In this appendix we give the details of Section~\ref{sec:gcd}, with $d=\gcd(n,M)$.
Let $\overline P:=\sum_{j=0}^{n-1}\ZZ\bar\epsilon_j$ be the weight lattice of
$\mathfrak{sl}(n)$, with $\bar\epsilon_j$ as in Section~\ref{sec:QF}, and recall that the root
lattice $Q$ is a sublattice of $\overline P$, with $\overline P/Q\cong\ZZ_n$. Our discrete spins are the elements $\va=\sum_{i=1}^na_i\epsilon_{i-1}=\sum_{i=1}^na_i\bar\epsilon_{i-1}$ of $Q$ modulo $MQ$, while
the local states of \cite{Quano:1992wc} with the cyclic condition $a\sim a+M\bar\epsilon_i$ are the classes of $a_0+\overline P$ modulo
$M\overline P$, and the dictionary \eqref{QFdict} is the map $Q/MQ\to\overline P/M\overline P$ induced by the inclusion $Q\subset\overline P$,
up to the shift by $a_0$. Its kernel is $(Q\cap M\overline P)/MQ\cong\ZZ_d$, generated by the class of $\va_{\rm ker}$ of \eqref{kerQF}, and its image is the subgroup of index $d$ of the states whose charge, the class in $\overline P/(Q+M\overline P)\cong\ZZ_d$, vanishes.
For $d=1$ the map is a bijection, as used in Section~\ref{sec:QFproof}; for $d>1$ it is $d$-to-one onto the charge-zero states (for $(n,M)=(2,2)$, $(3,3)$, $(4,2)$:
$2$, $9$, $8$ spins onto $1$, $3$, $4$ states).

For $d>1$ it is $Q/MQ$, rather than $\overline P/M\overline P$, which is compatible with the recursion \eqref{QFrec}. Indeed, \eqref{QFrec} and the
initial condition $\overline W^{\rm QF}_{aa}(0)=1$ determine $\bar w_{ab}(u_{\rm QF})$ for $a-b\in Q$ up to a spin-independent factor $f(u_{\rm QF})$, $f(0)=1$, which
cancels in the star-star relation and is fixed by \eqref{KMnW}. In particular, by \eqref{QFid}, both $\mathsf W_\theta(\va,\boldsymbol b)$ and
$\mathsf W_\theta(\va+\va_{\rm ker},\boldsymbol b)$ are determined, and they are different functions of the spins: for $(n,M)=(2,2)$, with
$\va=(a,-a)$ and $\boldsymbol b=(b,-b)$, \eqref{KMnW} gives
\begin{equation}
  \frac{\mathsf W_\theta(\va+\va_{\rm ker},\boldsymbol b)}{\mathsf W_\theta(\va,\boldsymbol b)}=\pm\,\Xi(\theta)^{\pm1},\qquad
  \Xi(\theta):=\frac{\ovth_1(\frac\theta2-\frac\pi4\,|\,T)}{\ovth_1(\frac\theta2+\frac\pi4\,|\,T)}\;,
\end{equation}
with the exponent $+1$ for $a=b$ and $-1$ for $a\neq b$ (in the other cases the ratio is not even a gauge factor $f(\va)f'(\boldsymbol b)$, as we
have checked numerically). For $d>1$ the cyclic condition of \cite{Quano:1992wc}, which identifies $\va$ and $\va+\va_{\rm ker}$, is thus
incompatible with the recursion; this is where their assumption $\gcd(n,N_{\rm QF})=1$ enters. The weights \eqref{KMnW}--\eqref{KMnS} on
$(\ZZ_M)^{n-1}\cong Q/MQ$, on the other hand, are well-defined for all $n$ and $M$ (Section~\ref{sec:KMreal}). For $\gcd(n,M)>1$ the model at
$\omega_0=\pi/n$ is therefore an extension of the Quano--Fujii model, defined by the recursion \eqref{QFrec} with the initial condition on $Q/MQ$
instead of $\overline P/M\overline P$ and normalized by \eqref{KMnW}, and \eqref{KMnssr} is a star-star relation for this extension.

A proof of \eqref{KMnssr} for $\gcd(n,M)>1$ along the lines of \cite{Quano:1992wc} would require the Yang--Baxter equation for this extension.
The natural state space is $\overline P/MQ$, with $nM^{n-1}$ elements and $n$ sectors labeled by the charge in $\overline P/Q\cong\ZZ_n$, the
charge-zero sector being $Q/MQ$. The recursion defines $\bar w_{ab}$ within each sector ($a-b\in Q$), where it is generically non-zero, and we
define the box operator on $\mathcal H\otimes\mathcal H$, $\mathcal H$ spanned by $\overline P/MQ$, as the direct sum of the box operators
\eqref{QFS} of the sectors: $\langle a,d|\mathcal S|b,c\rangle$ is given by \eqref{QFS} if $a,b,c,d$ have the same charge and vanishes otherwise.
Within a sector the box weights are invariant under changes of the lifts of their four spins by $MQ$, as follows from the quasi-periodicity of
$\ovth_1$ (and as we have checked numerically), so that this operator is well-defined; whether this block-diagonal extension is the right one for
the intertwining construction of \cite{Quano:1992wc} is open. The sectors are equivalent: the cyclic shift
\begin{equation}\label{Cshift}
  C:\ (a_1,\ldots,a_n)\ \mapsto\ (a_2,\ldots,a_n,a_1-1)\;,
\end{equation}
which lowers the charge by one, satisfies $\mathsf W_\theta(C\va,C\boldsymbol b)=\mathsf W_\theta(\va,\boldsymbol b)$ and
$\mathsf S(C\va)=(-1)^{n-1}\mathsf S(\va)$ at $\omega_0=\pi/n$, by \eqref{KMnW}--\eqref{KMnS} and the periodicity \eqref{thetaper}. The
$L$-operators of \cite{Quano:1992wc}, whose matrix elements connect $a$ and $a+\bar\epsilon_i$, shift the charge by one. Their construction on
$\overline P/MQ$, i.e.\ the behaviour of their matrix elements under changes of the representatives modulo $MQ$, and the uniqueness of the
intertwiner, which in \cite{Quano:1992wc} rests on the cyclic condition, remain to be worked out (see also Section~\ref{sec:discussion}). For
$\gcd(n,M)>1$ we have verified \eqref{KMnssr} at $\omega_0=\pi/n$ numerically (Appendix~\ref{app:numerics}).

\section{Symmetries and general torsion points}\label{app:torsion}

In this appendix we explain the results on general torsion points summarized in Section~\ref{sec:torsion}. We label
$\omega_0=\pi(\kappa MT-\nu)/n$ by $(\kappa,\nu)$, or equivalently by $(\mathsf X,\mathsf Y):=(-\nu,\kappa)$. The real points \eqref{phi0} are
$(\kappa,\nu)=(\kappa,\kappa r)$ with $\gcd(\kappa,n)=1$ and $r\,|\,M$, and the Quano--Fujii point $\omega_0=\pi/n$ is $(\kappa,\nu)=(0,-1)$, i.e.\
$(\mathsf X,\mathsf Y)=(1,0)$. Numerically, \eqref{KMnssr} holds with $\mathcal R=1$ at every primitive point, $\gcd(\kappa,\nu,n)=1$. At the non-primitive points
it fails whenever the weights are non-degenerate, for example at $\omega_0=\pi/2$ for $(n,M)=(4,2)$ and at $\omega_0=0$ for $(n,M)=(2,3)$. In some
cases $\mathsf S$ vanishes identically, so that both sides are zero and no non-trivial model arises: at $\omega_0=0$ two components
$a_i\equiv a_j$ (mod $M$) give a vanishing factor of $\mathsf S$, which happens for every spin if $n>M$.

The points other than \eqref{phi0} do not come from a real configuration: in the original
variables their differences $\xi^0_{ij}$ have non-zero imaginary parts, and some edge arguments leave the strip of validity of \eqref{KYasym}.
We explain these observations by symmetries of \eqref{KMnssr} and their orbits, and then discuss the weaker notion of modular images of
the Quano--Fujii weights and the points with odd $\kappa$ for even $n$.

\subsection{Symmetries of the star-star relation}\label{app:sym}
We first list the transformations of the parameters $(T,\omega_0,\theta)$ which map the star-star relation \eqref{KMnssr} with the
weights \eqref{KMnW}--\eqref{KMnS} to itself, possibly at another torsion point. They fall into three classes: (a) exact symmetries of the
weights up to relabelings and constants, (b) the shift of $\omega_0$ by the period $\pi MT$, and (c) a class of modular transformations.
\begin{itemize}
  \item[(a)] In \eqref{KMnW}--\eqref{KMnS} the shift $\omega_0\to\omega_0+\pi$ is equivalent to the relabeling $a_i\to a_i+i$ of all spins: the
        factors of $\mathsf W_\theta(\va,\boldsymbol b)$ with $k<a_i-b_j$ become those with $k<a_i-b_j+i-j$, up to the spin-independent product over
        $0\le k<i-j$, and $\mathsf S$ changes accordingly. For odd $n$ a further uniform shift of all spins brings them back to $\sum_ia_i=0$; for
        even $n$ the same holds for $\omega_0\to\omega_0+2\pi$. Hence only $\nu$ modulo $n$ for odd $n$, or modulo $2n$ for even $n$, matters. The
        shift $T\to T+1$ changes the weights only by constants, since $\sum_{i,j}(a_i-b_j)=0$, and maps $\nu\to\nu+\kappa M$. Together with the mirror
        symmetry $\omega_0\to-\omega_0$ these are exact symmetries.
  \item[(b)] Consider $\omega_0\to\omega_0+\pi MT$, i.e.\ $\kappa\to\kappa+n$. By the quasi-periodicity \eqref{thetaper}, the weight
        $\mathsf W_\theta(\va,\boldsymbol b)$ changes by a factor $\EXP^{F(\va)-F(\boldsymbol b)}$ and by $\EXP^{(2\ii n\theta/M)\langle\rho,\va+\boldsymbol b\rangle}$,
        with $\rho$ as in \eqref{xi0}. The first factor cancels in \eqref{KMnssr}. The second one cancels on the external spins, as in the derivation of
        \eqref{KMnW}, and produces $\EXP^{2\pi\ii n\langle\rho,\va_H\rangle/M}$ around the internal vertex, since the four spectral variables there add up to $\pi$.
        The same quasi-periodicity gives $\mathsf S(\va)\to\mathsf S(\va)\,\EXP^{-2\pi\ii n\langle\rho,\va\rangle/M}$, where we used
        $\sum_{i<j}(j-i)a_{ij}=-n\langle\rho,\va\rangle$. The two factors cancel exactly.
  \item[(c)] Consider $g=\big(\begin{smallmatrix}\alpha&\beta\\\gamma&\delta\end{smallmatrix}\big)\in \mathrm{SL}(2,\ZZ)$ with $\alpha\equiv1$ and $\gamma\equiv0$
        (mod $2M$), acting as $(T,\omega_0,\theta)\to(gT,\omega_0/\mathcal J,\theta/\mathcal J)$ with $\mathcal J:=\gamma T+\delta$. Since $(k+1/2)\pi/\mathcal J$
        and $\pi/(2\mathcal J)$ differ from $(k+1/2)\pi$ and $\pi/2$ by elements of $M(\pi\ZZ+\pi gT\ZZ)$, the step of the discrete spins and the
        crossing parameter are preserved up to periods. The resulting Gaussian and quasi-periodicity factors either are gauge factors, or cancel on
        the external spins between the two sides, or depend only on the internal spin and cancel against the transformation of $\mathsf S$
        (Appendix~\ref{app:modular}). Thus (c) maps \eqref{KMnssr} at $(T,\omega_0)$ to \eqref{KMnssr} at $(gT,\omega_0/\mathcal J)$. On $(\mathsf X,\mathsf Y)$
        it acts by the matrix
        \begin{equation}\label{hmat}
          \begin{gathered}
            \check g=\begin{pmatrix}\check g_{11}&\check g_{12}\\\check g_{21}&\check g_{22}\end{pmatrix}:=\begin{pmatrix}\alpha&-\beta M\\-\frac{\gamma}{M}&\delta\end{pmatrix},\\
            \check g_{11}\equiv\check g_{22}\equiv1\ (\mathrm{mod}\ 2M),\qquad \check g_{12}\equiv0\ (\mathrm{mod}\ M),\qquad \check g_{21}\equiv0\ (\mathrm{mod}\ 2)\;,
          \end{gathered}
        \end{equation}
        and conversely every $\check g\in\mathrm{SL}(2,\ZZ)$ satisfying these congruences arises from a transformation (c). We denote by $\mathsf H$ the group of these matrices.
\end{itemize}

\subsection{Orbits}\label{app:orbits}
The relation \eqref{KMnssr} is obtained from the quasi-classical expansion at the real points, and, for $\gcd(n,M)=1$, it is proven at the
Quano--Fujii point (Section~\ref{sec:QFproof}). The orbits of these points under (a)--(c) are as follows.

\begin{proposition}\label{prop:orbits}
  Let $\mathsf G$ be the group generated by (a)--(c), acting on the pairs $(\kappa,\nu)\in\ZZ^2$ with $\gcd(\kappa,\nu,n)=1$, and let $\hat n:=2n$ if $n$ and $M$
  are both even and $\hat n:=n$ otherwise. Then:
  \begin{itemize}
    \item[(i)] $\nu$ matters only modulo $\hat n$ and $\kappa$ only modulo $n$. In particular, for even $n$ and odd $M$ the shift $\nu\to\nu+n$ is a symmetry of \eqref{KMnssr}, although it is not a relabeling.
    \item[(ii)] $\mathsf G$ acts transitively on the primitive points if and only if $\gcd(n,2M)=1$.
    \item[(iii)] $(\kappa,\nu)$ lies in the orbit of the Quano--Fujii point if and only if $\nu\equiv\pm1$ (mod $\hat d$) and, for even $n$, $\kappa$ is even, where
          $\hat d:=d$ for odd $n$ and $\hat d:=2d$ for even $n$, with $d=\gcd(n,M)$.
  \end{itemize}
\end{proposition}
\begin{proof}
  (i) The group $\mathsf G$ contains the translations $(\kappa,\nu)\to(\kappa,\nu+\bar n)$ with $\bar n:=n$ for odd $n$ and $\bar n:=2n$ for even $n$,
  by the relabeling (a), and $(\kappa,\nu)\to(\kappa+n,\nu)$ by (b). Conjugating (b) by $T\to T\pm1$, i.e.\ applying $T\to T+1$, (b), $T\to T-1$ and the inverse of (b)
  in this order, gives $(\kappa,\nu)\to(\kappa,\nu-nM)$. Hence $\mathsf G$ contains the translation by $(0,\gcd(\bar n,nM))=(0,\hat n)$. Conversely, the lattice
  $\Lambda:=\hat n\ZZ\times n\ZZ$ in the $(\mathsf X,\mathsf Y)$-plane is invariant under $\mathsf H$ and under $-1$, since $M\,|\,\check g_{12}$ and $\hat n\,|\,nM$. Therefore
  $\mathsf G$ acts on $V:=\ZZ/\hat n\times\ZZ/n$ through the linear group generated by $\mathsf H$ and $-1$.

  (ii), (iii) Let $K:=2M\hat n$. Since $\mathrm{SL}(2,\ZZ)\to\mathrm{SL}(2,\ZZ/K)$ is surjective and the congruences in \eqref{hmat} are conditions modulo divisors of $K$, the image of $\mathsf H$ in
  $\mathrm{SL}(2,\ZZ/K)$ consists of all residue classes satisfying \eqref{hmat}. By the Chinese remainder theorem this image decomposes into local factors over the
  primes $p\,|\,K$, and $V$ decomposes into the product of $V_p:=\ZZ/p^{v_p(\hat n)}\times\ZZ/p^{v_p(n)}$ over the primes $p\,|\,n$. Since $\hat n$ and $n$ have the
  same prime divisors, the local factors with $p\nmid n$ act trivially on $V$. Hence the orbits of $\mathsf H$ on $V$ are products of the orbits of the local groups $\mathsf H_p$ on $V_p$ for $p\,|\,n$. Write $e:=v_p(n)$ and $m:=v_p(M)$, and
  consider a primitive vector $(\mathsf X,\mathsf Y)\in V_p$, i.e.\ one with $p\nmid\gcd(\mathsf X,\mathsf Y)$. For odd $p$ with $m=0$ the group $\mathsf H_p$ is all of
  $\mathrm{SL}(2,\ZZ/p^{e})$, which acts transitively on the primitive vectors, since every primitive vector is the first column of a matrix of determinant one. For odd $p$ with $m\ge1$, $\mathsf H_p$ consists of the matrices with
  $\check g_{11}\equiv\check g_{22}\equiv1$, $\check g_{12}\equiv0$ (mod $p^m$) and $\check g_{21}$ arbitrary. Hence $\mathsf X$ modulo $p^{\mu}$, $\mu:=\min(m,e)$, is invariant, and the orbit of $(1,0)$, which consists
  of the first columns $(\check g_{11},\check g_{21})$, is $\{\mathsf X\equiv1\ (\mathrm{mod}\ p^\mu)\}$. For $p=2$, i.e.\ even $n$, $\mathsf H_2$ consists of the matrices with $\check g_{11}\equiv\check g_{22}\equiv1$ (mod $2^{m+1}$),
  $\check g_{12}\equiv0$ (mod $2^m$) and even $\check g_{21}$; here $v_2(\hat n)=e+1$ if $m\ge1$ and $v_2(\hat n)=e$ if $m=0$. Hence $\mathsf Y$ modulo $2$ is invariant, and the orbit of $(1,0)$ is
  $\{\mathsf X\equiv1\ (\mathrm{mod}\ 2^{\min(m,e)+1}),\ \mathsf Y\ \text{even}\}$. In the last two cases every pair $(x,y)$ with the stated congruences is the first
  column of an element of $\mathsf H_p$, namely of $\big(\begin{smallmatrix}x&0\\y&x^{-1}\end{smallmatrix}\big)$. Combining the local statements with the
  global sign $-1$ gives (iii), with $\hat d=\prod_{p\,|\,d}p^{\min(v_p(n),v_p(M))}$ times $2$ for even $n$. Finally, if $\gcd(n,2M)=1$ all local groups act
  transitively and so does $\mathsf G$, while if $\gcd(n,2M)>1$ either $\mathsf Y$ modulo $2$ or $\mathsf X$ modulo some $p^\mu$ with $\mu\ge1$ is an invariant up to sign, and the primitive
  points $(1,0)$ and $(0,1)$ lie in different orbits. This proves (ii).
\end{proof}

We have also confirmed (i)--(iii) by enumerating the orbits for $n,M\le10$. Consequently the primitive points fall into the following
(overlapping) classes, which are included in Table~\ref{tab:status}.
\begin{itemize}
  \item[(1)] The orbit of the Quano--Fujii point, characterized by (iii). For $\gcd(n,M)=1$, \eqref{KMnssr} is proven there by
        Section~\ref{sec:QFproof} and (a)--(c); this orbit contains all primitive points if $\gcd(n,2M)=1$.
  \item[(2)] The orbits of the real points. There \eqref{KMnssr} follows from \eqref{dssr} under the assumptions (A1)--(A3).
  \item[(3)] The remaining points. For $n,M\le10$ they occur exactly for $(n,M)=(4,4)$, $(4,8)$, $(6,3)$, $(6,6)$, $(6,9)$, $(8,4)$, $(8,8)$, $(9,9)$,
        $(10,5)$, $(10,10)$, for example $\omega_0=\pi(2MT-3)/6$ for $(n,M)=(6,3)$. There \eqref{KMnssr} is supported by numerical evidence only
        (checked for $(n,M)=(4,4)$, $(4,8)$, $(6,3)$, $(6,6)$, Appendix~\ref{app:numerics}).
\end{itemize}
For even $n$ the real points have odd $\kappa$, so that (1) and (2) are disjoint. For odd $n$ the real point $(\kappa,\kappa r)$ lies in (1) if
and only if $\kappa r\equiv\pm1$ (mod $\gcd(n,M)$), which yields a proof only for $\gcd(n,M)=1$.

\subsection{Modular images of the Quano--Fujii weights}\label{app:QFimages}
A weaker relation between the models at different torsion points is obtained from the transformations
$g=\big(\begin{smallmatrix}\alpha&\beta\\\gamma&\delta\end{smallmatrix}\big)\in \mathrm{SL}(2,\ZZ)$ with $\gamma\equiv0$ and $\alpha\equiv\delta\equiv\pm1$
(mod $M$). They map $\ovth_1(\cdot/M|T)$ to theta functions with nome $gT$ and preserve the step $\pi/M$ of the discrete spins. Combined with the
relabelings $a_i\to a_i+mi$ of (a), they relate the Boltzmann weights at $(T,\omega_0)$ and $(gT,\omega_0/(\gamma T+\delta))$ up to factors
$F(\va)F'(\boldsymbol b)$, but in general not the star-star relations, since the crossing parameter is not preserved. (For even $n$ and odd $m$
the relabeling is combined with the cyclic shift $C^{n/2}$ of \eqref{Cshift} and a uniform shift, which bring the spins back to $Q$ and are exact
symmetries of the weights at the Quano--Fujii point.) These transformations act on the points $\omega_0/M$ and $\pi/M$ of $\CC/(\pi\ZZ+\pi T\ZZ)$
as group isomorphisms, and $n\,\omega_0/M=\pi/M$ at the Quano--Fujii point. Hence the model at $\omega_0$ can be a modular image of the
Quano--Fujii model only if
\begin{equation}\label{QFcond}
  n\,\omega_0\equiv\pm\pi\pmod{\pi\gcd(n,M)\ZZ+\pi MT\ZZ}\;,
\end{equation}
i.e.\ $\gcd(n,M)\,|\,\nu\pm1$. The condition is also sufficient. Indeed, these transformations act on $(\mathsf X,\mathsf Y)$, with
$\mathsf X$ and $\mathsf Y$ modulo $n$, by the matrices \eqref{hmat} with the congruences replaced by $\check g_{11}\equiv\check g_{22}\equiv\pm1$,
$\check g_{12}\equiv0$ (mod $M$), and the local argument in the proof of Proposition~\ref{prop:orbits}~(iii), in which no parity condition arises
now, shows that the orbit of $(1,0)$ consists of all primitive $(\mathsf X,\mathsf Y)$ with $\mathsf X\equiv\pm1$ (mod $\gcd(n,M)$). For such a
point the model at the nome $T$ is the modular image of the Quano--Fujii model at $\tau_{\rm QF}=g^{-1}T$, with
$u_{\rm QF}=(\mathcal J\theta-\pi/2)/(\pi M)$, $\mathcal J=\gamma\tau_{\rm QF}+\delta$; the sign $-1$ corresponds to the reversed order of the
components. In particular, for $\gcd(n,M)=1$ every primitive point is a modular image of the Quano--Fujii point in this sense. For $n=2$ and odd
$M$ both Kashiwara--Miwa points $\nu=0,1$ are modular images of the broken $\ZZ_M$ model in the form \eqref{QFrec}, while for even $M$ only $\nu=1$
is. For $\gcd(n,M)\nmid\nu\pm1$, for example $\omega_0=\pi MT/n$ with $n\,|\,M$, the model is not a modular image of the Quano--Fujii model,
although it satisfies \eqref{KMnssr} as well.

\subsection{The points with odd \texorpdfstring{$\kappa$}{kappa} for even \texorpdfstring{$n$}{n}}\label{app:oddkappa}
For even $n$ the parity of $\kappa$ is an invariant of (a)--(c), and it has a simple meaning in terms of the weights: for $\omega_0=\pi(\kappa MT-\nu)/n$
the factors of \eqref{KMnW} and \eqref{KMnS} with $i-j\equiv n/2$ (mod $n$) contain the shift $\pm\kappa\pi T/2$ of the argument of
$\ovth_1(\cdot|T)$, which is a half period if and only if $\kappa$ is odd. For odd $\kappa$ these factors are therefore of $\ovth_4$-type, and for even
$\kappa$ of $\ovth_1$-type. For $n=2$ the points with odd $\kappa$ are the Kashiwara--Miwa model in its original form, with the factors $w_4$ (see
below \eqref{KMnssr}), whose star-triangle relation was proven in \cite{Kashiwara:1986tu,Hasegawa:1990du} through the cyclic representations of
the Sklyanin algebra, while the point $\kappa=0$ is the broken $\ZZ_M$ model in the form of \cite{Quano:1992wc}. The real points $(\kappa,\kappa r)$
for even $n\ge4$ may thus be regarded as the $\mathfrak{sl}(n)$ analogues of the original Kashiwara--Miwa model. When $\gcd(n,M)=1$ their weights are
modular images of the Quano--Fujii weights by the sufficiency of \eqref{QFcond}, but the corresponding transformations have odd $\gamma/M$, and
they do not preserve \eqref{KMnssr}: the crossing parameter $\pi/2$ is shifted by a half period of $\CC/(\pi M\ZZ+\pi MT'\ZZ)$, $T':=gT$, and the
four internal spectral parameters add up to $\pi/\mathcal J$ instead of $\pi$. The resulting relation is not a specialization of the Yang--Baxter
equation of \cite{Quano:1992wc} of the type used in Section~\ref{sec:QFproof}, and we have checked numerically that the extra factors do not
reduce to a character of the internal spin. A proof at these points presumably requires an analogue of the construction of
\cite{Kashiwara:1986tu,Hasegawa:1990du} for Belavin's $R$-matrix; the models of \cite{Sergeev:1992ap,Hasegawa:1993cs} are natural candidates,
which we have not compared with \eqref{KMnW}.

\section{\texorpdfstring{Modular covariance of the elliptic $\mathfrak{sl}(n)$ Kashiwara--Miwa type models}{Modular covariance of the elliptic sl(n) Kashiwara-Miwa type models}}\label{app:modular}

In this appendix we give the details of the computation for the transformations (c) of Appendix~\ref{app:torsion}, with $g$ and $\mathcal J=\gamma T+\delta$
as there. By \eqref{thetamod} and \eqref{thetaper}, each theta function contributes
\begin{equation}
  \ovth_1\Big(\frac{x}{M}\,\Big|\,T\Big)=c_\vartheta\,\EXP^{-\frac{\ii\gamma x^2}{\pi\mathcal JM^2}}\EXP^{2\ii k_2y-\pi\ii k_2^2gT}(-1)^{k_1+k_2}\ovth_1(y|gT)\;,
\end{equation}
where $y$ is the corresponding argument of the transformed model and $x/(M\mathcal J)=y+k_1\pi-k_2\pi gT$. Using $1/\mathcal J=\alpha-\gamma\,gT$,
the integers $k_1,k_2$ are found as follows. For the factor with index $k$ of a $W$-edge, $x=(i-j)\omega_0+\theta+(k+1/2)\pi$; for the factor with
index $k$ of a $\olW$-edge, $\theta$ is replaced by $\pi/2-\theta$; and for the factor $(i,j)$ of $\mathsf S$, $x=\pi(a_i-a_j)-(j-i)\omega_0$. Since
$\omega_0$ and $\theta$ are rescaled by $1/\mathcal J$, only the remaining multiples of $\pi$ contribute, and we find
\begin{equation}
  (k_1,k_2)=\begin{cases}
    \Big(\frac{(2k+1)(\alpha-1)}{2M},\,\frac{(2k+1)\gamma}{2M}\Big) & \text{for } W\;,         \\
    \Big(\frac{(k+1)(\alpha-1)}{M},\,\frac{(k+1)\gamma}{M}\Big)         & \text{for } \olW\;,      \\
    \Big(\frac{(a_i-a_j)(\alpha-1)}{M},\,\frac{(a_i-a_j)\gamma}{M}\Big) & \text{for } \mathsf S\;.
  \end{cases}
\end{equation}
These are integers, since $\alpha\equiv1$ and $\gamma\equiv0$ (mod $2M$). For $\olW$ and $\mathsf S$ they
are even, while for $W$ their parities are those of $(\alpha-1)/(2M)$ and $\gamma/(2M)$, independently of $k$. The signs and the constants $c_\vartheta$ therefore cancel, because $\sum_{i,j}(a_i-b_j)=0$. The remaining exponents are polynomials in $a_i-b_j$ and $i-j$, and in $a_i-a_j$, $j-i$ for $\mathsf S$. After the sums over the
components they therefore depend on the spins only through the moments $\sum_ia_i^s\,i^t$, while $\sum_{i=1}^ni^t$ is a polynomial in $n$.
We have evaluated the sum of the exponents over \eqref{KMnssr} with computer algebra in two ways: for $n=2,3,4$ with $M$, $\omega_0$, $g$ and the spectral
parameters as symbols, and for general $n$, treating the moments $\sum_ia_i^s\,i^t$ of the internal and external spins as independent symbols and
evaluating $\sum_{i=1}^n i^t$ by Faulhaber's formula. In both cases the derivatives of the total exponent of each side with respect to the
internal-spin moments, and those of the difference of the two sides with respect to the external-spin moments, vanish identically. In fact, the
difference $\mathcal E_{\rm L}-\mathcal E_{\rm R}$ of the total exponents of the two sides vanishes identically. Thus the total exponent of each side does not
depend on the internal spin, and the transformed relation holds again with $\mathcal R=1$.

\section{Numerical checks}\label{app:numerics}

The results of this paper were checked numerically with an independent implementation of $\Phi(z,m)$ and of the theta functions
(Python/NumPy, double precision). The scripts are available at \url{https://github.com/masahitoyamazaki/Lens_root-of-unity}. Unless stated otherwise we use $\tau=0.05+0.35\ii$, rapidities $(u,u',v,v')=(0.9,0.5,0.2,0.1)$ and random
external spins. For a relation $L=R$ we compute $L/R$ for several (three to ten) random choices of the external spins, and in some checks also of
the rapidities. Saddle points are found by Newton's method with many initial points in a fundamental domain; the saddle point of the second star is
obtained by Newton's method starting from that of the first one, and a non-corresponding saddle point is detected by \eqref{cssr}. Limits
$\hbar\to0$ are extracted by Richardson extrapolation. In Table~\ref{tab:numerics}, ``spin-independent'' means that $L/R$ is independent of the
discrete spins to the quoted accuracy, and ``$\mathcal R=1$'' or ``$\mathcal R$ of \eqref{Kdef}'' that $L/R$ agrees with the stated normalization.

\begin{table}[p]\centering\footnotesize
  \renewcommand{\arraystretch}{1.0}
  \setlength{\tabcolsep}{4pt}
  \begin{tabular}{>{\raggedright\arraybackslash}p{5.4cm}>{\raggedright\arraybackslash}p{7.0cm}>{\raggedright\arraybackslash}p{2.3cm}}
    \hline
    relation                                                                                                                               & cases                                                                                                      & accuracy                            \\
    \hline
    \eqref{ssr}, numerical integration over $\sigma_H$                                                                                     & $n=2$ ($r=1,3$), $n=3$ ($r=1,2$)                                                                           & machine precision                   \\
    \eqref{Phiprop}, \eqref{shift}, \eqref{shift1}, \eqref{pairexact}; dictionary \eqref{PhiGamma}                                         & $r=1,2,3$                                                                                                  & $10^{-16}$                          \\
    \eqref{omegalim}--\eqref{Pdef}                                                                                                         & $(r,N)=(2,3)$, $(3,2)$, $(1,3)$, $(3,1)$, $(1,5)$, $|a|\le7$                                               & $10^{-6}$--$10^{-7}$ ($O(\hbar^2)$) \\
    \eqref{psidef}                                                                                                                         & $(r,N)=(1,1)$, $(1,3)$, $(2,2)$, $(3,2)$                                                                   & $10^{-6}$                           \\
    \eqref{KYasym}--\eqref{KYsub} for all $m$                                                                                              & $(r,N)=(2,3)$, $(3,2)$, $(1,3)$, $-3\le m\le r+2$                                                          & $10^{-6}$                           \\
    $\xi$-dependence of \eqref{pairasym}, \eqref{pairlim}; asymptotics of $\log\ovth_1(x|1/N+\ii\hbar/\pi)$ including $\pi^2/(12N^2\hbar)$ & $(r,N)=(2,3)$, $(3,2)$, $(1,4)$; $N=1,2,3$                                                                 & $10^{-13}$                          \\
    modular form \eqref{omegamod}                                                                                                          &                                                                                                            & $10^{-14}$                          \\
    \hline
    classical star-star relation \eqref{cssr}; \eqref{q4leg}                                                                               & $(n,r,N)=(2,1,3)$, $(2,2,2)$, $(3,1,2)$, $(3,2,2)$, $(3,1,3)$                                              & $10^{-16}$; $10^{-15}$              \\
    \eqref{dssr}, spin-independent                                                                                                         & $(n,r,N)=(2,1,3)$, $(2,2,3)$, $(2,3,2)$, $(2,2,4)$, $(3,1,2)$, $(3,2,2)$, $(3,1,3)$, $(3,3,1)$, $(4,1,2)$  & $10^{-11}$--$10^{-14}$              \\
    \eqref{dssr}, $\mathcal R$ of \eqref{Kdef}                                                                                             & $(n,r,N)=(2,1,3)$, $(2,2,3)$, $(2,3,2)$, $(3,1,2)$, $(3,2,2)$, $(3,1,3)$, $(4,1,2)$                        & $10^{-14}$                          \\
    \eqref{dssr} at the real configurations \eqref{xi0}, $\gcd(\kappa,n)=1$                                                                & $(n,r,N)=(2,1,3)$, $(3,1,2)$, $(3,2,1)$, $(3,1,3)$, $(3,2,2)$, $(4,1,2)$                                   & $10^{-12}$--$10^{-14}$              \\
    \hline
    $n=2$ weights vs.\ \cite{Kels:2017vbc}                                                                                                 & $(r,N)=(2,3)$, $(3,2)$, $(1,5)$, $(3,1)$                                                                   & $10^{-13}$                          \\
    equivalence \eqref{tau1def}                                                                                                            & $(n,r,N)=(2,2,3)$, $(2,3,2)$, $(3,2,2)$, $(3,3,1)$                                                         & $10^{-13}$                          \\
    \eqref{KMnssr}, $\mathcal R=1$, all primitive $\omega_0=\pi(\kappa MT-\nu)/n$ (fails at non-degenerate non-primitive points)           & $(n,M)=(2,3)$, $(3,2)$, $(3,3)$, $(3,4)$, $(4,2)$, $(4,3)$; $\kappa\in\{-1,\ldots,3\}$, all $\nu$          & $10^{-13}$                          \\
    \eqref{KMnssr}, $\mathcal R=1$, points in the set (3) of Appendix~\ref{app:torsion}                                                    & $(n,M,\kappa,\nu)=(4,4,0,5)$, $(4,4,2,5)$, $(4,8,0,5)$, $(4,8,1,6)$, $(6,3,2,3)$, $(6,3,4,3)$, $(6,6,2,3)$ & $10^{-12}$                          \\
    \eqref{QFid} with the signs $\varsigma$, random lifts                                                                                  & $(n,M)=(2,3)$, $(3,2)$, $(3,3)$, $(3,4)$, $(4,2)$, $(4,3)$, $(5,2)$                                        & $10^{-13}$                          \\
    \eqref{QFYBE}, and its reduction to \eqref{SSRQF}                                                                                      & $(n,M)=(2,3)$, $(3,2)$, $(2,5)$, $(3,4)$, $(4,3)$                                                          & $10^{-12}$--$10^{-15}$              \\
    modular images \eqref{QFcond}                                                                                                          & $(n,M)=(2,3)$, $(3,2)$, $(3,4)$, $(4,3)$, all $\nu$                                                        & $10^{-12}$                          \\
    reality \eqref{KMconj}                                                                                                                 & $(n,M)=(2,3)$, $(3,2)$, $(3,3)$, $(4,2)$                                                                   & $10^{-14}$                          \\
    Proposition~\ref{prop:orbits} by enumeration                                                                                           & $n,M\le10$                                                                                                 & exact                               \\
    ratio $\mathsf W_\theta(\va+\va_{\rm ker},\boldsymbol b)/\mathsf W_\theta(\va,\boldsymbol b)$ not a gauge factor; formula for $(2,2)$  & $(n,M)=(2,2)$, $(2,4)$, $(3,3)$, $(3,6)$, $(4,2)$, $(4,4)$                                                 & $10^{-15}$                          \\
    symmetry \eqref{Cshift}; lift-invariance of \eqref{QFS}                                                                                & $(n,M)=(2,2)$, $(3,2)$, $(3,3)$, $(4,2)$, $(4,3)$                                                          & $10^{-14}$                          \\
    \eqref{Qcompare}--\eqref{Qrel}                                                                                                         & $n\le6$, random spins                                                                                      & exact                               \\
    \hline
    \eqref{trigW} in \eqref{KMnssr} (fails without the chiral phase for $n\ge3$)                                                           & $(n,M)=(2,3)$, $(2,5)$, $(3,2)$, $(3,3)$, $(3,4)$, $(4,2)$                                                 & $10^{-13}$                          \\
    \eqref{triggenW} with saddle points from \eqref{triggensaddle}, random external $\psi$                                                 & $(n,M)=(2,3)$, $(2,4)$, $(3,2)$, $(3,3)$, $(4,2)$                                                          & $10^{-14}$                          \\
    \eqref{BKMSform} and the weights \cite[(3.19)--(3.22)]{Bazhanov:1990qk} on the curve \eqref{bkmscurve}, random M\"obius maps           & $(n,M)=(2,3)$, $(2,5)$, $(3,2)$, $(3,3)$, $(3,4)$, $(4,2)$, $(4,3)$, $(5,2)$                               & $10^{-13}$                          \\
    compressed BKMS form of \cite{Moosavian:2025nyv}                                                                                       & holds only with $\varpi^{-k}$ in the numerator                                                             &                                     \\
    \hline
  \end{tabular}
  \caption{Numerical checks. Accuracies are maximal deviations of $L/R$ from a constant (spin-independence) or from the stated normalization.}\label{tab:numerics}
\end{table}

\clearpage
\bibliographystyle{nb}
\bibliography{refs}

\end{document}